\documentclass{article}
\usepackage{booktabs} 
\usepackage{graphicx}
\usepackage{caption} 
\usepackage{subcaption} 
\usepackage{amsmath, amsthm, amssymb, amsfonts}
\usepackage{commath, mathtools, enumerate, dsfont,tikz-cd, hyperref, tikz, pgf, ytableau}
\usepackage{mathrsfs}
\usepackage{mathtools}
\usepackage{float}
\usepackage{makecell}
\usepackage{indentfirst}
\usepackage{setspace}
\usepackage{calc}
\usepackage{aliascnt}
\usepackage[capitalize]{cleveref}
\usepackage{comment}
\usepackage{multicol}
\usepackage[edges]{forest}
\usepackage{tikz}
\usepackage[normalem]{ulem}
\usetikzlibrary{arrows.meta, positioning, calc, patterns}

\tikzset{
  box/.style={
    rectangle, draw, rounded corners,
    align=center, minimum width=2.8cm, minimum height=1cm
  },
  arrow/.style={
    -{Stealth[length=3mm]}, thick
  }
}

\usepackage[
   includehead,
   includefoot,
     left = 1in, 
      top = 1in, 
    right = 1in,
   bottom = 1in
]{geometry}

\makeatletter
\newcommand*{\toccontents}{\@starttoc{toc}}
\makeatother

\usepackage{bm}

\usepackage{color}

\usepackage{enumitem}
\usepackage{mathtools}

\usepackage{hyperref}
\hypersetup{
    colorlinks=true,
    linkcolor=blue,
    filecolor=magenta,      
    urlcolor=teal,
    citecolor=cyan,
}

\newtheorem{theorem}{Theorem}[section]

\newaliascnt{proposition}{theorem}
\newtheorem{proposition}[proposition]{Proposition}
\aliascntresetthe{proposition}

\newaliascnt{corollary}{theorem}
\newtheorem{corollary}[corollary]{Corollary}
\aliascntresetthe{corollary}

\newaliascnt{lemma}{theorem}
\newtheorem{lemma}[lemma]{Lemma}
\aliascntresetthe{lemma}

\newaliascnt{definition}{theorem}
\newtheorem{definition}[definition]{Definition}
\aliascntresetthe{definition}

\newaliascnt{observation}{theorem}
\newtheorem{observation}[observation]{Observation}
\aliascntresetthe{observation}

\newaliascnt{exercise}{theorem}

\aliascntresetthe{exercise}

\newaliascnt{problem}{theorem}

\aliascntresetthe{problem}

\newaliascnt{note}{theorem}

\aliascntresetthe{note}

\newaliascnt{question}{theorem}

\aliascntresetthe{question}

\newaliascnt{conjecture}{theorem}

\aliascntresetthe{conjecture}

\newaliascnt{fact}{theorem}
\newtheorem{fact}[fact]{Fact}
\aliascntresetthe{fact}

\theoremstyle{definition}
\newaliascnt{example}{theorem}
\newtheorem{example}[example]{Example}
\aliascntresetthe{example}

\newaliascnt{remark}{theorem}
\newtheorem{remark}[remark]{Remark}
\aliascntresetthe{remark}

\numberwithin{equation}{section}

\crefname{theorem}{Theorem}{Theorems}
\Crefname{theorem}{Theorem}{Theorems}
\crefname{proposition}{Proposition}{Propositions}
\Crefname{proposition}{Proposition}{Propositions}
\crefname{corollary}{Corollary}{Corollaries}
\Crefname{corollary}{Corollary}{Corollaries}
\crefname{lemma}{Lemma}{Lemmas}
\Crefname{lemma}{Lemma}{Lemmas}
\crefname{definition}{Definition}{Definitions}
\Crefname{definition}{Definition}{Definitions}
\crefname{observation}{Observation}{Observations}
\Crefname{observation}{Observation}{Observations}
\crefname{exercise}{Exercise}{Exercises}
\Crefname{exercise}{Exercise}{Exercises}
\crefname{problem}{Problem}{Problems}
\Crefname{problem}{Problem}{Problems}
\crefname{note}{Note}{Notes}
\Crefname{note}{Note}{Notes}
\crefname{question}{Question}{Questions}
\Crefname{question}{Question}{Questions}
\crefname{conjecture}{Conjecture}{Conjectures}
\Crefname{conjecture}{Conjecture}{Conjectures}
\crefname{fact}{Fact}{Facts}
\Crefname{fact}{Fact}{Facts}
\crefname{example}{Example}{Examples}
\Crefname{example}{Example}{Examples}
\crefname{remark}{Remark}{Remarks}
\Crefname{remark}{Remark}{Remarks}

\usepackage{algorithm}
\usepackage[noend]{algpseudocode}
\algnewcommand\algorithmicinput{\textbf{Input:}}
\algnewcommand\algorithmicoutput{\textbf{Output:}}
\algnewcommand\Input{\item[\algorithmicinput]}%
\algnewcommand\Output{\item[\algorithmicoutput]}

\tikzstyle{tensor}=[rectangle,draw=blue!70,fill=blue!20,thick]
\tikzstyle{phys_sym}=[circle,draw=green!30,fill=green!20,thick]
\tikzstyle{bond_sym}=[circle,draw=purple!50,fill=purple!20,thick]

\DeclareMathOperator*{\argmin}{arg\,min}

\def\ad{\mathrm{ad}}

\def\C{\mathbb{C}}

\def\R{\mathbb{R}}
\def\S{\mathbf{S}}

\def\su{\mathfrak{su}}

\def\Z{\mathbb{Z}}

\def\idty{\mathds{1}}

\def\comm{\mathrm{comm}}

\def\spec{\mathrm{spec} }

\def\eps{\varepsilon}
\def\gap{\mathrm{gap}}

\def\Sn{\mathsf{S_n}}
\def\whatSn{\widehat{\mathsf{S}}_n}
\def\SU{\mathsf{SU}}

\def\sfG{\mathsf{G}}

\def\calA{\mathcal{A}}
\def\calB{\mathcal{B}}

\def\calE{\mathcal{E}}

\def\calH{\mathcal{H}}
\def\calI{\mathcal{I}}
\def\calL{\mathcal{L}}
\def\calJ{\mathcal{J}}

\def\calP{\mathcal{P}}
\def\calS{\mathcal{S}}
\def\calT{\mathcal{T}}

\def\up{\uparrow}
\def\down{\downarrow}

\newcommand{\wt}[1]{\widetilde{#1}}
\newcommand{\inprod}[1]{\left\langle #1 \right\rangle}
\newcommand{\bra}[1]{\left\langle #1 \right\vert}
\newcommand{\ket}[1]{\left\vert #1 \right\rangle}

\newcommand{\paran}[1]{\left( #1 \right)}

\newcommand{\Tr}{\mathrm{Tr}}

\newcommand{\loc}{{\mathrm{loc}}}
\newcommand{\tot}{{\mathrm{tot}}}

\newcommand{\indicator}[1]{{\boldsymbol{1}_{#1}}}
\newcommand{\ketbra}[2]{{\ket{#1}\!\!\bra{#2}}}
\newcommand{\EV}{{\mathbb{E}}}

\DeclareMathOperator\atanh{atanh}

\let\emph\relax 
\DeclareTextFontCommand{\emph}{\bfseries} 

\definecolor{gray}{rgb}{0.93,0.93,0.93}
\definecolor{light-gold}{rgb}{0.96,0.8,0}
\definecolor{light-red}{rgb}{1,0.4,0.4}
\definecolor{light-green}{rgb}{0.5,1,0.5}
\definecolor{light-blue}{rgb}{0.4,0.4,1}

\newcommand{\commt}[2]{{\left[#1,#2\right]}}
\newcommand{\acommt}[2]{{\left\{#1,#2\right\}}}
\newcommand{\lind}{{\mathcal{L}}}

\newcommand{\sS}{\mathsf{S}}

\newcommand{\CF}{\mathcal{F}}
\newcommand{\CL}{\mathcal{L}}

\newcommand{\BR}{\mathbb{R}}

 \newcommand*{\tr}{\mathrm{Tr}}
\newcommand{\rd}{\mathrm{d}}

\usepackage[
backend=biber,
style=alphabetic,
sorting=nyt
]{biblatex}
\definecolor{kspurple}{RGB}{160, 32, 240}

\usetikzlibrary{positioning,calc}

\title{Spectral Gap of the Davies Generator \\for the Mean-Field Heisenberg Model}
\author{%
Joao Basso\textsuperscript{1},
Thiago Bergamaschi\textsuperscript{2},
Lin Lin\textsuperscript{1,3,4},
Michael Ragone\textsuperscript{1}\thanks{Corresponding author. Email addresses: \texttt{\{joao.basso,thiagob,linlin,micragone,kstubbs\}@berkeley.edu.}},
and Kevin~D. Stubbs\textsuperscript{1}\\[0.6em]
\parbox{0.95\textwidth}{\centering\small
\textsuperscript{1}Department of Mathematics, University of California, Berkeley\\
\textsuperscript{2}Department of Electrical Engineering and Computer Science, University of California, Berkeley\\
\textsuperscript{3}Applied Mathematics and Computational Research Division, Lawrence Berkeley National Laboratory\\
\textsuperscript{4}Department of Computing and Mathematical Sciences, California Institute of Technology}}

\date{\today}

\begin{document}

\maketitle

\begin{abstract}
The mean-field Heisenberg ferromagnet is a quantum spin model on the complete graph with isotropic spin-1/2 interactions.
This non-commuting Hamiltonian is permutation and $\mathsf{SU}(2)$ invariant, and its Gibbs states undergo an $\mathsf{SU}(2)$ symmetry breaking phase transition at inverse temperature $\beta=2$.
We consider the associated Davies generator, a canonical model of open-system thermalization, and prove tight asymptotic estimates for its spectral gap at all noncritical temperatures.
For fixed $\beta<2$, the gap as a function of number of qubits $n$ is $\Theta(1)$, while for fixed $\beta>2$ the gap is $\Theta(n^{-1})$. 
The matching upper bound of the spectral gap is witnessed by the total magnetization order parameter, suggesting that the low-temperature ($\beta>2$) slowdown is associated with broken continuous symmetry.
Two key ingredients in our approach are a comparison argument, which introduces auxiliary generators to bound dissipation on nontrivial representations of the symmetry groups $\SU(2)$ and $\sS_n$, and a decomposition of the space of observables into spherical tensor operators to reveal a form of monotonicity.

\end{abstract}

\newpage
\setcounter{tocdepth}{2} 
\setlength{\columnsep}{25pt}
\begin{multicols}{2}
{\small \tableofcontents}
\end{multicols}

\section{Introduction}

One of the most compelling future applications of quantum computers is the simulation of quantum matter at finite temperature, 
as probing thermal phases using classical methods generically requires computations over an exponentially large Hilbert space. 
A quantum algorithm capable of preparing thermal states across different temperatures and interaction parameters
could serve as a key subroutine for such simulations.
In part, this motivation has led to a modern line of works on dissipative quantum algorithms for thermal state preparation \cite{TemmeOsborneVollbrechtEtAl2011,Rall2023thermalstate, chen2023efficient, chen2025efficient, ding2024efficient, gilyen2024quantum, jiang2024quantum}, akin to the design of Markov chain Monte Carlo methods. 

Recently there has been significant interest in proofs of efficient convergence for these new families of quantum algorithms. This often amounts to bounds on the \textit{mixing time} of a simulated dissipative process, and such bounds are now known under several structural assumptions. Notably, for quantum spin systems at high-temperatures \cite{kastoryano2016quantum, rouze_efficient_2025, rouze26, bakshi2025dobrushinconditionquantummarkov}, and one-dimensional systems \cite{Bardet2023Rapid, Bardet2024, bergamaschi2025quantumspinchainsthermalize}. By and large, however, existing techniques are perturbative and only apply in regimes where thermal phase transitions are absent, or deliberately avoided. Consequently, we lack a precise understanding of the performance of these algorithms exactly where they could be most applicable:
as they approach criticality.

Here we confront this challenge in a paradigmatic model of quantum magnetism, the mean-field Heisenberg model. In this model, each pair of qubits interacts through a ferromagnetic, isotropic (or $\SU(2)$ symmetric) coupling. Its Hamiltonian, which may be viewed as an interaction on the complete graph on $n$ qubits, is written as:
\begin{equation} \label{def:heisenberg model}
    H = -\frac{1}{n} \sum_{1\leq i<j \leq n} \mathbf{S}_i\cdot \mathbf{S}_j = - \frac{1}{n} \sum_{1\leq i < j \leq n} \big(S_i^X S_j^X + S_i^Y S_j^Y + S_i^Z S_j^Z\big),
\end{equation}
where
$\mathbf{S}_i=(S_i^X,S_i^Y,S_i^Z)$ are the spin-$\frac{1}{2}$ Pauli matrices on site $i$. Its Gibbs states undergo an $\SU(2)$ symmetry breaking phase transition at the inverse temperature $\beta=2$ which can be detected by the variance of the total magnetization order parameter $\mathbf{S}_{\tot}:=\sum_i \mathbf{S}_i$~\cite{fannes_equilibrium_1980,toth1990phase,bjornberg_quantum_2020}.

Mean-field models have long served as testbeds for thermalization: by replacing spin-spin interactions with interactions between a spin and the ``average'' of all other spins, a mean-field model retains a remarkable amount of structure of the original model while remaining analytically tractable. A prime example is provided by the Curie-Weiss (or mean-field Ising) model, which provides useful information not only on the static thermal phases, including bounds on the pressure and critical temperature~\cite{simon_statistical_1993,friedli_statistical_2017}, but also on the phases of the corresponding Glauber dynamics \cite{ding2009mixing, levin2010glauber, levin_markov_2017}. Meanwhile, these questions can be significantly more challenging to address on lattices \cite{GS23Ising, galanis2024plantingmcmcsamplingpotts}.
In a similar fashion, the mean-field ferromagnetic Heisenberg model is thought to serve as a prototype
for the local Heisenberg model on sufficiently high-dimensional lattices $\Z^d$, which is believed to likewise exhibit a phase transition for $d\geq 3$.\footnote{Despite a wealth of physical evidence, this problem remains open after half a century. See~\cite{simon_phase_2025} for a discussion.}

In this spirit, we study the open-system thermalization of the mean-field Heisenberg ferromagnet. Inspired by Glauber dynamics, we model the thermalization process via the Davies generator with single-site Pauli couplings, a canonical model of quantum heat-bath dynamics \cite{davies74, davies76}. Our main result is a sharp characterization of the spectral gap of the Davies generator at every fixed noncritical inverse temperature $\beta\neq 2$. Furthermore, efficient simulation of the resulting dynamics gives a polynomial-time quantum algorithm for preparing the Gibbs state at every constant temperature. Although the model itself is classically simulable, it provides a useful while analytically tractable example of a non-commuting system with a phase transition, and we believe the tools and insights gained from this analysis will serve as a stepping stone towards the simulation of more complex quantum models exhibiting spontaneous symmetry breaking.

\subsection{Main Result}

The Davies generator \cite{davies74, davies76} is a Lindbladian specified by a Hamiltonian $H$, an inverse temperature $\beta\geq 0$, and a collection of jump operators. We take the collection of single-site Pauli jump operators, and thus the Schr\"odinger-picture time-evolution yields a quantum Markov semigroup whose unique fixed point is the Gibbs state $\rho\propto \exp(-\beta H)$. The evolution of an arbitrary state $\sigma$ is given by\footnote{with $\CL_{S_i^\alpha}^\dagger$ the Davies generator associated to a single Pauli operator $S_i^\alpha$. For precise definitions and setup, see~\cref{section:background}.} 
\begin{equation}
    \frac{\rd}{\rd t}\sigma := \CL_\loc^\dagger(\sigma), \quad \text{where}\quad \CL_\loc^\dagger = \sum_{i\in [n]}\sum_{\alpha\in \{X, Y, Z\}} \CL_{S_i^\alpha}^\dagger.
\end{equation}
From this point forward, we will work in the Heisenberg picture and focus on $\lind_\loc$. 
The Davies generator is Kubo-Martin-Schwinger detailed-balanced and negative semi-definite, and thus the rate of convergence to $\rho$ can be captured by the spectral gap of $\lind_\loc$.
Our main theorem gives the asymptotically-sharp system-size dependence of this gap at every fixed noncritical temperature, as well as bounds at criticality.

\begin{theorem}[The spectral gap of $\CL_\loc$]\label{theorem:main_davies_gap}
Fix an inverse temperature $\beta\geq 0$. Let $\lind_\loc$ be the Davies generator with single-site Pauli jumps associated to the $n$-qubit mean-field quantum Heisenberg model.
Then, the spectral gap of $\CL_\loc$ as a function of $n$ satisfies
\begin{equation}
    \mathrm{gap}(\calL_{\loc}) = \begin{cases}
                \Theta(1) & \text{if }\beta <  2  \qquad \textsf{(high temperature)}, \\
                \Theta(n^{-1}) & \text{if }\beta > 2  \qquad  \textsf{(low temperature)}.
                \end{cases}
\end{equation} When $\beta=2$, the spectral gap of $\calL_{\loc}$ is bounded by $\Omega(n^{-1})\leq \gap(\calL_{\loc}) \leq O(n^{-1/2})$.
\end{theorem}

 \cref{theorem:main_davies_gap} is optimal in its dependence on the system size $n$ for fixed $\beta$ everywhere away from the critical point of $\beta=2$, and  we leave open a complete characterization of when $\beta=2\pm o(1)$ asymptotically approaches criticality. The observable which witnesses the tightness of these bounds is the total magnetization $\mathbf{S}_{\tot}$, the same order parameter that detects the static symmetry breaking transition. 
 To the best of our knowledge, \cref{theorem:main_davies_gap} is the first systematic, non-perturbative characterization of the spectral gap of the dynamics of a non-commuting system in the presence of a thermodynamic phase transition.

\subsection{Related Work}
\label{section:related-work}

\vspace{1em}
\noindent \textbf{Quantum heat-bath mixing.}
A number of techniques now exist to bound spectral gaps and log-Sobolev constants for detailed-balanced quantum Markov semigroups. For commuting Hamiltonians, decay-of-correlations assumptions are known to imply fast or rapid mixing in several settings \cite{kastoryano2016quantum,Bardet_2021,Bardet2024,capel2021modified,kochanowski2024rapid,capel2024quasi}. Existing proofs for non-commuting systems are largely based on perturbative arguments in space and in temperature \cite{rouze_efficient_2025, bergamaschi2025quantumspinchainsthermalize, rouze26,bakshi2025dobrushinconditionquantummarkov,tong2025fast,vsmid2025polynomial,vsmid2025rapid,bakshi2026rapid, B26all}, including, notably, those based on quantum analogues of the Dobrushin conditions \cite{dobrushin1985completely, dobrushin1985constructive, 1968_Dobrushin, MajewskiZegarlinski1995QSDI, MajewskiZegarlinski1996QSDII, MajewskiOlkiewiczZegarlinski1998}. We remark that with the exception of the recent work of \cite{B26all} (at high-temperatures), proofs of mixing times in quantum systems are largely limited to lattice Hamiltonians. Finally, code Hamiltonians and special quantum-to-classical reductions provide additional tractable examples \cite{alicki_thermalization_2009,alicki2010thermal,ding2024polynomial, BCL2024, bergamaschi2026rapid,paezvelasco2025efficientsimplegibbsstate,basso_quantum_2025}.

\vspace{1em}
\noindent \textbf{Classical mean-field analogs.}
Glauber dynamics on the Curie-Weiss model (or, mean-field Ising model) is exponentially slow to mix at low-temperatures, since local updates must cross a macroscopic ``energy'' barrier between two discrete symmetry-broken phases \cite{griffiths1966relaxation,ding2009mixing,levin2010glauber}. The closer classical analog of the Heisenberg model is the classical $\mathsf{O}(d)$ mean-field model, whose spins have continuous symmetry; Becker and Menegaki proved an $n^{-1}$ low-temperature spectral gap scaling for the associated Ginzburg--Landau dynamics \cite{becker_spectral_2020}. In both the classical $\mathsf{O}(d)$ mean-field model and the mean-field quantum Heisenberg model, the low-temperature slowdown is polynomial rather than exponential, because the slow relaxation occurs along a continuously degenerate symmetry-broken manifold rather than across a discrete free-energy barrier. This seems to indicate the primary bottleneck for mixing stems from spin-wave (Goldstone) excitations associated with this broken continuous symmetry.

\vspace{1em}
\noindent \textbf{Pauli master equations and invariant reductions.}
Several analyses of quantum Markov semigroups exploit situations where the diagonal algebra in an energy eigenbasis is invariant under the Lindbladian \cite{temme_lower_2013,chen2021fast,ramkumar2025mixing,chen_randomized_2024,basso_quantum_2025}. Here, leveraging the $\SU(2)\times \Sn$ symmetry, we choose to restrict the dynamics to a classical Markov chain whose state space consists of the \textit{energies} of $H$, rather than an eigenbasis. However, because of this choice, this ``coarse-grained'' Markov chain alone does not suffice to control the Lindbladian's spectral gap.

\subsection{Implication and Outlook}
\label{section:Implication and Outlook}

\vspace{1em}
\noindent \textbf{Efficient Lindblad dynamics and Gibbs sampling.} The spectral gap estimate implies that the Davies dynamics thermalizes the mean-field quantum Heisenberg ferromagnet in polynomial time \cite{kastoryano_quantum_2013}. We remark that since the spectra of the Heisenberg model are integers (after rescaling), a simulation of a time-step of the dynamics can be performed using a polynomial-sized quantum circuit, see e.g. \cite[Lemma 4.3]{BCL2024}, \cite[Lemma I.1]{chen2023quantumthermalstatepreparation}, thus implying efficient Gibbs state preparation on a quantum computer. For conciseness, we record only the implementability mechanism and do not spell out the circuit-level implementation of the Davies dynamics. We obtain the following corollary.
\begin{corollary}
    [Efficient Gibbs sampling at every fixed temperature] \label{cor:gibbs-prep}Fix an inverse temperature $\beta\geq 0$. There exists a quantum circuit composed of $\mathsf{poly}(n, \log \frac{1}{\epsilon})$ elementary gates which prepares the Gibbs state of the mean-field quantum Heisenberg ferromagnet at inverse temperature $\beta$ to error $\epsilon$ in trace distance, with constants allowed to depend on $\beta$.
\end{corollary}

The algorithmic corollary follows by simulating the Davies evolution for the mixing time implied by \cref{theorem:main_davies_gap}. For every fixed $\beta$, the gap lower bound and the $O(n)$ energy range give polynomial mixing time in trace distance. We do not claim that this is the most efficient way to prepare the Gibbs state for this highly symmetric model. Methods that use the $\SU(2)\times \sS_n$ decomposition directly may give simpler or more efficient preparation procedures, including, e.g. algorithms based on the Schur transform. 

\vspace{1em}
\noindent \textbf{Symmetric models and dynamical phase diagrams.} We believe the methods developed here may be useful for studying thermalization in other symmetry-rich quantum spin systems. Natural candidates include complete-graph XY and XXZ models, higher-spin or higher-local-dimension interchange models, and random-loop spin models \cite{ueltschi_random_2013,bjornberg_quantum_2020,bjornberg_dimerization_2021, tasaki_physics_2020,hohenberg_theory_1977}. Our analysis begins with a general observation: when a Hamiltonian and its coupling operators respect a symmetry group $\mathsf{G}$, the corresponding heat-bath Lindbladian is a $\mathsf{G}$-intertwiner. Schur's lemma then decomposes the observable algebra into invariant subspaces, allowing the dynamics on each sector to be analyzed separately. Particularly important are the decompositions into $\mathsf{G}$-invariant and non-invariant sectors; intuitively, one might expect dissipative dynamics of a symmetric model to drive observables towards their $\mathsf{G}$-invariant projections. The \textit{group mixer comparison argument} we develop here provides a path to prove this type of intuition. Complementarily, the spherical tensor basis we employ here reflects this representation-theoretic decomposition, and it seems likely that similar choices could make other symmetric dynamics more transparent.

When the Heisenberg model, or any model with an on-site symmetry $\mathsf{G}$, is placed on the lattice $\Z^d$, the Davies generator still inherits the corresponding on-site symmetry. The absence of global permutation symmetry, however, prevents the reduction to $\sS_n$-invariant observables that drives the present proof. Nevertheless, translation invariance supplies a decomposition into momentum sectors and may provide a starting point towards these geometrically-local models.

We point out that it is rather remarkable that the total magnetization order parameter witnesses the correct system-size scaling of the spectral gap. Many models' static phase diagrams are described by order parameters, including ferromagnetic XY or XXZ models and antiferromagnetic Heisenberg models \cite{tasaki_physics_2020,hohenberg_theory_1977}. This observation suggests a simple diagnostic for probing these models' dynamical phase diagrams.

\subsection{Organization}

In~\cref{section:background}, we recall the Davies generator and the Schur-Weyl representation theory used for the mean-field Heisenberg Hamiltonian. In~\cref{section:main-results}, we prove the main lower and upper bounds, with technical estimates deferred to the appendices. In~\cref{section:proof outline} we provide an outline of this proof. The appendices contain the representation-theoretic background (\cref{secapp:unitary reps}), symmetry inheritance of the Davies generator (\cref{secapp:proving intertwiners}), Wigner-Eckart computations (\cref{sec:wigner_eckart}), the coarse-grained Pauli master equation analysis (\cref{sec:proof_gap_of_L}), the group-mixer comparison estimates (\cref{secapp:the simulation argument}), and the high-temperature monotonicity argument (\cref{section:monotonicity-and-al-gap}.)

\section{Background}
\label{section:background}

\noindent \textbf{Notation.} We consider quantum spin systems on $n$ qubits. The associated Hilbert space is denoted by $\calH = (\C^2)^{\otimes n}$, and the space of bounded linear operators on this Hilbert space is denoted by $\calB(\calH)$. We call a linear map between the spaces of observables $\calL:\calB(\calH)\rightarrow \calB(\calH)$ a superoperator. We denote by $\calL^\dagger$ its adjoint with respect to the Hilbert-Schmidt inner product on observables, which is defined by
\begin{equation}
    \inprod{X,Y}:= \Tr\left[ X^\dagger Y \right] \qquad \text{where } X, Y\in \calB(\calH).
\end{equation}

 Given a full rank density matrix $\rho$, the \emph{Kubo-Martin-Schwinger} (KMS) inner product of two operators $X, Y$ with respect to $\rho$ is the scalar
\begin{equation}
    \langle X, Y \rangle_{\rho} := \Tr\left[ \rho^{\frac{1}{2}} X^\dagger \rho^\frac{1}{2} Y \right], \quad \|X\|_{\rho}:=  \langle X, X \rangle_{\rho}^{1/2},
\end{equation}
with $\|X\|_{\rho}$ referred to as the KMS or $\rho$ norm. Given a subspace $\calA$ of observables $\calB(\calH)$, we denote by $\calA^{\perp_\rho}:=\mathrm{span}\{X\in \calB(\calH):\langle X, Y\rangle_\rho=0\; \;\forall \,Y\in \calA\}$ the KMS orthogonal complement of $\calA$. If a superoperator $\calL$ is self-adjoint with respect to this inner product, i.e.
\begin{equation}\label{def:KMS DBC}
    \langle X, \calL (Y) \rangle_\rho = \langle \calL (X), Y \rangle_\rho,
\end{equation} then we call $\calL$ \emph{KMS detailed-balanced}. As usual, if a KMS detailed-balanced superoperator $-\calL$ has only nonnegative eigenvalues $\lambda \geq 0$, we write $-\calL\geq 0$ and refer to it as positive semi-definite.

The \emph{spectral gap} of a positive semi-definite operator $-\calL\geq 0$ is defined to be its smallest nonzero eigenvalue. We remind the reader that this quantity may be computed variationally (see, e.g. \cite[Lemma 3]{temme_lower_2013}). 

\begin{fact}
    [The Spectral Gap]\label{definition:gap} Given a negative semi-definite superoperator $\calL\leq 0$ which is KMS detailed-balanced with respect to $\rho$, the spectral gap of $\calL$ can be computed as:
    \begin{align}\label{eq:gap_def}
    {\gap}(\mathcal{L})  = \min \left\{\frac{\langle X, -\mathcal{L} (X) \rangle_{\rho}}{\|X\|_{\rho}^2} \; : \quad X\neq 0 \text{ and } X\in (\ker{\calL})^{\perp_\rho} \right\}.
\end{align}
\end{fact}
When the kernel of $\calL$ is spanned by $\idty$, the variational expression in~\cref{eq:gap_def} is often recast in terms of the \emph{Dirichlet form} and the \emph{variance}:
\begin{equation} \label{def:dirichlet form and variance}
    \calE[X] := \inprod{X,-\calL(X)}_\rho, \qquad \mathsf{Var}_\rho[X] := \norm{X - \Tr[\rho X]}_\rho^2 . 
\end{equation}

The \emph{commutant} of a collection of observables $\calJ\subseteq \calB(\calH)$ is the unital $C^*$ algebra 
\begin{equation}
    \comm(\calJ) := \{ X\in \calB(\calH): XJ=JX \quad \text{ for all } J\in\calJ\}.
\end{equation} Note that when $\calH$ is a representation of a group $\mathsf{G}$, we will frequently abuse notation and write $\comm(\mathsf{G})$ to refer to the commutant of the representatives of this group.

\subsection{Schur-Weyl Duality}\label{sec:Heisenberg_spectrum}
The Heisenberg model $H$ on the complete graph \cref{def:heisenberg model} possesses a high degree of symmetry, commuting with both the actions of $\SU(2)$ and $\Sn$ on $\calH$:
\begin{equation}\begin{split}
    U^{\otimes n} H (U^{\dagger})^{\otimes n} &= H \qquad \text{for all }U\in \SU(2), \\
    \sigma H \sigma^{-1} &= H \qquad \text{for all }\sigma\in \Sn . 
\end{split}\end{equation} We dedicate~\cref{secapp:unitary reps} to a primer on representation theory. Now, Schur-Weyl duality describes the decomposition of $\calH$ into irreps of the product group $\SU(2)\times \Sn$, where the action is given by $(U,\sigma)\mapsto U^{\otimes n}\sigma$.\footnote{Given a representation $V$ of $\SU(2)$ and a representation $W$ of $\Sn$, the external tensor product representation $V\boxtimes W$ of $\SU(2)\times \Sn$ is the vector space $V\otimes W$ where the action of $(U,\sigma)\in \SU(2)\times \Sn$ is given by
$(U,\sigma) v\otimes w = Uv\otimes \sigma w$ for all $v\otimes w\in V\boxtimes W$. This is not to be confused with the tensor representation $V\otimes W$ of a fixed group $\sfG$, wherein the action is the diagonal action $g\cdot(v\otimes w) = gv\otimes gw$.}
\begin{theorem}[{Qubit Schur-Weyl duality, adapted from~\cite[Theorem IX.II.3]{simon_representations_1996}}]\label{thm:qubit schur-weyl}
Take the subset of spins $\calS\subseteq \widehat{\SU}(2)$ to be $\calS = \{0,1,\dots, n/2\}$ if $n$ is even and $\calS = \{\frac{1}{2}, \frac{3}{2}, \dots, n/2\}$ if $n$ is odd.

Then there is a one-to-one map $Q:\calS\to \widehat{\sS}_n$ such that the representation $\calH$ decomposes as
\begin{align*}
    \calH &= \bigoplus_{s\in \calS} V_s\boxtimes W_{Q(s)} . 
\end{align*} 
\begin{enumerate}[label=(\roman*)]
    \item Each irrep $V_s\boxtimes W_{Q(s)}$ of $\SU(2)\times \Sn$ appears with multiplicity 1.
    \item $\dim(V_s) = 2s+1$.
    \item $\dim(W_{Q(n/2)}) = 1$ and $\dim(W_{Q(s)}) = \binom{n}{n/2 - s} - \binom{n}{n/2 - s - 1}$ for $s<n/2$.
\end{enumerate} 
\end{theorem} To connect this to other presentations of Schur-Weyl duality, one need only note that the collection of valid spins $\calS$ may be determined by repeatedly using the Clebsch-Gordan decomposition applied to $\calH = (\C^2)^{\otimes n}$. The dimension of $W_{Q(s)}$ may then be computed using the hook-length formula (see e.g.~\cite[Theorem VI.2.3]{simon_representations_1996}) or alternatively by a standard character calculation.

Up to a sign and a shift in the ground state energy, the Heisenberg model $H$ is exactly the quadratic Casimir on the representation $\calH$. As a result, both the spectral projectors and the eigenvalues of $H$ are completely determined by the irrep decomposition of $\calH$.
We relegate the proof of the following well-known proposition to~\cref{secapp:casimir on n qubits}.

\begin{proposition}[Heisenberg spectrum]\label{prop:diagonalization of H}
    Let the set of spins $\calS$ be as in~\cref{thm:qubit schur-weyl}. Then, up to a shift and rescaling of the spectrum, the Heisenberg model Hamiltonian~\cref{def:heisenberg model} may be diagonalized as 
    \begin{align}
        H = -\frac{1}{n}\sum_{s\in \calS} s(s+1) \Pi_s, 
    \end{align}
    where $\Pi_s$ is the orthogonal projection onto the irrep $V_s\boxtimes W_{Q(s)}$ appearing in~\cref{thm:qubit schur-weyl}. 
\end{proposition}

Let us make a few comments. The spectrum of $H$ consists of the $O(n)$ rational energy labels $E_s=-s(s+1)/n$. The ground state eigenspace corresponds to the highest spin $s=n/2$ and has dimension $n+1$. This space of ferromagnets is the symmetric subspace of $\calH$, the set of vectors invariant under all permutations. Meanwhile, antiferromagnetic states correspond to the lowest spin $s=0$ for even $n$ (or $s=1/2$ for odd $n$), and this eigenspace has dimension $\binom{n}{n/2} - \binom{n}{n/2 -1}$,
scaling exponentially in $n$.

\subsection{The Davies Generator}

 The Davies generator describes an idealized model of thermalization in open quantum systems. 
 In the Schr\"odinger picture, its time evolution defines a quantum Markov semigroup of completely positive and trace-preserving maps.
 We dedicate this section to a basic introduction, and to highlight some of its key properties. 

Let $H=H^\dagger$ be a generic Hamiltonian. Given an inverse temperature $\beta>0$, the \emph{Gibbs state} $\rho\in \calB(\calH)$ is the density matrix given by \begin{equation}
    \rho := \frac{1}{Z(\beta)} e^{-\beta H}, \qquad Z(\beta) := \Tr \left[e^{-\beta H}\right]. 
\end{equation}
Let the set of \emph{Bohr frequencies} be $ B(H) := \{\lambda_i - \lambda_j : \lambda_i,\lambda_j \in \spec(H) \}.$
We write $\Pi_\lambda$ to denote the projection onto the eigenspace of energy $\lambda\in \spec(H) $. 

\begin{remark}
    As discussed in \cref{sec:Heisenberg_spectrum}, the eigenspaces of the Heisenberg Hamiltonian \cref{def:heisenberg model} can be indexed by spins $s\in\calS$, where the associated eigenvectors have energy $\lambda_s=-s(s+1)/n$.  For simplicity of notation, when discussing the Heisenberg model we refer to $\Pi_s:=\Pi_{\lambda_s}$ interchangeably.
\end{remark}

\noindent Given an operator $A^a\in \calB (\cal H)$, its \emph{Bohr frequency decomposition} is the expansion:
\begin{align}\label{eq:bohr_frequency_decomposition}
    A^a = \sum_{\omega \in B(H)} A^a(\omega), \quad \text{where} \quad A^a(\omega) = \sum_{\lambda \in \text{spec}(H)} \Pi_{\lambda+\omega} \,  A^a \, \Pi_\lambda.
\end{align}
We also require a weight function $\gamma\colon \mathbb{R} \to \mathbb{R}^+$ which satisfies the classical detailed-balance condition
\begin{align}\label{eq:metropolis_weight}
    \frac{\gamma(\omega)}{\gamma(-\omega)} = e^{-\beta\omega}, \quad \text{ e.g. }\quad  \gamma_\mathsf{M}(\omega) = \min\{ 1, e^{-\beta\omega}\} .
\end{align}
For the present work, it will suffice to consider the \emph{Metropolis weight} $\gamma:=\gamma_\mathsf{M}.$ We are now in a position to define the Davies generator. 
\begin{definition}[The Davies Generator]\label{def:davies_generator}
    Given a set of jump operators $\calJ=\{A^a\}_{a}$, the Davies generator $\lind_\calJ \colon \calB(\calH) \to \calB(\calH)$ in the Heisenberg picture is defined as the superoperator:
\begin{align}\label{eq:lindblad_def}
    \lind_{\calJ}(X) = \sum_{\omega \in B(H)} \gamma(\omega) \sum_{A^a \in \calJ} \left[ A^a(\omega)^\dagger X A^a(\omega) - \frac{1}{2} \acommt{A^a(\omega)^\dagger A^a(\omega)}{X} \right].
\end{align}
If $\calJ$ has only one jump $A$, we suppress the braces and write $\lind_A$. The generator $\calL_{\loc}$ corresponds to the set of single-site Pauli jumps $\{S_i^{\alpha}: \alpha\in \{X,Y,Z\}, i\in [n]\}$.
\end{definition}

With this setup, it is known that $\lind$ is KMS detailed-balanced (c.f.~\cref{def:KMS DBC}) and negative semi-definite $\calL\leq 0$.\footnote{In fact, it is well known that the Davies generator satisfies the stricter condition known as GNS detailed-balance, which is self-adjointness with respect to the GNS inner product $\inprod{A,B}_{\mathrm{GNS}} = \Tr \rho A^\dagger B$.} Evidently $\calL(\idty)=0$, and so its Hilbert-Schmidt adjoint fixes the Gibbs state $\CL^\dagger(\rho) = 0$. We choose the jumps in $\calJ$ so that $\comm(\calJ) = \C \idty$ and thus $\calL$ will be \emph{primitive}, i.e. $\calL^\dagger$ has $\rho$ as the \textit{unique} fixed point of the time evolution $\exp(t\CL^\dagger)$.
For a KMS detailed-balanced primitive Lindbladian $\CL$, the mixing time can be partially captured by an estimate on its spectral gap (see e.g.~\cite{kastoryano_quantum_2013,chen_quantum_2025-1}):
\begin{equation}\label{eq:gap_and_mixing_time}
\frac{1}{{\gap}(\mathcal{L})}\,\log\!\left(\frac{\lambda_{\min}(\rho)}{2\,\epsilon}\right)
\; \leq \;
t_{\mathrm{mix}}(\epsilon) \;\le\; \frac{1}{2\, {\gap}(\mathcal{L})}\,\log\!\left(\frac{1}{\lambda_{\min}(\rho)\,\epsilon^2}\right),
\end{equation}
where $\lambda_{\min} (\rho)$ is the smallest eigenvalue of $\rho$, and $t_{\mathrm{mix}}(\epsilon)$ is the time to converge to $\rho$ up to trace distance $\epsilon$.
Integral in our analysis will be the following explicit formula for the Dirichlet form of the Davies generator.

\begin{lemma}[{Davies divergence form~\cite[Proposition 2.9]{basso_quantum_2025}}]\label{lem:dirichlet_divergence} In the context of \cref{def:davies_generator}, the Dirichlet form $\calE$ (\cref{def:dirichlet form and variance}) of the Davies generator $\lind_\calJ$ can be written as
\begin{align}
    \calE[X] = \frac{1}{2} \sum_{a,\omega} h_\omega \cdot \bigg\| \commt{A^a(\omega)}{X} \bigg\|_\rho^2.
\end{align}
with $h_\omega:=e^{\frac{\beta\omega}{2}} \gamma(\omega)$. Under the Metropolis weight $\gamma_{\mathsf{M}}$, $h_\omega=
e^{-\beta|\omega|/2}.$
\end{lemma}

\section{Proof of \texorpdfstring{\cref{theorem:main_davies_gap}}{Theorem 1.1}}
\label{section:main-results}

We dedicate this section to the proof of \cref{theorem:main_davies_gap}, beginning with an outline in~\cref{section:proof outline}.

\subsection{Proof Outline} \label{section:proof outline}

The content of \cref{theorem:main_davies_gap} is an estimate of the spectral gap of a KMS self-adjoint, negative semi-definite generator with a one-dimensional kernel, $\C\idty$. Following \cref{definition:gap}, we compute the spectral gap using its variational characterization:
    \begin{align}\label{eq:gap_def_intro}
 {\gap}(\calL_\loc)  =
 \min_{\substack{X \in \idty^{\perp_\rho} \\[.2ex] X \neq 0}}
 \frac{\langle X, -\calL_\loc (X) \rangle_{\rho}}{\|X\|_{\rho}^2}\,.
\end{align}
This formulation reduces the computation to evaluating the Rayleigh quotient of arbitrary observables $X\in \calB(\calH)$.
\cref{section:Bh-decomp,sec:perm group mixer lower bound,sec:pauli_master_equation,sec:low temp lower bound A ell spaces,sec:high temp monotonicity} comprise the proof of lower bounds on the gap, and we will outline these sections below. The matching upper bounds are presented in \cref{sec:order parameter witnesses the gap} and amount to computing this Rayleigh quotient for the total magnetization $S_{\tot}^Z$. 
To proceed, we first organize $\calB(\calH)$ according to the symmetries in the Hamiltonian.

\vspace{1em}
\noindent \textbf{Decomposing $\calB(\calH)$ using symmetry (\cref{section:Bh-decomp}).}
The Heisenberg Hamiltonian and the collection of single-site Pauli jumps both admit a $\SU(2)\times\Sn$ symmetry. A natural decomposition of the observable algebra $\calB(\calH)$ then lies in the direct sum of subspaces which commute with both of these group actions, and those (KMS) orthogonal:
\begin{align}
    \calB(\calH)&=\comm(\Sn)\oplus \comm(\Sn)^{\perp_\rho} \quad \text{and} \label{eq:BH-decomp}\\
    \comm(\Sn)&=
    \comm(\SU(2)\times\Sn) \oplus  \bigg( \comm(\SU(2))^{\perp_\rho}\cap \comm(\sS_n)\bigg).
\end{align}
where we first divide $\calB(\calH)$ into permutation-invariant and non-invariant terms, and then further fine-grain $\comm(\Sn)$ based on rotation-invariance. 
Here, $\calA^{(0)}:=\comm(\SU(2)\times\Sn)$ is the sector of $\SU(2)$ rotation- and permutation-invariant observables, and will play a particularly important role in our analysis. 

The starting point to our argument is the observation that the Davies generator $\calL_\loc$ inherits this symmetry, and so by Schur's lemma each subspace above is an invariant subspace of the dynamics. Thus, we can separately estimate the Dirichlet form on each of these sectors as follows.

\vspace{1em}
\noindent \textbf{The permutation group mixer comparison argument (\cref{sec:perm group mixer lower bound}).} Following the decomposition in \cref{eq:BH-decomp}, we first wish to understand the dynamics $\calL_{\loc}$ on $\comm(\Sn)^{\perp_\rho}$, the complement of permutation-invariant observables. 
Intuition suggests, due to the $\sS_n$ symmetry and related analyses in the Curie-Weiss model \cite{ding2009mixing, levin2010glauber}, that these non-invariant degrees of freedom should converge relatively quickly. 
To capture this intuition, we introduce the \textit{group mixer comparison argument}, which consists of two steps: 

\begin{enumerate}
    \item A choice of group mixer, which is an auxiliary Davies generator $\calL_{\sS_n}$ introduced only for the purposes of analysis. $\calL_{\sS_n}$ acts akin to depolarizing noise on the nontrivial irreps of $\sS_n$ comprising $\comm(\Sn)^{\perp_\rho}$.
\end{enumerate}

A rich supply of group mixers is provided by the method of group transference~\cite{bardet_group_2021}, which allows one to transfer certain classical Markov generators on groups to quantum Markov semigroup generators on group representations in a manner which preserves many spectral properties. Here, by transferring a carefully chosen random walk on $\sS_n$, we obtain the group mixer $\calL_{\Sn}$ satisfying
\begin{equation}
    \ker(\calL_{\Sn})=\comm(\Sn),\qquad \gap(\calL_{\Sn})=\Omega(1).
\end{equation}

\begin{enumerate}
    \item[2.] A simulation argument, which argues that the Davies generator $\calL_{\Sn}$ can be simulated by the single-site dynamics $\calL_{\loc}$, by comparing their Dirichlet forms:
\end{enumerate}
\begin{equation}
 \forall X\in \calB(\calH):\quad    \langle X,-\calL_{\Sn}(X)\rangle_\rho \leq c_\beta \cdot\langle X,-\calL_{\loc}(X)\rangle_\rho.\label{eq:D-form-comparison}
\end{equation}
Here, $c_\beta>0$ is finite for each fixed $\beta$ and independent of $n$. 
We use the same method again for the $\SU(2)$ group mixer, and the argument can be a reusable ingredient for studying other symmetric models. Combining the constant gap of $\calL_{\Sn}$ with the comparison gives a constant lower bound for $-\calL_{\loc}$ on $\comm(\Sn)^{\perp_\rho}$.

\vspace{1em}
\noindent \textbf{The coarse-grained Pauli master equation (\cref{sec:pauli_master_equation}).}
We next treat the $\SU(2)\times \Sn$-invariant sector $\calA^{(0)}$, which for the mean-field Heisenberg model coincides with the span of its energy projectors:
\begin{equation}
    \calA^{(0)}=\comm(\SU(2)\times\Sn)=\mathrm{span}(\Pi_s:s\in \calS).
\end{equation}
The restriction of $\calL_{\loc}$ to this span induces a classical Markov on the total-spin labels $\calS$, and can be thought of as a lumped Pauli master equation \cite{lidar_lecture_2019,kemeny_finite_1976}.
Furthermore, the Wigner-Eckart selection rules make this Markov chain a birth-death chain.
Cheeger's inequality and a discrete version of Laplace's method then give a constant gap on $\calA^{(0)}$ for all $\beta\neq 2$, and an $\Omega(n^{-1/2})$ lower bound on this sector at $\beta=2$.

\vspace{1em}
\noindent \textbf{The $\SU(2)$ group mixer comparison argument (\cref{sec:low temp lower bound A ell spaces}).}
To handle the observables which are permutation-invariant but not $\SU(2)$ invariant, $\comm(\SU(2))^{\perp_\rho}\cap \comm(\sS_n)$ in~\cref{eq:BH-decomp}, a second group mixer comparison gives the correct scaling in $n$ at low temperatures.
We introduce an $\SU(2)$ group mixer $\calL_{\su(2)}$ (which is the quadratic Casimir on observables), whose kernel is $\comm(\SU(2))$ and whose spectrum is bounded away from zero by a constant on nontrivial $\SU(2)$ sectors, $\comm(\SU(2))^{\perp_\rho}$.
A unitary freedom argument for jump operators shows that
\begin{equation}
    \calL_{\loc}=\frac{1}{n}\calL_{\su(2)}+\calL_{\mathrm{rest}}\, ,
\end{equation}
where $\calL_{\mathrm{rest}}$ is negative semi-definite.
Therefore $-\calL_{\loc}$ is bounded below by $- n^{-1}\cdot \calL_{\su(2)}$, giving an $\Omega(n^{-1})$ lower bound on every nontrivial $\SU(2)$ sector.
This bound is sharp in the low-temperature phase $(\beta>2)$.

\vspace{1em}
\noindent \textbf{The spherical tensor basis and monotonicity in $\ell$ (\cref{sec:high temp monotonicity}).}
In the high temperature regime ($\beta<2$), the $\SU(2)$ group mixer comparison argument above gives only an $n^{-1}$ lower bound on the nontrivial sectors, which (by extrapolating to the $\beta=0$ limit) cannot possibly be tight. To improve the spectral gap to constant up to the phase transition, we pass to the $\SU(2)$-isotypic decomposition of $\comm(\Sn)$:
\begin{equation}\label{eq:Sn-decomp}
    \comm(\Sn) = \calA^{(0)}\oplus \bigoplus_{\ell = 1}^n \calA^{(\ell)}, 
\end{equation} where the direct sum over nonzero $\ell$ further decomposes $\comm(\SU(2))^{\perp_\rho}\cap \comm(\sS_n)$. 

An explicit basis reflecting this decomposition is given by the spherical tensor operators $T_{s,\ell,q}$, whence Wigner-Eckart selection rules expedite the analysis. 
We then prove that the smallest eigenvalue of $-\calL_{\loc}|_{\calA^{(\ell)}}$ increases monotonically with $\ell$. It remains to estimate the base sector $\calA^{(1)}$, where an explicit Dirichlet form bound gives a lower bound of order $2-\beta$, hence a constant lower bound for each fixed $\beta<2$.

\subsection{Decomposing \texorpdfstring{$\calB(\calH)$}{B(H)} Using Symmetry}
\label{section:Bh-decomp}
The story begins by using the  $\SU(2)$ and $\Sn$ symmetries to decompose the space of observables $\calB(\calH)$. As we recall in~\cref{secapp:reps in this work}, the actions of $\SU(2)$ and $\Sn$ on $\calH$ induce a natural representation $\calB(\calH)$ by adjoint action, and since they commute, we may think of $\calB(\calH)$ as a representation of the product group $\SU(2)\times \Sn$ by defining the map $\Delta:\SU(2)\times \Sn\to \calB(\calB(\calH))$,
\begin{equation}
    \Delta_{U,\sigma} (X) := (U^{\otimes n} \sigma) X (U^{\otimes n} \sigma)^{-1} \, \qquad \text{for all } X\in \calB(\calH) . 
\end{equation}
Note that since the Gibbs state $\rho$ commutes with every $U^{\otimes n}$, $U\in \SU(2)$, and every $\sigma\in \Sn$,~\cref{obs:KMS unitarity from hamiltonian symmetry} implies that $\Delta$ is a KMS unitary representation of the product group $\SU(2)\times \Sn$.
It is perhaps intuitive that the Davies generator $\calL_{\loc}$ should inherit the symmetries of the problem.\footnote{We note that essentially the same proof implies that a rather wide class of Gibbs samplers \cite{chen2023efficient, ding2024efficient, jiang2024quantum} inherits the symmetries of their respective Hamiltonians, a statement which may find independent interest. See~\cref{rem:lind_inherit_symmetry_operator_fourier transform}.} The proof of this statement is relegated to~\cref{secapp:proving intertwiners}.

\begin{theorem}[$\lind_\loc$ is an intertwiner]\label{thm:davies is intertwiner}
    The Davies generator $\calL_{\loc}$ \cref{def:davies_generator} is an intertwiner for the action of $\SU(2)\times \Sn$, i.e. for any observable $X\in\calB(\calH)$ we have
    \begin{align}
        \calL_{\loc}\circ \Delta_{U,\sigma}(X) = \Delta_{U,\sigma} \circ \calL_{\loc}(X) \qquad \text{ for all } (U,\sigma)\in \SU(2)\times \Sn.
    \end{align}
\end{theorem}

This theorem opens the door to prolific use of Schur's~\cref{lem:schur}, which ensures that $\calL_{\loc}$ may only map nontrivially between irreps of the same type.  
Our first decomposition is a simple one: by KMS unitarity, we may decompose into the space of permutation-invariant observables and its orthogonal complement, 
\begin{equation}
    \calB(\calH) = \comm(\Sn)\oplus \comm(\Sn)^{\perp_\rho}.
\end{equation}
In representation-theoretic terms, $\comm(\Sn)$ is the direct sum of all trivial irreps of $\Sn$ appearing in $\calB(\calH)$. Thus it is an invariant subspace of $\calL_{\loc}$, and since $\calL_{\loc}$ is KMS detailed-balanced, the orthogonal complement is similarly invariant.

\begin{observation} \label{obs:L respects the BH decomp}
    By Schur's~\cref{lem:schur} and~\cref{thm:davies is intertwiner}, both the commutant of $\Sn$ and its orthogonal complement are invariant subspaces of the Davies generator $\calL_{\loc}$, i.e. 
    \begin{align}
        \calL_{\loc}:\comm(\Sn)\to \comm(\Sn), \qquad \calL_{\loc}:\comm(\Sn)^{\perp_\rho}\to \comm(\Sn)^{\perp_\rho}. 
    \end{align} In particular, since $\calL_{\loc}$ is primitive and $\idty \in \comm(\Sn)$, we have
    \begin{equation}
        \gap(\lind_\loc) = \min\left\{ \gap\left(\lind_\loc\Big\vert_{\comm(\Sn)} \right), \; \min \spec \left(-\lind_\loc\Big\vert_{\comm(\Sn)^{\perp_\rho}} \right) \right\},
    \end{equation}
    where $\min\spec$ denotes the smallest eigenvalue of a positive linear map.
\end{observation}

We will require a finer decomposition of $\comm(\Sn)$ for the upcoming analysis, and it is here that we will leverage $\SU(2)$ symmetry. By complete reducibility of unitary representations (c.f.~\cref{thm:complete reducibility}) we may take the isotypic decomposition of $\comm(\Sn)$, thought of as a representation of $\SU(2)$: 
\begin{equation} \label{eq:isotypic decomp comm Sn}
    \comm(\Sn) = \calA^{(0)}\oplus \calA^{(1)} \oplus \dots \oplus \calA^{(n)}, 
\end{equation} where each $\calA^{(\ell)}$ is a direct sum of irreps of total spin $\ell\in \widehat{\SU}(2)$ (c.f. \cref{eq:SU2 irreps are half integers}). We pause to note that the space of spin-0 irreps $\calA^{(0)}$ is particularly important: this is exactly the space of trivial reps of $\SU(2)$, and so unraveling definitions reveals that $\calA^{(0)} = \comm(\SU(2)\times \Sn)$. As a consequence of Schur-Weyl duality and that the Heisenberg model is the quadratic Casimir, we may obtain a third characterization of this space.
\begin{lemma}\label{lem:comm algebra A}
The spectral projectors of $H$ generate a commutative subalgebra of the observables $\calB(\calH)$, and we have
    \begin{equation} 
    \mathrm{span}\{\Pi_{s}: s\in \calS \} = \calA^{(0)} =  \comm(\SU(2)\times \Sn).
\end{equation}
\end{lemma}

\begin{proof}
    It is immediate that the span of the projectors is contained in $\calA^{(0)}$, since the spectral projectors are invariant under $\SU(2)\times \Sn$. To see the equality, we use the perspective $\calA^{(0)} = \comm(\SU(2)\times \Sn)$.
    Schur-Weyl duality (c.f.~\cref{thm:qubit schur-weyl}) guarantees that each irrep $V_s\boxtimes W_{Q(s)}$ arising in the decomposition of $\calH$ is of multiplicity 1, so Schur's~\cref{lem:schur} implies that the dimension of $\comm(\SU(2)\times \Sn)$ is at most $\abs{\calS}$. This proves the claim.
\end{proof}

We continue with the decomposition program.
Just as before, Schur's Lemma does the heavy lifting.

\begin{observation}[$\lind_\loc$ respects the $\comm(\Sn)$ decomposition]\label{obs:L respects the comm Sn decomp}
    By Schur's~\cref{lem:schur}, the spaces in the isotypic decomposition \cref{eq:isotypic decomp comm Sn} are invariant subspaces of the intertwiner $\calL_{\loc}$, meaning it restricts to maps 
    \begin{equation} \calL_{\loc}:\calA^{(\ell)}\to \calA^{(\ell)}.
    \end{equation} Since $\calL_{\loc}$ is primitive and $\idty\in \calA^{(0)}$, we have 
    \begin{align}
        \gap\left(\lind_\loc\Big\vert_{\comm(\Sn)}\right) = \min\left\{ \gap\left(\lind_\loc\Big\vert_{\calA^{(0)}}\right), \; \min_{\ell\geq 1} \paran{ \min \spec \left(-\lind_\loc\Big\vert_{\calA^{(\ell)}} \right)} \right\}.
    \end{align}
\end{observation}

\subsection{The Lower Bound on \texorpdfstring{$\comm(\Sn)^{\perp_\rho}$}{comm(Sn) perp}} \label{sec:perm group mixer lower bound}
Our goal in this section is to show the following lower bound. 

\begin{theorem}[$\calL_{\loc}$ lower bound on $\comm(\Sn)^{\perp_\rho}$] \label{thm:Lloc lower bound on commSn perp}
    When restricted to $\comm(\Sn)^{\perp_\rho}$, the Davies generator $\calL_{\loc}$ is bounded away from zero uniformly in system size:
    \begin{equation}
        -\calL_{\loc}\Big\vert_{\comm(\Sn)^{\perp_\rho}} \geq \Omega(1) \idty \Big\vert_{\comm(\Sn)^{\perp_\rho}},
    \end{equation} where we have suppressed the explicit $\beta$ dependence. 
\end{theorem}

The proof will require two key steps.
The first is the introduction of a new Davies generator $\calL_{\Sn}$ whose jump operators are given by a certain collection of transpositions in $\Sn$. As the notation suggests, this Lindbladian is intimately connected to the group action of $\Sn$, and crucially its kernel and gap may be directly analyzed using only representation-theoretic data, thanks to the method of group transference~\cite{bardet_group_2021}. This technique allows one to import bounds on mixing times and spectral gaps from certain classical Markov chains on groups to corresponding quantum Markov semigroups. We call the resultant Lindbladians ``group mixers'', as they allow us to guarantee fast mixing on nontrivial irreps.\footnote{Another group mixer, corresponding to the group $\SU(2)$, will appear later in~\cref{sec:low temp lower bound A ell spaces}.}

The second component is an argument showing that the single-site dynamics $\calL_{\loc}$ ``simulates'' the group mixer $\calL_{\Sn}$. The output yields an operator inequality bounding $\calL_{\loc}$ away from 0 on $\comm(\Sn)^{\perp_\rho}$. The simulation incurs a $\beta$-dependent price but crucially is independent of the system size $n$.

\paragraph{The \texorpdfstring{$\calL_{\Sn}$}{L_Sn} generator.}
This is a Davies generator (c.f. \cref{def:davies_generator}) where the coupling operators are a subset $T \subseteq \Sn$ of pairwise transpositions:\footnote{While this equation may not look like \cref{def:davies_generator}, we note the latter simplifies to \cref{eq:lind_Sn_def} since $[H, \sigma]=0$ and thus transposition operators only have Bohr frequency $0$ under the Heisenberg Hamiltonian.}
\begin{align}\label{eq:lind_Sn_def}
    \calL_{\Sn}(X) &:= \sum_{\sigma \in T} \calL_{\sigma}(X) = \sum_{\sigma \in T} \sigma X \sigma - X.
\end{align}
The case where $T$ consists of every transposition $\mathsf{SWAP}_{ij}$ (i.e. the complete graph), is the generator of the  transposition quantum Markov semigroup studied by \cite[Section IV.D]{bardet_group_2021}. However, the spectral gap of their generator scales as $n^{-2}$; a more frugal choice of the set $T$ will allow us to extract a constant spectral gap. We pick $T$ to be the transpositions defined on the edges of an expander graph on $n$ vertices. 
There are many explicit constructions of expander graphs, and for concreteness we invoke:
\begin{fact}[Expander walks, \cite{reingold_entropy_2002}]\label{fact:expander_properties}
There exists an explicit infinite family of $O(1)$-regular graphs $G_n = ([n], E_n)$ such that the transpositions corresponding to $E_n$ generate $\Sn$ (i.e. $E_n$ is connected) and a random walk on $G_n$ has spectral gap $\Omega(1)$.
\end{fact}

\begin{proposition} \label{prop:Sn mixer spectrum}
     The kernel and spectral gap of the Lindbladian $\calL_{\Sn}$ are 
     \begin{equation}
         \ker(\calL_{\Sn}) = \comm (\mathsf{S}_n),  \qquad \mathrm{gap}(\calL_{\Sn}) = \Omega(1).
     \end{equation}
\end{proposition}
\begin{proof}
The kernel of this Lindbladian is the commutant of its jump operators \cite{wolf_quantum_2012}. Since the jump operators generate $\Sn$, this is exactly $\comm(\Sn)$. To compute the gap of $\lind_{\Sn}$, we note that it is a transferred version of the classical random walk on $G_n$. Combining \cref{fact:expander_properties} with \cite[Theorems II.4 and III.2]{bardet_group_2021} and observing that this random walk has a right-invariant kernel finishes the proof. 
\end{proof}

\paragraph{The simulation argument.}
The following proposition makes rigorous the idea that $\calL_{\loc}$  ``simulates'' $\calL_{\Sn}$. Intuitively, one might hope for this sort of result by noting that the jump operators of $\calL_{\loc}$ are single-site Pauli operators, while those of $\calL_{\Sn}$ are merely ``two-site'' operators. Akin to path-comparison arguments in the classical Markov chain literature \cite{Diaconis1993COMPARISONTF}, one could then expect that the thermalization under the latter may be accomplished by running the former (perhaps up to some $\beta$-dependent penalty).

\begin{proposition}[$\lind_\loc$ simulates $\lind_{\Sn}$]\label{prop:removing_NN} There exist universal constants $c_1, c_2\in \mathbb{R}^+$ such that, for every $\beta \in \mathbb{R}^+$, the following operator inequality holds: 
\begin{align}
    - \calL_{\Sn} \leq - c_1\cdot e^{\beta c_2}\cdot \calL_\loc.
\end{align}
\end{proposition}

To make the intuition behind \cref{prop:removing_NN} precise, we rely on techniques from \cite{chen2025quantumMarkov, bergamaschi2025quantumspinchainsthermalize, BCV25}, to relate the Dirichlet form of the two Davies generators assuming the convergence of the \textit{complex-time evolution} of local operators in operator norm. We find these statements may be of independent interest. The key feature of the Heisenberg model that makes this possible is that, by the Wigner-Eckart~\cref{thm:wigner eckart}, a single-site Pauli jump may not change the energy by more than a constant in system size (c.f. \eqref{eq:wg-cte} below). The proof may be found in~\cref{sec:removing_NN}.

One arrives at~\cref{thm:Lloc lower bound on commSn perp} by simply combining~\cref{prop:Sn mixer spectrum,prop:removing_NN}.

\subsection{The Gap on \texorpdfstring{$\calA^{(0)}$}{A(0)}
}\label{sec:pauli_master_equation}
\noindent Our next goal is to understand $\lind_\loc$ restricted to the subspace $\calA^{(0)}=\mathrm{span}\{\Pi_s:s\in \calS\}$. The main contribution of this section is the following characterization of $\lind_\loc|_{\calA^{(0)}}$ below, above, and at the critical point $\beta=2$.

\begin{theorem}[The spectral gap of $\lind_\loc|_{\calA^{(0)}}$]\label{thm:gap_of_LA}
    For any $\beta$, the spectral gap of the generator $\lind_\loc|_{\calA^{(0)}}$  can be bounded below as follows:
    \begin{align}
      \gap\big(\lind_\loc|_{\calA^{(0)}}\big) =     \begin{cases}
        \Omega(1) & 0 \leq \beta < 2, \\[.5ex]
        \Omega(n^{-1/2}) & \beta = 2, \\[.5ex]
        \Omega(1) & \beta > 2.
      \end{cases}
    \end{align}
  \end{theorem} 
  In what follows we provide a high-level presentation of the two key ingredients of the proof:
  
  \begin{enumerate}
      \item The starting point is to understand $\lind_\loc|_{\calA^{(0)}}$ as a classical Markov chain generator $(L_{s, s'})_{s, s'\in \calS}$ whose state space $\calS$ corresponds to the (degenerate) eigenspaces of $H$. This chain can be seen as the coarse-graining of the usual Pauli master equation, i.e. the classical Markov chain whose states are the $2^n$ eigenvectors. The Wigner-Eckart theorem will ensure the resulting matrix $L_{s, s'}$ is    \textit{tridiagonal} and so describes a classical birth-death chain (\cref{prop:L is tridiagonal}). 
      \item The mixing time of the birth-death chain can be analyzed via Cheeger's inequality (\cref{thm:cheeger}), and ultimately reduces to sharp tail bounds on the cumulative distribution function of the stationary probability distribution $\pi$ on $\calS$, where $\pi(s):= \Tr[\rho \, \Pi_s]$. The tail bounds are computed using a careful application of Laplace's method. \\
  \end{enumerate}

\noindent \textbf{Wigner-Eckart and the birth-death chain.} We begin by computing the explicit matrix entries of $\lind_\loc|_{\calA^{(0)}}$, which we  denote by $L_{s, s'}$ for $s,s'\in\calS$. Since $\{\Pi_s \colon s \in \calS\}$ forms an orthogonal (unnormalized) basis for $\calA^{(0)}$, it follows that
\begin{align}\label{eq:pauli_master_generator}
    L_{s, s'} := \frac{\langle \Pi_{s'}, \lind_\loc(\Pi_{s}) \rangle_\rho}{\langle \Pi_{s'}, \Pi_{s'} \rangle_\rho} = \frac{\langle \Pi_{s'}, \lind_\loc(\Pi_{s}) \rangle}{\langle \Pi_{s'}, \Pi_{s'} \rangle} = 
    \frac{\Tr\left[ \Pi_{s'} \lind_\loc(\Pi_{s}) \right]}{\Tr\left[ \Pi_{s'} \right]},
    \end{align}
    where we used $\rho \, \Pi_{s} = Z(\beta)^{-1} e^{(\beta / n) s (s + 1)} \Pi_{s}$ in the second equality to change the KMS inner products to Hilbert-Schmidt inner products.
    The matrix \(L_{s, s'}\) can be interpreted as the generator of a classical Markov chain with state space $\calS$, since
    \begin{enumerate}
        \item  $\sum_{s} L_{s, s'} = 0$, because $\sum_{s} \Pi_{s} = \idty$ and $\lind_\loc(\idty)=0$ in \cref{eq:pauli_master_generator};
        \item  $L_{s, s'} \geq 0$ for $s \neq s'$, which may be directly seen from~\cref{eq:lindblad_def}.
    \end{enumerate}
    
\noindent 
We can thus use classical tools to analyze $L$, noting that by definition $\gap\big(\lind_\loc|_{\calA^{(0)}}\big)=\gap(L)$.\footnote{Our convention for the generator $L_{s,s'}$ is the transpose of the usual since we consider $\lind_\loc$ instead of $\lind_\loc^\dagger$.}
To explicitly compute its entries, we employ the Wigner-Eckart Theorem, which greatly restricts the allowed transitions between eigenspaces from single-site Pauli jumps. The key consequence\footnote{More precisely, this is a consequence of the selection rules of the Clebsch-Gordan coefficients.} we use is that no single jump may change the energy by too much, i.e. for all $i\in [n]$ and $\alpha = X,Y,Z$, 
\begin{equation}
    \Pi_s S_i^\alpha \Pi_{s'} = 0 \qquad \text{ if } s\notin \{s'-1,s',s'+1\}.\label{eq:wg-cte}
\end{equation} After some calculation, one sees that the matrix $L_{s, s'}$ inherits this tridiagonality, whence the computation of its entries becomes a tractable linear algebra problem. For more on Wigner-Eckart and a proof of \cref{prop:L is tridiagonal}, see \cref{sec:wigner_eckart}.
    \begin{proposition}[$L$ is a birth-death chain]\label{prop:L is tridiagonal}
        The generator $L \in \mathbb{R}^{|\calS| \times |\calS|}$ \cref{eq:pauli_master_generator} is tridiagonal and its entries are given by
\begin{align}\label{eq:entries_of_L}
        L_{s, s'} =
        \begin{cases}
           \displaystyle \gamma\left(\frac{-2 s - 2}{n}\right) \frac{2 s + 3}{2 s + 1} \frac{n - 2s}{4} &\quad \text{if} \quad s' = s+1, \\[2ex]
           \displaystyle \gamma\left(\frac{2 s}{n}\right) \frac{2 s - 1}{2 s + 1} \frac{n + 2s + 2}{4} &\quad \text{if} \quad s' = s-1, \\[2ex]
            -L_{s + 1, s} - L_{s - 1, s} &\quad \text{if} \quad s'=s, \\
            0, &\quad \text{otherwise.}
        \end{cases}
    \end{align}
  \end{proposition}
  We are now in good shape to leverage classical techniques to analyze this Markov chain. \\

  \noindent \textbf{Tail bounds on $\pi$ via Laplace's method.} We use Cheeger's inequality to extract a lower bound on the spectral gap of the generator $L$, by relating it to certain static properties of the stationary measure $\pi$. We refer the reader to \cref{sec:proof_gap_of_L} for the relevant background. The most technically involved step of this argument requires determining a sharp estimate for the cumulative distribution function \(\pi([s_{\min},m])\), where $s_{\min}$ denotes the minimal spin in $\calS$. We introduce normalized coordinates \(x_{s} := (2 s) / n\), which take values in $[0,1]$, to express the stationary measure in a clearer way.
  In particular, one finds that there exist functions \(g_{\beta}, f_{\beta} : [0, 1] \to \mathbb{R}\) such that the stationary measure is\footnote{Here, the notation \(F(s) \sim G(s)\) means there are constants \(C, c > 0\) independent of \(s, n, \beta\) so that \(c G(s) \leq F(s) \leq C G(s)\).  \(H_{b}(p) := - p \log p - (1 - p) \log(1-p)\) denotes the binary entropy function.}
  \begin{equation}
    \pi(s) \sim g_{\beta}(x_{s}) e^{- n f_{\beta}(x_{s})}, \quad \text{where} \quad f_{\beta}(x) = -\frac{\beta}{4} x^{2} - H_{b}\left( \frac{1 + x}{2}\right).
  \end{equation}  
  It is useful to think of $f_\beta$ as an approximation to the free energy of this system, with $g_\beta$ describing polynomial corrections. Using this bound for the stationary measure and that the points \(x_{s}\) are uniformly spaced with spacing \(2 / n\), we have that
  \begin{equation}
    \sum_{s = s_{\min}}^{m} \pi(s) \sim \sum_{s = s_{\min}}^{m} g_{\beta}(x_{s}) e^{- n f_{\beta}(x_{s})} \approx \frac{n}{2} \int_{0}^{x_{m}} g_{\beta}(y) e^{- n f_{\beta}(y)} d y.
  \end{equation}
  The last integral can then be sharply estimated using Laplace's method. If the function \(f_{\beta}(x)\) has a unique minimum at \(x_{*}\) on \([0,x_{m}]\) and \(f_{\beta}''(x_{*}) > 0\) then Laplace's method gives \cite[Chapter II.1]{wong_asymptotic_2001}
\begin{equation}
  \int_{0}^{x_{m}} g_{\beta}(x) e^{- n f_{\beta}(x)} d x =  \sqrt{\frac{2\pi}{n |f_{\beta}''(x_{*})|}} g_{\beta}(x_{*}) e^{- n f_{\beta}(x_{*})} \cdot (1 + O(n^{-\frac{1}{2}})).
\end{equation}
  Around \(x = 0\), the ``free energy'' \(f_{\beta}\) admits the  Taylor expansion 
  \begin{equation}
    f_{\beta}(x) = -\ln(2) + \frac{1}{4} (2 - \beta) x^{2} + \frac{1}{12} x^{4} + O(x^{6}) ,
  \end{equation}
  which reveals the nature of phase transition at \(\beta = 2\).
  For \(\beta \in [0, 2]\), one can verify that \(x = 0\) is the unique minimizer of \(f_{\beta}(x)\) on \([0, 1]\) whereas for \(\beta > 2\) the minimizer of \(f_{\beta}\) lies at some \(x \in (0, 1)\).

  While the application of Laplace's method is fairly standard, since we are interested in a non-asymptotic statement, significant care is required to pass between the discrete sum and the continuous integral.
  A careful analysis, given in \cref{sec:proof_gap_of_L}, results in the bounds presented in \cref{thm:gap_of_LA}.

 
\subsection{The Low-Temperature Bound on \texorpdfstring{$\calA^{(\ell)}$}{A(ell)}
} \label{sec:low temp lower bound A ell spaces}
It remains to show the following lower bounds on $-\CL_\loc$ restricted to the spaces of observables $\calA^{(\ell)}$ for $\ell\geq 1$ appearing in~\cref{obs:L respects the comm Sn decomp}.\footnote{To connect to the statement provided in the outline, recall that $\comm(\SU(2))^{\perp_\rho}\cap \comm(\Sn) = \bigoplus_{\ell=1}^{n} \calA^{(\ell)}$.} We begin with the low-temperature bound.

\begin{theorem}[$\calL_{\loc}$ lower bound on $\calA^{(\ell)}$, low temperature] \label{thm:Lloc lower bound on Aell}
    Fix an integer $1\leq \ell \leq n$ and an inverse temperature $\beta \geq 2$. When restricted to $\calA^{(\ell)}$, the Davies generator $\calL_{\loc}$ is bounded away from zero:
    \begin{equation}
    \calL_{\loc}\Big\vert_{\calA^{(\ell)}} \geq \Omega(n^{-1})\cdot  \idty \Big\vert_{\calA^{(\ell)}}.
    \end{equation}
\end{theorem}
The proof of this bound is rather simple, requiring only the construction the group mixer $\calL_{\su(2)}$ for the group $\SU(2)$ and leveraging group transference. In fact the bound applies equally well at high temperatures, but there it is not tight and one requires a more sophisticated argument discussed in \cref{sec:high temp monotonicity}.\footnote{In fact, this more sophisticated argument can be applied to low temperatures so as to subsume the group mixer argument, but we feel the group mixer argument worthwhile to present for its simplicity and potential applications beyond this model.} \\

\noindent \textbf{The $\calL_{\su(2)}$ generator.} The spaces of observables $\calA^{(\ell)}$ contain only nontrivial irreps of $\SU(2)$ when $\ell\geq 1$, and so we again turn to a group mixer (and a simulation argument) in hopes of extracting a lower bound on the eigenvalues of the restriction of $\calL_{\loc}$.
Since $\SU(2)$ is a continuous group, the jump operators will be representatives of the Lie algebra $\su(2)$.
We define the group mixer $\calL_{\su(2)}$, which is a Davies generator (c.f. \cref{def:davies_generator}) where the coupling operators are the total Pauli spins defined by $S_{\tot}^\alpha= \sum_{i=1}^n S_i^\alpha$:
\begin{align}\label{eq:lind_su2_def}
   \calL_{\su(2)}(X) := \sum_{\alpha\in \{X,Y,Z\}} \calL_{S_{\tot}^\alpha}(X) = -\frac{1}{2}\sum_{\alpha\in \{X,Y,Z\} } [S_{\tot}^\alpha, [S_{\tot}^{\alpha}, X]], 
\end{align} 
where the expression is simplified by the $\SU(2)$ symmetry of $H$ (namely $[S^\alpha_\tot, H]=0$). This group mixer is a scalar multiple of the quadratic Casimir on $\calB(\calH)$, and this immediately leads to the following proposition, whose complete proof can be found in~\cref{secapp:Casimir on B(H)}.
\begin{proposition}\label{prop:su2 mixer spectrum}
    The kernel and spectral gap of the Lindbladian $\calL_{\su(2)}$ are 
     \begin{equation}
         \ker(\calL_{\su(2)}) = \comm (\SU(2)),  \qquad \mathrm{gap}(\calL_{\su(2)}) = 1.
     \end{equation}
    \end{proposition}
We remark \cref{eq:lind_su2_def} is the transferred quantum Markov semigroup generator from the heat semigroup on $\SU(2)$, and thus its gap can alternatively be computed via group transference~\cite[Theorem II.2]{bardet_group_2021}.
\\

\noindent \textbf{The simulation argument.} As a consequence of the unitary freedom of jump operators, $\CL_\loc$ can be written as a positive linear combination of the rescaled group mixer $\frac{1}{n}\calL_{\su(2)}$ and some other Davies generator. Intuitively, the factor of $n^{-1}$ arises since the jump operators of $\calL_{\su(2)}$ are the total spin operators, which have norm $\Theta(n)$. The proof may be found in~\cref{sec:removing_su(2)}.
\begin{proposition}[$\lind_\loc$ simulates $\lind_{\su(2)}$]\label{prop:removing_su(2)}
    The group mixer $\calL_{\su(2)}$ \cref{eq:lind_su2_def} is contained in the Davies generator $\calL_{\loc}$ of local Paulis, in the sense that there exists a collection of jump operators $\{B_4,\dots , B_{3n}\}$ such that
    \begin{equation}
        \calL_{\loc} = \frac{1}{n}\calL_{\su(2)} + \calL_{\{B_4,\dots , B_{3n}\}}.
    \end{equation} Moreover, since $-\calL_{\{B_4,\dots , B_{3n}\}} \geq 0$, we have $0\leq -\frac{1}{n}\calL_{\su(2)}\leq -\calL_{\loc}$.
\end{proposition}
Combining~\cref{prop:su2 mixer spectrum,prop:removing_su(2)} completes the proof of the low-temperature bound in~\cref{thm:Lloc lower bound on Aell}.

\subsection{The High-Temperature Bound on \texorpdfstring{$\calA^{(\ell)}$}{A(ell)}} \label{sec:high temp monotonicity}

We now turn to showing the high-temperature lower bounds for $-\CL_\loc$ restricted to the spaces of observables $\calA^{(\ell)}$ for $\ell\geq 1$ appearing in~\cref{obs:L respects the comm Sn decomp}.

\begin{theorem}[$\calL_{\loc}$ lower bound on $\calA^{(\ell)}$, high temperature]\label{thm:gap-Al}
    Fix $\beta < 2$. Then, the minimum eigenvalue of $\CL_\loc$ when restricted to $\calA^{(\ell)}$ is lower bounded by a constant independent of system size and the integer $\ell\geq 1$. That is, there exists a universal constant $c\in \mathbb{R}^+$ such that:
    \begin{equation}
      \forall\ell\geq 1:\quad  \min_{\substack{O \in \calA^{(\ell)} \\[.2ex] O \neq 0}} \frac{\langle O, -\CL_\loc(O)\rangle_\rho}{\|O\|_\rho^2} \geq c\cdot (2-\beta).
    \end{equation} 
\end{theorem}

There are two main steps to our proof.
\begin{enumerate}
    \item The first is a monotonicity statement,~\cref{prop:monotonicity}, which shows that the minimal eigenvalue of $-\calL_{\loc}\big|_{\calA^{(\ell)}}$ monotonically increases as a function of $\ell\geq 1$. This will require a suitable choice of basis of $\comm(\Sn)$, the intertwining property of $\calL_{\loc}$ (and the Pauli twirl $\calT$), and a smattering of facts about the Clebsch-Gordan coefficients.
    \item The second is the base case $\ell=1$,~\cref{{prop:CL-A1}}, which explicitly computes the Dirichlet form associated to the restricted dynamics $-\calL_{\loc}\big|_{\calA^{(1)}}$, resulting in a system-size independent constant lower bound on its spectrum at high temperatures. By monotonicity, this grants the same lower bound for all $\ell\geq 1$.
\end{enumerate}

We proceed to describing the high-level argument, whose full proof is the content of~\cref{section:monotonicity-and-al-gap}.

\paragraph{A finer decomposition of \texorpdfstring{$\comm(\Sn)$}{comm(Sn)}.} In a manner similar to presentations of quantum angular momentum in physics texts, we simultaneously diagonalize three commuting self-adjoint superoperators to construct a basis for $\comm(\Sn)$. 
\begin{definition}
    [Spherical Tensor Operators]\label{def:spherical tensor operators} There exists a KMS orthogonal basis of $\comm(\Sn)$ given by the collection of operators $T_{s,\ell,q}$, which are joint eigenvectors of the superoperators $\ad_{S_{\tot}^Z}$, $\CL_{\su(2)}$, and left-multiplication by $\S^2_\tot$:
    \begin{align}
        [S_{\tot}^Z, T_{s,\ell,q}] &= q\cdot  T_{s,\ell,q}, \\ \quad \S_{\tot}^2\cdot  T_{s,\ell,q}  &= T_{s,\ell,q}\cdot \S_{\tot}^2 = s(s+1)\cdot  T_{s,\ell,q}, \\ -2\CL_{\su(2)}( T_{s,\ell,q}) &= \ell(\ell+1) T_{s, \ell, m}.
    \end{align}
    In particular, we denote $\calA^{(\ell,q)}:=\mathrm{span}\{T_{s, \ell, q}:s\in \calS\text{ with }s\geq \ell/2\}$.
\end{definition}
This basis reflects a continued decomposition of $\comm(\Sn)$ starting from the isotypic decomposition in \cref{eq:isotypic decomp comm Sn}:
\begin{equation} \label{eq:continued decomp A ell q}
    \comm(\Sn) = \bigoplus_{\ell=0}^{n} \calA^{(\ell)}, \qquad \calA^{(\ell)}=\bigoplus_{q=-\ell}^{\ell} \calA^{(\ell,q)}.
\end{equation}

\begin{observation}\label{obs:L preserves l and q eigenvalues}
Since $\calL_{\loc}$ is an intertwiner (c.f.~\cref{thm:davies is intertwiner}), it commutes with the superoperators $-2\calL_{\loc}$ and $\ad_{S_{\tot}^Z}$ and so preserves their respective $\ell(\ell+1)$ and $q$ eigenspaces. Thus each $\calA^{(\ell,q)}$ is an invariant subspace of $\calL_{\loc}$:
\begin{equation}
    \calL_{\loc}: \calA^{(\ell,q)}\to \calA^{(\ell,q)}.
\end{equation} 
\end{observation} In other words, $\calL_{\loc}$ may only change the $s$ variable.
There are two important examples to highlight.
\begin{example}
    When $\ell=0$, the only valid $q$ is $0$ and so we have $\calA^{(0,0)}=\calA^{(0)}$. The generator $\calL_{\loc}\vert_{\calA^{(0)}}$ is the coarse-grained Pauli master equation generator $L_{s, s'}$ studied in~\cref{sec:pauli_master_equation}.
\end{example}
A particularly insightful example is provided by the space of ``highest-spin'' operators $\calA^{(\ell, \ell)}.$ In representation theoretic language, this collection supplies an orthogonal (but not orthonormal) basis for the highest weight space in the representation $\comm(\Sn)$ of $\SU(2)$.
\begin{example}
    The collection of operators $\{ T_{s,\ell,\ell}\}$ can be written as \begin{equation}
    \calA^{(\ell, \ell)}=\mathrm{span}\{T_{s,\ell,\ell}\} = \mathrm{span} \Bigl\{ (S_{\tot}^+)^\ell \Pi_{n/2}, \; (S_{\tot}^+)^\ell \Pi_{n/2-1}, \; \dots, \; (S_{\tot}^+)^\ell \Pi_{s_{\min}(\ell)} \Bigr\}, \end{equation}
     where $S_{\tot}^+ = S_{\tot}^X+iS_{\tot}^Y$ and $s_{\min}(\ell) = \min \{s\in \calS: s\geq \ell/2\}$.
\end{example}
 
Our upcoming computations will require explicit matrix elements for the $T_{s,\ell,q}$ operators in an orthonormal basis of $\calH$ compatible with Schur-Weyl duality (c.f.~\cref{thm:qubit schur-weyl}). One approach starts by leveraging the full version of the Wigner-Eckart theorem~\cite{sakurai_modern_2020} for tensor operators of rank $\ell$ to constrain these matrix elements. For the presentation here we instead opt for a constructive bare-hands approach, where we write down the proposed matrix elements of each $T_{s,\ell,q}$ and then afterwards demonstrate that they constitute the desired eigenbasis.
These matrix elements are given as functions of Clebsch-Gordan coefficients, which opens the door to a wealth of tools from the theory of quantum angular momentum. The very same Clebsch-Gordan selection rules which ensured the coarse-grained Pauli master equation generator was tridiagonal (c.f.~\cref{prop:L is tridiagonal}) likewise ensure that the restriction $\CL_\loc|_{\calA^{(\ell,q)}}$ is tridiagonal.

\paragraph{The monotonicity statement.}
\cref{obs:L preserves l and q eigenvalues} has now reduced the problem of lower bounding the minimal eigenvalue of each matrix $-\calL_{\loc}\vert_{\calA^{(\ell)}}$ to lower bounding the minimal eigenvalue of each $-\calL_{\loc}\vert_{\calA^{(\ell,q)}}$, effectively block-diagonalizing the problem. But $\SU(2)$ symmetry reveals that these matrix blocks are highly redundant:~\cref{lem:Llq is tridiagonal} shows that 
\begin{equation}
    -\CL_\loc(T_{s, \ell, q}) = \rd_{s, s}^{(\ell)} \cdot  T_{s, \ell, q} + \rd_{s, s+1 }^{(\ell)}\cdot  T_{s+1, \ell, q} + \rd_{s, s-1}^{(\ell)} \cdot T_{s-1, \ell, q}, 
\end{equation} where crucially the coefficients $\rd_{s, s'}^{(\ell)}$ are \textit{independent} of $q$. Thus, for each fixed $\ell$, it suffices to choose a single $q'$ and lower bound the eigenvalues of $-\calL_{\loc}\vert_{\calA^{(\ell,q')}}$. 

The goal from here is to show that the minimal eigenvalue of these restricted maps is a monotonically increasing function in $\ell$. The fact that the matrix elements of the $T_{s,\ell,q}$ are functions of Clebsch-Gordan coefficients (or equivalently, Wigner 3j coefficients) ultimately powers the analysis, and heavy application of symmetry properties and contraction identities allows one to find that the $\rd_{s, s'}^{(\ell)}$ are in fact rather clean functions of certain Wigner 6j coefficients. As a consequence, one sees that the matrix elements $\rd_{s, s'}^{(\ell)}$ are monotonically increasing in $\ell$, which since every entry is real and the off-diagonal elements are non-positive, implies that the minimal eigenvalue of $-\calL_{\loc}\vert_{\calA^{(\ell,q)}}$ is a monotone increasing function of $\ell$, as we record in~\cref{prop:monotonicity}.

\begin{proposition}
    [Monotonicity of $-\calL_{\loc}\big|_{\calA^{(\ell)}}$]\label{prop:monotonicity}
    The minimal eigenvalue of $-\calL_{\loc}\big|_{\calA^{(\ell)}}$ is a monotonically increasing function of $\ell$. That is, 
    \begin{equation}
    \min\mathrm{spec}\paran{-\calL_{\loc} \Big\vert_{\calA^{(\ell')}}} \geq \min\mathrm{spec}\paran{-\calL_{\loc} \Big\vert_{\calA^{(\ell)}}} \qquad \text{if } \ell' \geq \ell\geq 1.
\end{equation}
\end{proposition} 

\paragraph{The base of the tower.} In light of the monotonicity statement, it only remains to analyze the spectra of $\CL_\loc$ at the bottom of the $\ell$ hierarchy, in the space $\calA^{(1)}$. As a consequence of $q$-independence, we may choose any $q$ we please; we opt for $q=0$, where the $T_{s,1,0}$ operators are scalar multiples of perhaps more familiar faces:
\begin{equation}
    \calA^{(1,0)} =\mathrm{span}\{S^Z_\tot\cdot \Pi_s:s\in\calS, s>0\}.
\end{equation}
In this unnormalized orthogonal basis we may explicitly compute the Dirichlet form and variance of such observables to control the minimal eigenvalue. 

\begin{proposition}
    [A lower bound on $-\CL_\loc|_{\calA^{(1)}}$]\label{prop:CL-A1} There exists a universal constant $c\in \BR^{+}$, such that for any $\beta<2$ the generator $\CL_\loc$ satisfies:
    \begin{equation}
        -\CL_\loc|_{\calA^{(1)}} \geq c\cdot (2-\beta)\cdot \idty|_{\calA^{(1)}}.
    \end{equation}
\end{proposition}

Earlier, the generator for the coarse-grained Pauli master equation in \cref{sec:pauli_master_equation} was tridiagonal, allowing one to interpret it as a generator of a birth-death process. While the matrix $-\calL_{\loc}\vert_{\calA^{(1,0)}}$ is not a Markov chain generator, the proof of \cref{prop:CL-A1} hinges upon tools used to analyze such chains and is inspired by Hardy's inequality \cite{miclo_example_1999}. We refer the reader to \cref{section:chain-T10} for a complete proof. 

Put together, \cref{prop:CL-A1} and \cref{prop:monotonicity} conclude the proof of the high-temperature component of \cref{thm:gap-Al}.

\subsection{The Order Parameter Witnesses the Gap}\label{sec:order parameter witnesses the gap}
We dedicate this section to a proof that our lower bounds on the spectral gap of $\CL_\loc$ are tight, at least away from the critical point. To do so, we compute the Raleigh quotient of the total magnetization order parameter $S_{\tot}^Z$, which by the variational characterization of the spectral gap (\cref{definition:gap}), provides an upper bound. It should be noted that $S_{\tot}^Z\in \calA^{(1)}$, and so in this sense the observables in $\calA^{(1)}$ are the slowest to mix among all observables in $\calB(\calH)$. The main result of this section is the following theorem. 

\begin{theorem}[Spectral gap upper bounds]\label{thm:total mag witnesses the gap} Fix any $\beta\neq2$. Then, the spectral gap of $\CL_\loc$ satisfies the following upper bound:
\begin{alignat}{2}
\gap(\lind_\loc) \leq \frac{\calE[S_{\tot}^Z]}{\mathsf{Var}_\rho[S_{\tot}^Z]} \leq  \begin{cases}
    c & \text{ if }\beta <2,\\
  n^{-1}\cdot \frac{c\cdot \beta^2}{(\beta-2)^2} & \text{ if }\beta >2,
\end{cases}
\end{alignat}
where $c\in \BR^{+}$ is a universal constant.
\end{theorem}

\noindent The proof of \cref{thm:total mag witnesses the gap} is an immediate consequence of~\cref{lemma:Sz-dirichlet,lemma:sz-var}, which respectively compute the Dirichlet form and variance of $S^Z_\tot.$

\begin{lemma}
    [Dirichlet form of $S^Z_\tot$]\label{lemma:Sz-dirichlet} For any $\beta\geq 0$: 
    \begin{equation}
        \calE[S_\tot^Z]=\langle S_\tot^Z, -\mathcal{L}_\loc(S_\tot^Z) \rangle_{\rho} \leq  n.
    \end{equation}
\end{lemma}

\begin{proof}
Since $S_\tot^Z$ commutes with $H$ and thus its spectral projectors, it holds that $\commt{S_i^\alpha(\omega)}{S_\tot^Z} = \commt{S_i^\alpha}{S_\tot^Z}(\omega)$. Using the divergence form of the Dirichlet form (\cref{lem:dirichlet_divergence}), we can write
\begin{align}
\calE[S_\tot^Z] &= \frac{1}{2} \cdot \sum_{i=1}^n \sum_{\alpha=X,Y,Z} \sum_{\omega} h_\omega \cdot \norm{\commt{S_i^\alpha}{S_\tot^Z}(\omega)}^2_\rho = 2\cdot  \sum_{i=1}^n \sum_{\omega}  h_\omega\cdot  \left( \norm{ S_i^Y(\omega)}^2_\rho + \norm{ S_i^X(\omega)}^2_\rho \right). \label{eq:Dirichlet_S_tot_Z_expanded}
\end{align}
To proceed, we focus on the contribution of a single site $i\in [n]$. Leveraging the simple upper bound $h_\omega\leq 1$ and expanding the KMS inner product yields:
\begin{align}
     \sum_\omega h_\omega\norm{ S_i^X(\omega)}^2_\rho &\leq \sum_{s, \omega} \tr[\rho^{1/2} \Pi_s S_i^X \Pi_{s+\omega}\rho^{1/2}S_i^X] \\&= \sum_s \tr[\rho^{1/2} \Pi_s S_i^X  \rho^{1/2}S_i^X] =\tr[\rho^{1/2}  S_i^X \rho^{1/2}S_i^X], \label{eq:sum-over-omega}
\end{align}
where we have twice applied a resolution of identity, $\sum_s \Pi_s=\idty$ and $\sum_\omega \Pi_{s+\omega}=\idty$. Finally, H\"older's inequality implies $\cref{eq:sum-over-omega} \leq \|\rho^{1/2}\|_2^2\cdot \|S_i^X\|^2\leq 4^{-1}$. Incorporating the sum over $i\in [n]$, we arrive at $\cref{eq:Dirichlet_S_tot_Z_expanded} \leq n$.
\end{proof}

Next, we turn to the denominator. To derive a lower bound on the denominator, we will require certain tail bounds on the spin number $s$ under a sample from the Gibbs measure $\pi(s)=\tr[\Pi_s\rho]$.
\begin{lemma}
    [Variance of $S^Z_\tot$]\label{lemma:sz-var} The variance of $S^Z_\tot$ satisfies the following lower bound:
    \begin{align}
        \mathsf{Var}[S_{\tot}^Z]\geq \begin{cases}
            c\cdot \big(2^{-1}-\beta^{-1}\big)^2\cdot n^2 & \text{ if } \beta > 2\\
            c\cdot n & \text{ if } \beta <  2
        \end{cases}
    \end{align}
    where $c\in \BR^{+}$ is a universal constant. 
\end{lemma}

\begin{proof} First observe $\mathsf{Var}[S_\tot^Z] = \|S_\tot^Z\|_\rho^2$ since the total magnetization is centered, $\tr[\rho S_\tot^Z] = 0$. Then, since $S^Z_\tot$ commutes with $H$, it suffices to evaluate the square:
    \begin{align}
\Tr[\rho (S_\tot^Z)^2] = \frac{1}{3} \Tr[\rho\,  \S^2_\tot] = \frac{1}{3} \EV_\pi[s(s+1)] \ge \frac{1}{6} \EV_\pi[s^2],
\end{align}
where we used \cref{def:heisenberg model} and the $\mathsf{SU}(2)$ symmetry of the model to relate the expected value of $(S_\tot^Z)^2$ to that of $\S_\tot^2$, and in turn, the spin number $s$. We denote by $\EV_\pi$ the expectation under the Gibbs measure. \\

\noindent \textbf{At high temperatures} ($\beta<2$), we observe that the distribution $\pi(s)$ is an (unnormalized) product of a binomial and a monotonically increasing function of $s$ (see \cref{thm:qubit schur-weyl} or \cref{eq:formula-for-dim}):
    \begin{equation}
        \pi(s)\propto \dim (\Pi_s) \cdot e^{\beta s(s+1)/n} = \binom{n}{\frac{n}{2}-s}\cdot \underbrace{\bigg(\frac{(2s+1)^2}{n/2+s+1}\cdot e^{\beta s(s+1)/n}\bigg)}_{h(s), \text{ monotonically increasing for }s\in \mathcal{S}}.
    \end{equation}
    Let $h(s)$ be so that $\pi(s) = Z_\beta^{-1} \binom{n}{\frac{n}{2} - s} h(s)$ where we recall $Z_\beta$ is the partition function for $\pi(s)$.
    Consider the binomial random variable $p \sim \text{Binom}(n, \frac{1}{2})$ which has distribution $p(s) = Z_p^{-1} \binom{n}{\frac{n}{2} - s}$ where $Z_p$ is the partition function for $p$. 
    We have
    \[
    \EV_\pi[s^2] = \sum_{s \in \mathcal{S}} s^2 Z_\beta^{-1} h(s) \binom{n}{\frac{n}{2} - s} = \EV_p[ s^2 Z_p Z_\beta^{-1} h(s) ]
    \]
    It's easily checked that $h(s)$ monotone increasing in $s$ implies $\text{cov}(s^2, h(s)) \geq 0$ so recalling the covariance identity $\EV[X Y] = \EV[X] \,\EV[Y] + \text{cov}(X, Y)$ we conclude $\EV_\pi[s^2] \geq \EV_p[s^2] \, \EV_p[ Z_b Z_\beta^{-1} h(s) ] = \EV_p[ s^2] \, \EV_\pi[ 1 ]  = \Theta(n)$. \\

\noindent \textbf{At low temperatures} ($\beta>2$), we observe that $\pi(s)$ is monotonically increasing in $s$, if $0\leq s\leq c_\beta\cdot n:$
\begin{align}
    \frac{\pi({s+1})}{\pi(s)} 
    &=e^{2\beta (s+1)/n} \cdot \frac{n-2s}{n+2s+4} \cdot \frac{(2s+3)^2}{(2s+1)^2} \\
    &\geq e^{2\beta (s+1)/n} \cdot \frac{n-2s}{n+2s+4} \\
    &\geq \exp\bigg[\frac{2\beta (s+1)}{n} - \left( \frac{n + 2s + 4}{n-2s} -1 \right)\bigg] \\
    &= \exp\bigg[2(s+1)\bigg(\frac{\beta}{n}-\frac{2}{n-2s}\bigg)\bigg] \geq 1 \quad \text{ if }\quad s \leq c_\beta n:=  \bigg(\frac{1}{2}-\beta^{-1}\bigg)\cdot n.
\end{align}
In the second inequality above, we leveraged $1+x\leq e^x \Rightarrow (1+x)^{-1} \geq e^{-x}$.

Now set $s_* \in \mathcal{S}$ to be the largest spin so that $s_* < \tfrac{1}{2} c_\beta n$.
Since $\pi(s)$ is monotone increasing for $s \leq c_\beta n$, $\pi([0, s_* ]) \leq \pi([ s_* + 1, c_\beta n ])$ since we can pair up each term on the left hand side with one on the right hand side, i.e., $\pi(s_{\min}) \leq \pi(s_* + 1)$, $\pi(s_{\min} + 1) \leq \pi(s_* + 2)$, etc.
It follows that
\[
1 = \pi([0, s_*]) + \pi([s_* + 1, \tfrac{n}{2}]) \leq 2\,\pi([s_* + 1, \tfrac{n}{2}]) \quad \Rightarrow \quad \tfrac{1}{2} \leq \pi([s_* + 1, \tfrac{n}{2}])
\]
Consequently, $\EV_\pi[s^2] \geq s_*^2 \pi([s_* + 1, \tfrac{n}{2}]) = \Omega(c_\beta^2 n^2)$.
\end{proof}

\section*{Acknowledgments}

This work was partially supported by the Challenge Institute for Quantum Computation (CIQC) funded by National Science Foundation (NSF) through grant numbers OMA-2016245 (J.B., L.L.) and CCF-2420130 (J.B.), the Simons Targeted Grants in Mathematics and Physical Sciences on Moir\'e Materials Magic (K.D.S., L.L.), and the Simons Investigator in Mathematics award through Grant No. 825053 (L.L.). We thank Ehud Altman, Anthony Chen, Zhiyan Ding, Sergio Escobar, Marius Junge, Bruno Nachtergaele, Ojas Parekh, Pablo Sala, Nikhil Srivastava, James Sud, Monica Vazirani, Thuy-Duong Vuong, and Ruizhe Zhang for insightful discussions.

\printbibliography
\appendix

\crefalias{section}{appendix}

\newpage

We now present a brief overview of the contents of this appendix.

In~\cref{secapp:unitary reps}, we review some of the representation theory of $\SU(2)$ and $\Sn$, notably including explicit descriptions of the representations appearing in this work as well as the quadratic Casimir in its many forms. 
In~\cref{secapp:proving intertwiners}, we state and prove a formal version of the idea ``A Lindbladian inherits the symmetry of its Hamiltonian and its jump operators''.
In~\cref{sec:wigner_eckart}, we review the Wigner-Eckart theorem and explain the tight constraints it lays on the dynamics of $\calL_{\loc}$, ultimately leading to the observation that the coarse-grained Pauli master equation describes a birth-death process.
In~\cref{sec:proof_gap_of_L}, we compute the spectral gap of this birth-death process.
In~\cref{secapp:the simulation argument}, we present the simulation arguments for the group mixers $\calL_{\su(2)}$ and $\calL_{\Sn}$.
Finally, in~\cref{section:monotonicity-and-al-gap}, we present an analysis of the dynamics of $\calL_{\loc}$ restricted to permutation-invariant observables which provides a sharp lower bound for the gap at high temperatures.

\section{Unitary Representations of \texorpdfstring{$\SU(2)$}{SU(2)} and \texorpdfstring{$\Sn$}{Sn}} \label{secapp:unitary reps}

\subsection{Generalities on Compact Groups}
We very briefly recall some useful material on unitary representations of compact groups, which may be found in standard texts such as~\cite{hall_lie_2015}. For our purposes the two relevant examples are the compact group $\SU(2)$, the $2\times 2$ unitary matrices of determinant 1, and the finite group $\Sn$ of permutations on $n$ letters.

Let $\sfG$ be a group and let $V$ be a finite dimensional Hilbert space. A \emph{unitary} representation is the Hilbert space $V$ with a group homomorphism $R:\sfG\to \calB(V)$ mapping group elements to unitary operators, i.e. $R(g)^\dagger R(g)=\idty$ for all $g\in \sfG$.\footnote{As a linguistic note, authors differ on whether ``representation'' refers to the vector space $V$ or the homomorphism $R$. Ultimately both are fine, as it is the shared data $(R,V)$ which matters. We use the convention of calling $V$ the representation.} We call a subspace $W\subseteq V$ an \emph{invariant subspace} of a representation $V$ if $R(g)w\in W$ for all $w\in W$ and $g\in \sfG$. We call a $V$ an \emph{irreducible representation} or \emph{irrep} if the only invariant subspaces of $V$ are the trivial subspaces $V$ and $\{0\}$.  

\begin{lemma}
    Let $R:\sfG\to \calB(V)$ be a unitary representation, and let $W$ be an invariant subspace of $V$. Then the orthogonal complement $W^\perp$ is also an invariant subspace of this representation.
\end{lemma}
\begin{proof}
    If $v\in W^{\perp}$ and $w\in W$, then for any $g\in \sfG$ we have $\inprod{R(g)v,w} = \inprod{v,R(g)^\dagger w} = \inprod{v,R(g^{-1})w}$. But since $R(g^{-1})w\in W$, this is $0$, so $R(g)v\in W^{\perp}$.
\end{proof}

By iterating this argument, one arrives at the following fundamental theorem of complete reducibility. A bit of notation will prove useful: we say two representations $V_1,V_2$ are \emph{equivalent} if there is an invertible linear map $T:V_1\to V_2$ such that $T R_1(g) = R_2(g)T$ for all $g\in \sfG$. Then we may meaningfully define the collection $\widehat{\sfG} := \{\text{equivalence classes of irreps of }\sfG\}. $

\begin{theorem}[Complete Reducibility] \label{thm:complete reducibility}

Every finite dimensional unitary representation $V$ of a group $\sfG$ is completely reducible, i.e. it may be written as an orthogonal direct sum of irreps
\begin{equation}
    V \cong \bigoplus_{s\in\calS} V_s^{\oplus m_s},
\end{equation} where $\calS\subseteq \widehat{\sfG}$ is a collection of labels of irreps and $m_s$ is the multiplicity of the irrep $V_s$.
\end{theorem}

We make a simple but useful remark which plays a central role in the main text. It is often difficult to find the change of basis which manifests this orthogonal decomposition, but knowledge about the irrep content of a subspace still provides useful information. 
\begin{corollary}\label{cor:consequence of complete reducibility}
    Suppose that $W$ is an invariant subspace of $V$ consisting of irreps of a single type, say $W = V_{s}^{\oplus m_s}$. Then the subspace $U = \bigoplus_{s'\neq s} V_{s'}^{\oplus m_{s'}}$ is contained in the orthogonal complement $U\subseteq W^{\perp}$.
\end{corollary}

We conclude with the truly invaluable Schur's lemma.
\begin{lemma}[Schur's lemma]\label{lem:schur}

Let $V_1, V_2$ be irreps of a group $\sfG$ with homomorphisms $R_{1,2}:\sfG\to \calB(V_{1,2})$ and let $T:V_1\to V_2$ be an \emph{intertwiner}, i.e. a linear map which commutes with the group action 
\begin{equation}
    T R_1(g) = R_2(g) T\, \qquad \text{for all } g\in \sfG.
\end{equation} If $V_1$ and $V_2$ are not equivalent representations, then $T=0$. Conversely, if $T:V_1\to V_1$, then $T = c\, \idty$ where $\idty$ is the identity operator on $V_1$ and $c\in \C$.
\end{lemma}

\subsection{The Representations in This Work}\label{secapp:reps in this work}
In this section we describe the primary representations of $\SU(2)$ and $\Sn$ which we encounter in the main text, as well as their commonly known names in the mathematics literature.

Before proceeding, we note that any unitary representation of $\SU(2)$ induces a Hermitian representation of $\su(2)$, given by a Lie algebra homomorphism $r:\su(2)\to \calB(V)$ which satisfies $r(S)^\dagger = r(S)$ for all $S\in \su(2)$. Since we are considering complex representations, we may freely work with the complexification of $\su(2)$, which in this case coincides with the complex vector space of $2\times 2$ traceless matrices spanned by the spin-1/2 Pauli matrices $S^X,S^Y,S^Z$. The Pauli matrices enjoy the following commutation relations: 
\begin{equation} \label{eq:pauli comm rel}
    [S^X,S^Y] = iS^Z, \quad [S^Y,S^Z] = iS^X, \quad [S^Z,S^X] = iS^Y. 
\end{equation} Since $r$ is a Lie algebra homomorphism, the representatives $r(S^\alpha)$ inherit these commutation relations. We also take a moment to define the raising and lowering operators $S^{\pm}:= S^X\pm i S^Y$.

The equivalence classes of irreps $\widehat{\SU}(2)$ and $\whatSn$ are well-understood, although in our case we need only describe the former.\footnote{The classification of irreps of $\SU(2)$ presented here appears in physics texts under the name quantum angular momentum~\cite{sakurai_modern_2020}, and it is a special case of the theorem of highest weight in representation theory~\cite{hall_lie_2015}.} The irreps of $\SU(2)$ are labeled uniquely by nonnegative half integers called \emph{spins}:
\begin{equation} \label{eq:SU2 irreps are half integers}
    \widehat{\SU}(2) = \left\{0, \frac{1}{2}, 1, \frac{3}{2}, 2, \dots\right\}.
\end{equation} In particular, if $W$ is an irrep of $\SU(2)$, then there is a unique $s\in\widehat{\SU}(2)$ for which $W\cong V_s$. The dimension of the spin-$s$ irrep $V_s$ is $2s+1$, and one may readily obtain an orthonormal basis for $V_s$ by diagonalizing the operator $r(S^Z)$. The eigenvalues of $r(S^Z)$ are simple and take the integer values $m=-s,-s+1,\dots, s$, with the eigenvector $\ket{s}$ being called the vector of highest spin. In this basis one may use commutation relations to compute the matrix elements of the representatives of raising and lowering operators:
\begin{equation} \label{eq:matrix elements in spin s irrep}
    r(S^Z) \ket{m} = m\ket{m} , \qquad  r(S^\pm) \ket{m} = \sqrt{(s\mp m)(s\pm m + 1)} \ket{m \pm 1}, 
\end{equation} where the edge cases should be interpreted as $r(S^+)\ket{s}=0$ and $r(S^-)\ket{-s}=0$.

We now proceed to the main players in our story.

\begin{description}
 \item[\bfseries The Trivial Representation:] 
 Every group $\sfG$ has a trivial representation, which is the vector space $\C$ with the group homomorphism $R:\sfG\to \calB(\C)$ given by $R(g)=\idty$ for all $g\in \sfG$. If $\sfG=\SU(2)$, the homomorphism $r:\su(2)\to \calB(\C)$ is given by $r(S) = 0$ for all $S\in \su(2)$, and this is the spin-0 irrep.
 \item[\bfseries The Single Qubit Representation of $\SU(2)$:] 
This is the defining representation $\C^2$ of $\SU(2)$, also known as the spin-1/2 irrep. The group homomorphism $R:\SU(2)\to \calB(\C^2)$ and Lie algebra homomorphism $r:\su(2)\to \calB(\C^2)$ are  \begin{equation} \label{def:single_qubit_rep}
     R(U) = 
     U, \qquad r(S) = S, \qquad \text{ for all }U\in \SU(2), \; S\in \su(2).
 \end{equation}
 \item[\bfseries The $n$-Qubit Tensor Representation of $\SU(2)$:] This is the $n^{th}$-tensor power $\calH = (\C^2)^{\otimes n}$ of the defining representation of $\SU(2)$. The group homomorphism $R:\SU(2)\to \calB(\calH)$ and Lie algebra homomorphism $r:\su(2)\to \calB(\calH)$ are 
 \begin{equation} \label{def:n_qubit_rep_su2}
     R(U) = 
     U^{\otimes n}, \qquad r(S) = S_{\tot} := \sum_{i=1}^n S_i\; ,\qquad \text{ for all }U\in \SU(2), \; S\in \su(2),
 \end{equation} where $S_i$ acts as the operator $S$ on qubit $i$ and identity on all other qubits.
 \item[\bfseries The $n$-Qubit Tensor Representation of $\Sn$:] This is the space $\calH = (\C^2)^{\otimes n}$ where the action of $\Sn$ is given by index permutation. The homomorphism maps a permutation $\sigma\in \Sn$ to the linear operator which acts on an arbitrary basis element $\ket{i_1, i_2,\dots , i_n}\in \calH$ as 
    \begin{equation} \label{def:n_qubit_rep_Sn}
        \sigma \ket{i_1, i_2,\dots , i_n} = \ket{i_{\sigma^{-1}(1)}, i_{\sigma^{-1}(2)}, \dots , i_{\sigma^{-1}(n)}},
    \end{equation} where we have abused notation and identified $\sigma$ with its representative linear operator. In this language, $\sigma = (i,j)$ represents the swap operator $\mathsf{SWAP}_{ij}$.
 \item[\bfseries The Observables Representation:] This is the space of operators $\calB(\calH)$ equipped with the adjoint action of either of the previous two representations, i.e. the maps $\Delta:\SU(2)\to \calB(\calB(\calH))$ and $\wt{\Delta}:\Sn\to \calB(\calB(\calH))$ are given by
    \begin{equation}\begin{split}
        \Delta_U (A) &:= U^{\otimes n} A (U^\dagger)^{\otimes n}, \quad
        \wt{\Delta}_\sigma(A) := \sigma A \sigma^{-1}
    \end{split}\end{equation} for all $A\in \calB(\calH)$, $U\in \SU(2)$, and $\sigma\in \Sn$. The group homomorphism $\Delta$ induces the Lie algebra homomorphism from $\su(2)$ to $\calB(\calB(\calH))$ defined by $S\mapsto \ad_{S_{\tot}}$, where $\ad_{S}(A) = [S,A]$ for all $A\in \calB(\calH)$. 
\end{description}

Essentially by construction the first three representations are unitary with respect to the inner product on $\calH$. The fourth is unitary w.r.t. the Hilbert-Schmidt inner product and the KMS inner product \cref{def:KMS DBC}, so long as the density matrix $\rho$ admits the appropriate corresponding symmetries.

\begin{observation} \label{obs:KMS unitarity from hamiltonian symmetry}
    Suppose $[\rho, U^{\otimes n}]=0$ for all $U\in \SU(2)$ and likewise $[\rho,\sigma]=0$ for all $\sigma\in \Sn$. Then the representations on observables $\calB(\calH)$ given by $\Delta,\wt{\Delta}$ are KMS unitary, i.e. for all $A,B\in \calB(\calH)$,
    \begin{equation}\begin{split}   \inprod{\Delta_U(A),\Delta_U(B)}_\rho &= \inprod{A,B}_\rho \qquad \text{for all }U\in \SU(2)
    \\ \inprod{\wt{\Delta}_\sigma(A),\wt{\Delta}_\sigma(B)}_\rho &= \inprod{A,B}_\rho \qquad \text{for all }\sigma\in \Sn.
    \end{split}\end{equation}
\end{observation}

\subsection{The Quadratic Casimir} \label{secapp:quadratic casimir}
In this section we recall the quadratic Casimir for $\su(2)$ and its useful properties. Because of its central role in this work, we write its explicit form for each of the representations appearing in~\cref{secapp:reps in this work}. Generalities for quadratic Casimirs of semisimple Lie algebras may be found in~\cite[Ch. 10]{hall_lie_2015}, and specifics for $\su(2)$ in the spin language may be found in~\cite[Ch. A.3]{tasaki_physics_2020}.

\begin{definition}
Take a Hermitian representation $V$ of $\su(2)$ with Lie algebra homomorphism $r:\su(2)\to \calB(V)$.
Define the Casimir element $C_r \in \calB(V)$ by
\begin{equation}
    C_r = (r(S^X))^2 + (r(S^Y))^2 + (r(S^Z))^2 . 
\end{equation} 
\end{definition}
One can readily show using the Pauli commutation relations~\cref{eq:pauli comm rel} that the Casimir element commutes with every representative:  $ [C_r, r(S)] = 0$ for all $S\in \su(2).$ This further means that $C_r$ commutes with every $e^{ir(S)}= R(e^{iS}) \in R(\SU(2))$, so by Schur's~\cref{lem:schur}, this immediately implies that if $V_s$ is an irrep of spin $s\in \widehat{\SU}(2)$, then $C_r = c\idty_s$ where $c$ is some complex constant. One may use the commutation relation of the Paulis to compute this constant precisely, see e.g.~\cite{tasaki_physics_2020}: 
\begin{equation} \label{eq:casimir is s(s+1)}
    C_r = s(s+1)\idty_s \qquad \text{if } (r, V_s) \text{ is the spin-$s$ irrep.}
\end{equation} To make contact with the theory of quantum angular momentum, we should record the correspondence between the Casimir element and the squared angular momentum operator: if $r$ is the tensor representation on $n$ qubits $V=(\C^2)^{\otimes n}$, i.e. $r(S^{\alpha}) = S_{\tot}^\alpha$, then the quadratic Casimir element is exactly the squared angular momentum operator
\begin{equation}\label{eq:S_tot_def}
    \mathbf{S}_{\tot}^2 := C_r = (S_{\tot}^X)^2 + (S_{\tot}^Y)^2 + (S_{\tot}^Z)^2.
\end{equation}

We now revisit the representations from~\cref{secapp:reps in this work} and compute their respective Casimirs.

\subsubsection{Casimir on a Single Qubit}\label{ex:1-site casimir}
Take the spin-$1/2$ irrep $V_{1/2} = \C^2$, i.e. $r(S^{\alpha}) = S^{\alpha}$ for each $\alpha\in \{X,Y,Z\}$. Then the Casimir element $C_1 =: \mathbf{S}^2$ here is given by  
    $\mathbf{S}^2 = (S^X)^2 + (S^Y)^2 + (S^Z)^2 = \frac{3}{4} \idty$.

\subsubsection{Casimir on Two Qubits}
Here, the representation $V$ is the tensor representation $V= V_{1/2}\otimes V_{1/2}\cong (\C^2)^{\otimes 2}$, and we have $r(S^\alpha) = S_{\tot}^{\alpha} = S_1^{\alpha} + S_2^{\alpha}$, where $\alpha\in \{X,Y,Z\}$. 
The quadratic Casimir $C_2 = \mathbf{S}_{\tot}^2$ is then
\begin{align}
    \mathbf{S}_{\tot}^2 = (S_1^X + S_2^X)^2 + (S_1^Y + S_2^Y)^2 + (S_1^Z + S_2^Z)^2 = 2 \paran{\frac{3}{4}\idty} + 2\mathbf{S}_1\cdot \mathbf{S}_2.
\end{align} 
In order to diagonalize $\mathbf{S}_{\tot}^2$, we must decompose the reducible $V$ into irreps. One can certainly use the sledgehammer~\cref{thm:qubit schur-weyl}, but in this case a quick Clebsch-Gordan decomposition yields $V\cong V_0 \oplus V_1$. 
$V_0$ is spanned by the vector $\ket{01}-\ket{10}$, and $V_1$ has the basis $\{\ket{00}, \ket{01} + \ket{10}, \ket{11}\}$. Let $\Pi_s$ denote the orthogonal projections onto the irreps of spin $s$. Then, by applying~\cref{eq:casimir is s(s+1)}, we obtain that 
\begin{equation}
    \mathbf{S}_{\tot}^2 = (1)(1+1)\Pi_1 + 0(0+1)\Pi_0 = 2\Pi_1.
\end{equation}

\subsubsection{Casimir on \texorpdfstring{$n$}{n} Qubits \texorpdfstring{and Proof of~\cref{prop:diagonalization of H}}{ }}\label{secapp:casimir on n qubits}
Here, the representation $\calH$ is the tensor representation $\calH\cong V_{1/2}^{\otimes n}$. As a vector space, $V\cong (\C^2)^{\otimes n}$, and the map is given by $r(S^\alpha) = S_{\tot}^\alpha$, $\alpha=X,Y,Z$.
The Casimir element $C_n\in \calB(\calH)$ on $n$ sites is exactly:
\begin{equation} \label{eq:n-site casimir}
    \mathbf{S}_{\tot}^2 := (S_{\tot}^X)^2 + (S_{\tot}^Y)^2 + (S_{\tot}^Z)^2 =  \sum_{i=1}^n \mathbf{S}_i^2 + \sum_{i\neq j} \mathbf{S}_i\cdot\mathbf{S}_j 
    = n\frac{3}{4}\idty -2n H ,
\end{equation} just as in the two qubit example, where $H$ is the Heisenberg model~\cref{def:heisenberg model}. 
To diagonalize the Casimir $\mathbf{S}_{\tot}^2$ we use the irrep decomposition $\calH = \oplus_{s\in \calS} V_s\boxtimes W_{Q(s)}$ afforded by \cref{thm:qubit schur-weyl}. Let $\Pi_s$ be the orthogonal projection onto a fixed spin subspace $V_s\boxtimes W_{Q(s)}$. Note that the Casimir commutes with the action of $\Sn$ and so acts as $\mathbf{S}_{\tot}^2\otimes \idty$ on $V_s\boxtimes W_{Q(s)}$. Applying~\cref{eq:casimir is s(s+1)} and using that $V_s$ is an irrep of spin-$s$,
\begin{equation} \label{eq:total spin squared diagonalization}
    \mathbf{S}_{\tot}^2 = \sum_{s\in \calS} s(s+1)\Pi_s \quad \Rightarrow \quad  H = \frac{3}{8}\idty - \frac{1}{2n}\sum_{s\in \calS} s(s+1)\Pi_s. 
\end{equation}
which is the first statement of~\cref{prop:diagonalization of H}. 

\subsubsection{Casimir on Observables on \texorpdfstring{$n$}{n} Qubits\texorpdfstring{ and Proof of~\cref{prop:su2 mixer spectrum}}{ }} \label{secapp:Casimir on B(H)}
For this final Casimir element we take the representation $V=\calB(\calH)$. The map $r:\su(2)\to \calB(V)$ is given by
$r(S^{\alpha}) = \ad_{S_{\tot}^\alpha}$, and the Casimir element $C_{\mathrm{obs}}$ then is 
\begin{equation}
    C_{\mathrm{obs}} = \ad_{S_{\tot}^X}^2 + \ad_{S_{\tot}^Y}^2+\ad_{S_{\tot}^Z}^2.
\end{equation} Using the Clebsch-Gordan decomposition and that $\calB(\calH)\cong \calH\otimes \calH^*$, one observes that only spins $\wt{\calS} := \{0,1,\dots, n\}\subseteq \widehat{\SU}(2)$ arise in the irrep decomposition of $\calB(\calH)$, and we may single out the kernel of $C_{\mathrm{obs}}$, the spin-$0$ subspace:
\begin{equation} \label{app-eq:decomp of B(H)}
    \calB(\calH) = \bigoplus_{\ell \in \wt{\calS}} V_\ell^{\oplus \wt{m}_\ell} = \ker(C_{\mathrm{obs}}) \oplus \bigoplus_{\ell \in \wt{\calS}:\,  \ell \geq 1} V_\ell^{\oplus \wt{m}_\ell},
\end{equation} where $\wt{m}_\ell$ is the multiplicity of the spin-$\ell$ irrep $V_\ell$. Then, using~\cref{eq:casimir is s(s+1)} to compute the eigenvalues yet again, since $\ell(\ell+1)$ is a monotone increasing function of $\ell$ we get that this quantity is minimized by choosing the smallest nonzero spin $\ell=1$:
\begin{equation}
    \gap(C_{\mathrm{obs}}) = 2.
\end{equation}Finally, the spin-0 subspace is the collection of trivial irreps of $\su(2)$, where we recall that a vector $A\in \calB(\calH)$ is in the trivial representation if $[S_{\tot}^{\alpha}, A] = 0$ for all $\alpha = X,Y,Z$. Equivalently, this is the collection of trivial irreps of $\SU(2)$, i.e. $\comm(\SU(2))$.
We arrive at a completed proof of~\cref{prop:su2 mixer spectrum} once we observe that \begin{equation} C_{\mathrm{obs}} = -2 \calL_{\su(2)}.  \end{equation}

\section{The Davies Generator Inherits Symmetry}\label{secapp:proving intertwiners}
In this section we prove~\cref{thm:davies is intertwiner}, which we restate here for convenience.

\begin{theorem}[$\lind_\loc$ is an intertwiner]
    The Davies generator $\calL_{\loc}$ (c.f. \cref{def:davies_generator}) is an intertwiner for the action of $\SU(2)\times \Sn$. That is, for any observable $X\in\calB(\calH)$ and $(U,\sigma)\in \SU(2)\times \Sn$, we have
    \begin{align}
        \calL_{\loc}((U^{\otimes n} \sigma) X(U^{\otimes n} \sigma)^{-1} ) = (U^{\otimes n} \sigma) \calL_{\loc}(X) (U^{\otimes n} \sigma)^{-1}
    \end{align} 
\end{theorem}

We start with following proposition, which ensures that if the Hamiltonian $H$ and the collection of jump operators $\calJ$ both possess a symmetry $\sfG$, then the Davies generator $\calL_\calJ$ inherits this same symmetry.
We need a definition to state this properly: given a collection of jump operators $\{A^a\}_{a\in \calJ}$, we may define a completely-positive map $\calT_\calJ$ by
\begin{equation}
\calT_\calJ(X) := \sum_{a\in \calJ} (A^a)^\dagger X A^a.
\end{equation}

\begin{proposition}[Lindbladians Inherit Symmetries] \label{prop:Lindbladians inherit symmetries}
    Let $\{A^a\}_{a\in \calJ}$ be a collection of jump operators and let $\calH$ be a unitary representation of a group $\sfG$ with homomorphism $R:\sfG\to \calB(\calH)$. Suppose that for all $g\in \sfG$ and $X\in \calB(\calH)$ \begin{equation}
    [H,R_g]= 0 \qquad \text{ and } \qquad \calT_\calJ(R_g X R_g^{-1}) = R_g \calT_\calJ (X) R_g^{-1} . 
    \end{equation} Then $\calL_\calJ$ is an intertwiner for the adjoint action of this representation, i.e.
    \begin{equation}
        \calL_{\calJ}(R_g X R_g^{-1}) = R_g \calL_{\calJ}(X) R_g^{-1} \qquad \text{ for all } g\in \sfG.
    \end{equation}
\end{proposition}
\begin{proof}
Throughout the proof we will repeatedly use the equivalence, true for any linear map $\CF$, 
\begin{equation} \label{eq:equivalence of intertwiner definitions}
    R_g \CF(R_g^{-1} X R_g) R_g^{-1} = \CF(X) \iff \CF( R_g X R_g^{-1}) = R_g  \CF( X )R_g^{-1} , 
\end{equation} where the equivalence may be seen by simply noting that the map $X\mapsto R_g X R_g^{-1}$ is invertible.

Now, observe that since $[H,R_g]=0$, we have $[\Pi_\lambda,R_g]=0$ for every eigenvalue $\lambda \in \mathrm{spec}(H)$. We may then use an expanded form of the Davies generator $\calL_\calJ$ (c.f. ~\cref{def:davies_generator}), obtained by writing out the jump operators in the frequency basis $S^a(\omega)$ appearing in~\cref{eq:bohr_frequency_decomposition}:
\begin{align} 
R_g \calL_{\calJ}( R_g^{-1} XR_g )R_g^{-1} &= \sum_{\omega \in B(H)}\gamma(\omega)\sum_{\substack{\lambda,\mu: \lambda - \mu = \omega \\ \lambda',\mu': \lambda'-\mu' = \omega}} R_g \Pi_{\mu'} \calT_\calJ(\Pi_{\lambda'} R_g^{-1} X R_g \Pi_{\lambda})\Pi_{\mu} R_g^{-1} \\
&\qquad \qquad \qquad \qquad \qquad 
- \frac{1}{2} R_g \{\Pi_{\mu'} \calT_\calJ(\Pi_{\lambda'} \Pi_{\lambda}) \Pi_\mu, R^{-1}_g X R_g \}R_g^{-1} \\
&= \sum_{\omega}\gamma(\omega)\sum_{\substack{\lambda,\mu: \lambda-\mu = \omega \\ \lambda',\mu': \lambda'-\mu' = \omega}} \Pi_{\mu'} R_g \calT_\calJ( R_g^{-1} \Pi_{\lambda'}X \Pi_{\lambda} R_g)R_g^{-1} \Pi_{\mu}  \\
&\qquad \qquad \qquad \qquad \qquad 
- \frac{1}{2}\{\Pi_{\mu'} R_g \calT_\calJ(R_{g}^{-1}\Pi_{\lambda'} \Pi_{\lambda}R_g) R_g^{-1} \Pi_\mu, X \},
\end{align} 
where in the second line we have inserted $R_g^{-1} R_g = \idty$ and used the algebraic manipulation
\begin{equation}
        R_g\{A, R_g^{-1}XR_g\} R_g^{-1} = R_gA R_g^{-1} XR_gR_g^{-1} + R_gR_g^{-1} XR_gAR_g^{-1} =  \{R_g A R_g^{-1}, X\}.
\end{equation} By using the assumption on $\calT_\calJ$ in equivalent form provided by~\cref{eq:equivalence of intertwiner definitions}, we see that 
\begin{align} \label{eq:expanded davies twirl}
R_g \calL_{\calJ}( R_g^{-1} XR_g )R_g^{-1} &= \sum_{\omega}\gamma(\omega)\sum_{\substack{\lambda,\mu: \lambda-\mu = \omega \\ \lambda',\mu': \lambda'-\mu' = \omega}} \Pi_{\mu'} \calT_\calJ(\Pi_{\lambda'}  X\Pi_{\lambda})\Pi_{\mu} 
- \frac{1}{2} \{\Pi_{\mu'} \calT_\calJ(\Pi_{\lambda'} \Pi_{\lambda}) \Pi_\mu, X\} \\
&= \calL_{\calJ}(X),
\end{align} concluding the proof.
\end{proof} 

\begin{remark} \label{rem:lind_inherit_symmetry_operator_fourier transform} This proposition similarly applies to other Lindbladians whose jumps are constructed via the operator Fourier transform, e.g. those appearing in \cite{chen2023efficient, ding2024efficient, jiang2024quantum}, since $[R_g,e^{iHt}]=0$ implies that \begin{equation}\label{eq:pulling group elements into bohr localized}
        R_g \widehat{A}(\omega) R_g^{-1} :=  \int_\R f(t) e^{-i\omega t} R_g e^{iHt} A e^{-iHt}R_g^{-1} \, dt = (\widehat{R_g A R_g^{-1}})(\omega),
    \end{equation}
    meaning the symmetry commutes with taking the operator Fourier transform.
\end{remark}

Now, the collection $\calJ$ of jump operators for $\calL_{\loc}$ is the collection of single-site Paulis. In this case the completely-positive map $\calT_\calJ$ appearing in~\cref{prop:Lindbladians inherit symmetries} is nothing more than the Pauli twirl $\calT:\calB(\calH)\to \calB(\calH)$ given by
\begin{equation}\label{def:twirl_def}
    \calT(X) = \sum_{i=1}^n \sum_{\alpha=X,Y,Z} S_i^\alpha X S_i^\alpha, \qquad X\in \calB(\calH),
\end{equation} where we have notationally suppressed the collection of jumps. We now show that the Pauli twirl is an intertwiner for the action of $\SU(2)\times \Sn$. At a high level, one expects $\SU(2)$ invariance because $\calT$ is a sum of depolarizing channels on each site, and $\Sn$ invariance because $\calT$ acts the same on each site.
 \begin{lemma}\label{lem:T is an intertwiner}
        $\calT$ is an intertwiner for the action of $\SU(2)\times \Sn$, i.e. for every $(U,\sigma)\in \SU(2)\times \Sn$ and $X\in \calB(\calH)$
        \begin{equation}
            \calT((U^{\otimes n} \sigma) X (U^{\otimes n} \sigma)^{-1}) = (U^{\otimes n} \sigma) \calT(X) (U^{\otimes n} \sigma)^{-1}.
        \end{equation}
    \end{lemma}
    \begin{proof}
        Since the actions of $\SU(2)$ and $\Sn$ commute we may separately check for $U$ and for $\sigma$. Let us decompose $\calT$ into
        $\calT = \sum_{i=1}^{n} \calT_i$, where each $\calT_i(X) = \sum_{\alpha} S^\alpha_i X S^{\alpha}_i$. One recognizes that each $\calT_i$ is exactly a copy of the depolarizing channel $\wt{\calT}$ on qubit $i$ with a shift, since 
        \begin{equation} \label{eq:single site twirl is depolarizing channel}
            \wt{\calT}(B):= S^X B S^X + S^Y B S^Y + S^Z B S^Z = \frac{\Tr(B)}{2}\idty - \frac{1}{4} B \qquad \text{for all } B\in \calB(\C^2).
        \end{equation} From this expression the intertwining property becomes clear, and we have that each $\calT_i$ satisfies 
        \begin{equation}
            \calT_i(U^{\otimes n} X (U^{\otimes n})^{\dagger}) =  U^{\otimes n} \calT_i(X)(U^{\otimes n})^{\dagger},
        \end{equation}
        and likewise, so does their sum.    Next, for any permutation $\sigma$, we have
    \begin{equation}
        \sigma\calT(X)\sigma^{-1} =   \sum_i \sigma \calT_i(X)  \sigma^{-1} = \sum_{i} \calT_{\sigma(i)}( \sigma X \sigma^{-1}) = \sum_j \calT_j (\sigma X \sigma^{-1}) = \calT(\sigma X \sigma^{-1}) ,
    \end{equation} as desired.    
    \end{proof}

By combining~\cref{prop:Lindbladians inherit symmetries,lem:T is an intertwiner}, we have proven~\cref{thm:davies is intertwiner}.

\begin{observation}
    As a corollary of the proof, 
    \begin{enumerate}[label=(\roman*)]
        \item For each site $i=1,\dots,n$, the Lindbladian 
        $\calL_{S_i^X} + \calL_{S_i^Y} + \calL_{S_i^Z}$ is an intertwiner for $\SU(2)$.
        \item For any single-site jump operator $S\in \calB(\C^2)$, the Lindbladian $
        \sum_{i=1}^n \calL_{S_i}$
         is an intertwiner for $\Sn$.
    \end{enumerate} As a consequence, for any Hamiltonian $\wt{H}$ with $\sfG$ symmetry for $\sfG=\SU(2)$, $\sfG=\Sn$, or any subgroups thereof, the generator $\calL$ given by choosing single-site Pauli jumps is also an intertwiner. This includes other Heisenberg models for $\SU(2)$ and other mean field models for $\Sn$, as well as generators in the scope of~\cref{rem:lind_inherit_symmetry_operator_fourier transform}.
\end{observation}

\section{Wigner-Eckart and the Coarse-Grained Pauli Master Equation}\label{sec:wigner_eckart}
In this section we prove~\cref{prop:L is tridiagonal}, which shows that the coarse-grained Pauli master equation is a birth-death process by explicitly computing the transition rates of its Markovian generator $L$. Here, we make fundamental use of the Wigner-Eckart theorem, which gives us control over the matrix elements of single-site Pauli operators in the energy eigenbasis of the Heisenberg Hamiltonian.

After a brief recap of the Wigner-Eckart theorem in~\cref{sec:wigner eckart}, we dedicate \cref{section:L-tridiagonal} to a proof that the generator $L$ is tridiagonal, and thus defines a birth-death process. The entries $L_{s, s'}$ correspond to transition rates between energy eigenspaces $\Pi_{s}$ to $\Pi_{s'}$; however, since the jump operators in $\calL_{\loc}$ are single-site Paulis, the Wigner-Eckart theorem prevents transitions between distant spins $s$ and $s'$.
Second, to compute the non-zero entries, one determines the action of the Pauli twirl $\calT$ on the projectors $\Pi_s$. This is done in a slightly indirect way by recognizing that $\mathbf{S}_{\tot}^2 \Pi_s = s(s+1)\Pi_s$ and that we may readily calculate $\calT((\mathbf{S}_{\tot}^2)^k)$ for $k=0,1,2$.

\subsection{The Wigner-Eckart Theorem} \label{sec:wigner eckart}

We begin with an explicit choice of basis common in the physics literature, e.g.~\cite{sakurai_modern_2020}.
We start with the decomposition from~\cref{thm:qubit schur-weyl}, recalled below:
\begin{equation}
\calH = \bigoplus_{s\in \calS} V_s \boxtimes W_{Q(s)}, \quad \text{where} \quad \dim(V_s) = 2s+1 \text{ and } \dim(W_{Q(s)}) = m_s.
\end{equation} 
A convenient orthonormal basis for this space reflecting this decomposition is given by $\{ \ket{s,m,r} : s\in \calS, \; m\in [-s, s], \; r\in [m_s]\},$
which is a joint eigenbasis of the pair of commuting operators $\mathbf{S}_{\tot}^2$ from~\cref{eq:n-site casimir} and $S_{\tot}^Z$ from~\cref{eq:matrix elements in spin s irrep}:
\begin{equation} \label{eq:smr eigenbasis}
    \mathbf{S}_{\tot}^2\ket{s,m,r} = s(s+1) \ket{s,m,r}, \qquad S_{\tot}^Z \ket{s,m,r} = m \ket{s,m,r}.
\end{equation} 
Fix a site $i\in \{1,\dots, n\}$.
We recall the spin operator $\mathbf{S}_i$, which is the tuple of operators $\mathbf{S}_i = (S_i^{+1}, S_i^0, S_i^{-1})$:
\begin{align}\label{eq:vector_operator_def}
    S_i^{+1} := -\frac{1}{\sqrt{2}}\paran{S_{i}^X + iS_i^Y} , \quad S_i^0 := S_{i}^Z, \quad S_i^{-1} := \frac{1}{\sqrt{2}} \paran{S_{i}^X - iS_i^Y} ,
\end{align}
The celebrated Wigner-Eckart theorem then lays constraints on the matrix elements of this operator in the basis $\ket{s,m,r}$ of the tensor representation $\calH$ of $\SU(2)$. 
Given two spins $s_1,s_2$ and two magnetizations $m_1,m_2$, we write the Clebsch-Gordan coefficient for spin $s$ and magnetization $m$ as $C^{sm}_{s_1 m_1 s_2 m_2}$.

\begin{theorem}[Special case of Wigner-Eckart {\cite[Section 3.11]{sakurai_modern_2020}}] \label{thm:wigner eckart}
    
    The matrix elements of the vector operator $\mathbf{S}_i$ (c.f. \cref{eq:vector_operator_def}) satisfy
    \begin{align}
        \inprod{s',m',r'|S_i^p |s,m,r} = C^{s',m'}_{s,m,1,p} \inprod{s',r'|| \mathbf{S}_i || s,r} \qquad \text{for } p\in \{-1,0,+1\}, 
    \end{align}
    where the scalar $\inprod{s',r'|| \mathbf{S}_i || s,r}$ does not depend on $m, m'$, or $p$, and $C^{s',m'}_{s,m,1,q}$ is a Clebsch-Gordan coefficient. 
\end{theorem}
The scalars $\inprod{s',r'|| \mathbf{S}_i || s,r}$ in general depend upon $r$ and $r'$. We can quickly extract a useful implication from a selection rule afforded by the Clebsch-Gordan coefficients, namely $C^{s,m}_{s_1, m_1, s_2, m_2}=0$ unless 
    $\abs{s_1-s_2} \leq s \leq s_1+s_2$ (see e.g.~\cite{sakurai_modern_2020}). By \cref{thm:wigner eckart} we have that, for all $q\in \{-1,0,1\}$,
\begin{equation}
    \inprod{s',m',r'| S_i^q | s, m, r} = 0 \text{ if } s' \not\in \{s-1, s, s+1\},
\end{equation}
ensuring a single-site Pauli operator cannot change the total spin by more than 1 and thus forcing $L$ to be tridiagonal.

Before proceeding, we pause to record a fact which will be useful in~\cref{section:monotonicity-and-al-gap}. Namely, one may isolate a single qubit, say the last qubit, and write
\begin{equation}
    \calH = \paran{ \bigoplus_{z\in \calS_{n-1}} V_{z}^{\oplus m_{z}^{(n-1)}}}\otimes V_{1/2} = \bigoplus_{z\in \calS_{n-1}} (V_{z} \otimes V_{1/2})^{\oplus m_{z}^{(n-1)}}
\end{equation} where $\calS^{(n-1)}$ denotes the set of valid spins on $n-1$ qubits and $m_z^{(n-1)}$ denotes the multiplicity of the spin $z$ irrep in this decomposition. The Clebsch-Gordan decomposition for $z>0$ gives\footnote{When $z=0$, this decomposition is $V_0\otimes V_{1/2} \cong V_{1/2}$.} $V_{z}\otimes V_{1/2} \cong V_{z-1/2}\oplus V_{z+1/2}$, whence we may relate the $\ket{s,m,r}$ basis to the new tensor basis $\ket{z,\wt{m},\eta}\ket{\up}, \ket{z,\wt{m},\eta}\ket{\down}$.

\begin{fact}[Single-site $\ket{s, m, \eta}$ matrix elements]\label{fact:matrix-elements-raising}
Fix $n\ge 2$, and let $z=s-\tfrac12$. Consider one copy of the spin-$z$ irrep on the first $n-1$ qubits, labeled by $\eta$. Define
\begin{equation}
    a_m:=\sqrt{\frac{s+m}{2s}},
\qquad
b_m:=\sqrt{\frac{s-m}{2s}}.
\end{equation}
Then the coupled basis vectors in the $z\otimes \frac12 = s\oplus (s-1)$ decomposition are given by:
\begin{align}
    \ket{s,m,\eta}
&=
a_m\,\ket{z,m-\tfrac12,\eta}\ket{\uparrow}
+
b_m\,\ket{z,m+\tfrac12,\eta}\ket{\downarrow}, \\
\ket{s-1,m,\eta}
&=
b_m\,\ket{z,m-\tfrac12,\eta}\ket{\uparrow}
-
a_m\,\ket{z,m+\tfrac12,\eta}\ket{\downarrow}.
\end{align}
\end{fact}

\subsection{\texorpdfstring{$L$}{L} is Tridiagonal} 
\label{section:L-tridiagonal}
The main content of this section is that $L$ is the generator of a birth-death process.
\begin{lemma}[$L$ is tridiagonal]\label{lem:L_tridiagonal}
The entries of the generator $L$ (c.f. \cref{eq:pauli_master_generator}) satisfy
\begin{align}
    L_{s, s'} = 0 \quad \text{if} \quad s' \not\in \{s-1,s,s+1\}.
\end{align}
\end{lemma}

\begin{proof}[Proof of \cref{lem:L_tridiagonal}]
Using the Pauli twirl~\cref{def:twirl_def}, we can write (e.g. as in~\cref{eq:expanded davies twirl})
\begin{align}\label{eq:lind_on_Pi_expanded}
    \calL_{\loc}(\Pi_s) = \sum_{s'\in \calS} \gamma(E_{s'}-E_s) \, \Pi_{s'} \calT(\Pi_s)\Pi_{s'} - \sum_{s'\in \calS} \gamma(E_s-E_{s'}) \, \Pi_{s} \calT(\Pi_{s'}) \Pi_{s},
\end{align}
where $E_s=-s(s+1)/n$ denotes the energy of the spin-$s$ sector.
Since $\calT$ is an intertwiner by~\cref{lem:T is an intertwiner}, Schur's \cref{lem:schur} implies that $\calT(\calA^{(0)}) \subseteq \calA^{(0)}$, so for any $s \in \calS$,
\begin{align}\label{eq:calT_Pi_expansion}
\calT(\Pi_s) = \sum_{s'\in\calS} c_{s, s'} \Pi_{s'},
\end{align}
where $c_{s, s'}$ are appropriate constants. 
However, most of these transitions are not realized due to constraints on the Clebsch-Gordan coefficients. Since each single-site Pauli $S_i^{\alpha}$ may be written as a linear combination of $\{S_i^{+1},S_i^0, S_i^{-1}\}$, it follows from \cref{thm:wigner eckart} that 
\begin{align}\label{eq:S_i_clebsh_gordan_constraint}
\inprod{s',m',r'| S_i^{\alpha}|s,m,r} = 0 \text{ if } s' \not\in \{s-1, s, s+1\} \qquad \text{ for }\alpha = X,Y,Z.
\end{align}
Combining \cref{def:twirl_def} with an expansion of $\Pi_s$ into the $\ket{s, m, r}$ eigenbasis, we have:
\begin{align}\label{eq:calT_Pi_simplified}
\calT(\Pi_s) = \sum_{m,r}\sum_{i,\alpha} S_i^\alpha \ketbra{s,m,r}{s,m,r} S_i^\alpha = c_{s, s-1} \Pi_{s-1} + c_{s, s} \Pi_{s} + c_{s, s+1} \Pi_{s+1},
\end{align}
where we applied the constraint in \cref{eq:S_i_clebsh_gordan_constraint} to simplify \cref{eq:calT_Pi_expansion}. Placed into \cref{eq:lind_on_Pi_expanded}, we then have
\begin{align}\label{eq:lind_Pi_simplified}
\lind_\loc(\Pi_s) &= \gamma(E_{s-1} - E_s) c_{s, s-1} \Pi_{s-1} + \gamma(E_{s+1} - E_s) c_{s, s+1} \Pi_{s+1} \nonumber \\
&\quad -\gamma(E_s - E_{s-1}) c_{s-1, s} \Pi_s - \gamma(E_s - E_{s+1}) c_{s+1, s} \Pi_s.
\end{align}
Multiplying both sides of the equality by \(\Pi_{s'}\) for some \(s' \in \mathcal{S}\) and taking the trace gives
\begin{equation}
  \Tr[ \Pi_{s'} \lind_{\loc}(\Pi_{s}) ] =
  \begin{cases}
    \gamma( E_{s-1} - E_{s}) c_{s, s-1} \Tr[ \Pi_{s-1} ] & \text{if } s' = s - 1, \\[1ex]
    \gamma( E_{s+1} - E_{s}) c_{s, s+1} \Tr[ \Pi_{s+1} ] & \text{if } s' = s + 1, \\[1ex]
    -\Big(\gamma( E_{s} - E_{s-1}) c_{s - 1, s} - \gamma(E_{s} - E_{s+1}) c_{s + 1, s} \Big)\Tr[ \Pi_{s} ] & \text{if } s' = s, \\[1ex]
    0 & \text{otherwise.}
  \end{cases}
\end{equation}
The result follows by recalling that $L_{s, s'} = \Tr\left[ \Pi_{s'} \right]^{-1}\Tr\left[ \Pi_{s'} \lind_\loc(\Pi_{s}) \right]$ (\cref{eq:pauli_master_generator}.)
\end{proof}

We take the opportunity to record a consequence of this proof that will be useful later.
\begin{corollary}[Local operators are band-diagonal]\label{cor:band-diagonal-local-operators}
For any $i\in [n]$ and $\alpha \in \{X,Y,Z\}$, the Bohr-frequency element of the single-site Pauli matrix $S_i^\alpha(\omega)$ is nonzero only if $\omega = O(1)$.
\end{corollary}
\begin{proof}
Follows from \cref{eq:S_i_clebsh_gordan_constraint} applied to the definition of the Bohr frequency decomposition in \cref{eq:bohr_frequency_decomposition}.
\end{proof}

\subsection{Non-Zero Entries of \texorpdfstring{$L$}{L}}
With the knowledge that $L$ is tridiagonal, computing the nonzero entries greatly simplifies. We will use the following identities. At a high level, these identities are at least plausible because $\calT$ maps $\calA^{(0)}$ to itself, and every element in $\calA^{(0)}$ may be expressed as a polynomial in $\mathbf{S}_{\tot}^2$.

\begin{proposition}[Twirling powers of $\mathbf{S}_{\tot}^2$]\label{prop:S_tot_twirl}
    For a system with \(n\) sites, we have the following identities for the Pauli twirl $\calT$ of $\mathbf{S}_\tot^2$ (c.f. \cref{eq:S_tot_def}):
  \begin{align}
    \calT \left( \left( \mathbf{S}_{\tot}^2 \right)^{0} \right)   & = \frac{3}{4} n \idty, 
    \label{eq:zero-twirl-identity} \\[1ex]
      \calT \left( \left( \mathbf{S}_{\tot}^2 \right)^{1} \right) & = \frac{3}{2} n \idty + \bigg(\frac{3}{4} n - 2 \bigg) \mathbf{S}_{\tot}^2, \label{eq:linear-twirl-identity} \\[1ex]
      \calT \left( \left( \mathbf{S}_{\tot}^2 \right)^{2} \right) & = 3 n \idty + \bigg( 5 n - 6 \bigg) \mathbf{S}_{\tot}^2 + \bigg( \frac{3}{4} n - 4 \bigg) (\mathbf{S}_{\tot}^2)^{2}. \label{eq:quadratic-twirl-identity}
  \end{align}
\end{proposition}
\begin{proof}
    We begin by observing that the Pauli twirl operator $\calT$ is comprised of a sum of $n$ commuting operators: the maps given by $A\mapsto  \sum_\alpha S_i^\alpha A S_i^\alpha$ where $i \in [n]$. A quick computation reveals that $\sum_{\alpha} S^\alpha S^\nu S^\alpha = -\frac{1}{4} S^\nu$ for $\nu=X,Y,Z$ and $\sum_{\alpha} S^\alpha S^0 S^\alpha = \frac{3}{4} S^0$, where we have adopted the convention that $S^0=\idty$. So, we may simultaneously diagonalize this collection of $n$ maps and thus the Pauli twirl $\calT$ by choosing the basis of $\calB(\calH)$ of Pauli strings: let $\vec{\nu} = (\nu_1,\dots, \nu_n)$ where each $\nu_i\in \{0,X,Y,Z\}$ and define the Pauli string to be 
    \begin{equation}
    S_{\vec{\nu}} := S_{1}^{\nu_1} S_{2}^{\nu_2} \dots S_{n}^{\nu_n}.
    \end{equation} We define the weight $\abs{\vec{\nu}}$ of the multi-index $\vec{\nu}$ to be the number of nonzero entries, e.g. $\abs{(0,X,Y)} = 2$. The eigenvalue corresponding to the Pauli string $S_{\vec{\nu}}$ depends only on the weight: for each nonzero entry in $\vec{\nu}$, we get a contribution of $-\frac{1}{4}$, and for each zero entry, we get a contribution of $\frac{3}{4}$, and so the promised diagonalization is
    \begin{equation} \label{eq:diagonalize pauli twirl}
        \calT(S_{\vec{\nu}}) = \paran{\frac{3}{4} n - k} S_{\vec{\nu}}, \qquad k=\abs{\vec{\nu}}.
    \end{equation}
    We immediately obtain~\cref{eq:zero-twirl-identity}, since $(\mathbf{S}_{\tot}^2)^0 = \idty$ corresponds to the trivial string $\vec{\nu} = (0,0,\dots,0)$. To obtain the next identity, we expand the operator $\mathbf{S}_{\tot}^2$ in the Pauli basis:
    \begin{equation}
        \mathbf{S}_{\tot}^2 = (S_\tot^X)^2 + (S_\tot^Y)^2 + (S_\tot^{Z})^2 = \frac{3}{4} n \idty + \sum_{\substack{\alpha = X,Y,Z\\ i\neq j}} S_i^{\alpha} S_j^{\alpha} . 
    \end{equation} We may then apply~\cref{eq:diagonalize pauli twirl} to arrive at
    \begin{equation}
        \calT(\mathbf{S}_{\tot}^2) = \frac{3^2}{4^2}n^2 \idty + \paran{\frac{3}{4}n - 2} \sum_{\substack{\alpha = X,Y,Z\\ i\neq j}} S_i^{\alpha} S_j^{\alpha} = \frac{3}{2}n\idty + \paran{\frac{3}{4}n - 2} \mathbf{S}_{\tot}^2 ,
    \end{equation} which is precisely~\cref{eq:linear-twirl-identity}. Onto the final identity. We begin by expanding $(\mathbf{S}_{\tot}^2)^2$:
    \begin{equation} \label{eq:expanding square of casimir}
        (\mathbf{S}_{\tot}^2)^2 = \frac{3^2}{4^2} n^2 \idty + \frac{3}{4}2 n \sum_{\alpha = X,Y,Z} \sum_{i\neq j} S_{i}^\alpha S_j^\alpha + \sum_{\alpha,\eta = X,Y,Z} \sum_{\substack{i_1 \neq j_1 \\ i_2\neq j_2}} S_{i_1}^\alpha S_{j_1}^\alpha S_{i_2}^\eta S_{j_2}^\eta . 
    \end{equation}
    A priori the Pauli strings $S_{i_1}^\alpha S_{j_1}^\alpha S_{i_2}^\eta S_{j_2}^\eta$ appearing in the final sum may be weight 2, 3, or 4, since $i_1\neq j_1$ and $i_2\neq j_2$. If $i_1,i_2,j_1,j_2$ are all distinct, the string is of weight 4, and if two pairs are equal, the string is of weight 2. A quick argument will reveal cancellation for the strings of weight 3 in the sum. The collection of Pauli strings forms an orthogonal (not necessarily normalized) basis of $\calB(\calH)$ with respect to the Hilbert-Schmidt inner product $\inprod{A,B} = \Tr A^\dagger B$. Each Pauli string $S_{\vec{\nu}}$ is Hermitian, and the operator $\mathbf{S}_{\tot}^2$ is Hermitian, so the coefficients of $\mathbf{S}_{\tot}^2$ in the Pauli string basis $\inprod{S_{\vec{\nu}},\mathbf{S}_{\tot}^2}$ are all purely real. However, observe that if e.g. $i_1=i_2$, then  
    \begin{equation}
        S_{i_1}^\alpha S_{j_1}^\alpha S_{i_2}^\eta S_{j_2}^\eta = i \eps_{\alpha \eta \gamma } S_{i_1}^\gamma S_{j_1}^\alpha S_{j_2}^\eta , 
    \end{equation} where $\eps$ denotes the Levi-Cevita symbol. Since $\eps_{\alpha \eta \gamma }$ can only take on values $\pm 1$, the coefficient of this operator in the Pauli basis is purely imaginary. This likewise occurs for the other weight 3 Pauli strings appearing in the sum, and so the coefficients $\inprod{S_{\vec{\nu}}, \mathbf{S}_{tot}^2} = 0$ whenever $\abs{\vec{\nu}}=3$. Thus~\cref{eq:expanding square of casimir} becomes
    \begin{equation} 
        (\mathbf{S}_{\tot}^2)^2 = \frac{3^2}{4^2} n^2 \idty + \frac{3}{4}2 n \sum_{\alpha = X,Y,Z} \sum_{i\neq j} S_{i}^\alpha S_j^\alpha + \sum_{\alpha,\eta = X,Y,Z} \sum_{\mathrm{pair}} S_{i_1}^\alpha S_{j_1}^\alpha S_{i_2}^\eta S_{j_2}^\eta + \sum_{\alpha,\eta = X,Y,Z} \sum_{\mathrm{dist}} S_{i_1}^\alpha S_{j_1}^\alpha S_{i_2}^\eta S_{j_2}^\eta 
    \end{equation} where the ``pair'' sum is over all possible $i_1\neq j_1, i_2\neq j_2$ where two of these variables are equal, and the ``dist'' sum is over all choices of $i_1,i_2,j_1,j_2$ where each is distinct. The terms in the ``pair'' sum consist only of weight 2 Paulis and the terms in the ``dist'' sum consist only of weight 4 Paulis, so we may now apply $\calT$. The final identity is then obtained by the requisite bookkeeping.
\end{proof}

We are now ready to establish \cref{prop:L is tridiagonal}.
\begin{proof}[Proof of \cref{prop:L is tridiagonal}]
\cref{lem:L_tridiagonal} already established the tridiagonality of $L$. To compute those entries,
we return to the proof of \cref{lem:L_tridiagonal}, in particular \cref{eq:lind_Pi_simplified}, and focus on computing the constants $c_{s, s \pm 1},c_{s, s}$. We will show that, if $|s-s'|\le1$,
\begin{equation}\label{eq:c_s_s_prime_generic}
c_{s, s'}
= \frac{2 s + 1}{2 s' + 1} \frac{ n + 2 + s'(s'+1) - s(s+1)}{4}.
\end{equation}
Combined with \cref{eq:lind_Pi_simplified}, the facts that $\gamma(\lambda_s - \lambda_{s+1}) = \gamma(2(s+1)/n)$ and $\gamma(\lambda_s - \lambda_{s-1}) = \gamma(-2s/n)$ prove \cref{eq:entries_of_L} and establish the proposition.

To compute the coefficients, we will employ the general following strategy. As an illustration, suppose we would like to compute $c_{s, s-1}$. We eliminate the other two terms in \cref{eq:calT_Pi_simplified} by multiplying both sides by $(\mathbf{S}_\tot^2 - \lambda_{s})(\mathbf{S}_\tot^2 - \lambda_{s+1})$. Since $\mathbf{S}_\tot^2 \Pi_\lambda = \lambda_\lambda \Pi_\lambda$ for any $\lambda \in \calS$, we get
\begin{align}
\calT(\Pi_s)(\mathbf{S}_\tot^2 - \lambda_{s})(\mathbf{S}_\tot^2 - \lambda_{s+1}) = c_{s, s-1} (\lambda_{s-1} - \lambda_s) (\lambda_{s-1} - \lambda_{s+1}) \Pi_{s-1}.
\end{align}
Taking traces and using the fact that $\calT(\cdot)$ is self adjoint with respect to the Hilbert-Schmidt inner product due to the ciclicity of trace, we get
\begin{align}
c_{s, s-1} (\lambda_{s-1} - \lambda_s) (\lambda_{s-1} - \lambda_{s+1}) \Tr\left[\Pi_{s-1} \right] &= \Tr\left[ \calT(\Pi_s)(\mathbf{S}_\tot^2 - \lambda_{s})(\mathbf{S}_\tot^2 - \lambda_{s+1}) \right] \\
&= \Tr\left[ \Pi_s \calT\left((\mathbf{S}_\tot^2 - \lambda_{s})(\mathbf{S}_\tot^2 - \lambda_{s+1} )\right) \right].
\end{align}
Now, by \cref{prop:S_tot_twirl}, it follows that 
\begin{align}
\calT\left((\mathbf{S}_\tot^2 - \lambda_{s})(\mathbf{S}_\tot^2 - \lambda_{s+1} )\right) = r_0 + r_1 (\mathbf{S}_\tot^2) + r_2 (\mathbf{S}_\tot^2)^2,
\end{align}
for some $r_0,r_1,r_2$. Hence
\begin{align}
\Tr\left[ \Pi_s \calT\left((\mathbf{S}_\tot^2 - \lambda_{s})(\mathbf{S}_\tot^2 - \lambda_{s+1} )\right) \right] = \Tr\left[ \Pi_s \left( r_0 + r_1 \lambda_s + r_2 \lambda_s^2 \right) \right]
\end{align}
and rearranging gives
\begin{align}
    c_{s, s-1} = \frac{\Tr[\Pi_s]}{\Tr[\Pi_{s-1}]} \frac{r_0 + r_1 \lambda_s + r_2 \lambda_s^2}{(\lambda_{s-1} - \lambda_s) (\lambda_{s-1} - \lambda_{s+1})}.
\end{align}
Since the dimensions of the eigenspaces are given in \cref{prop:diagonalization of H}, we have all the necessary information to compute this coefficient.

This procedure works in general, keeping in mind that the bulk spins and the edge spins must be treated separately since the latter only contain two nonzero terms in \cref{eq:calT_Pi_simplified}, as opposed to three. Careful casework finishes the proof.
\end{proof}

\section{The Spectral Gap of the Coarse-Grained Pauli Master Equation}\label{sec:proof_gap_of_L}
To derive a lower bound on the spectral gap of the coarse-grained Pauli master equation, we make use of Cheeger's inequality for classical Markov chains.

\begin{proposition}[Cheeger lower bound for $1$D chains] \label{thm:cheeger}
  Let \(L\) be the generator for a reversible birth-and-death chain with stationary distribution \(\pi\) on $\{0,\ldots,N\}$. 
  The spectral gap of \(L\) is lower bounded by
  \begin{equation}
    \mathrm{gap}(L) \geq \frac{\Phi_{*}^{2}}{2 \| L \|_{\infty, \infty}}
  \end{equation}
  where $\| L \|_{\infty,\infty} := \max_{i} \sum_{j}^{}  |L_{i,j}|$ and the conductance $\Phi_*$ is
  \begin{align}
    \Phi_{*} = \min_{0 < m \leq N} \frac{\pi(m - 1) L_{m, m-1}}{\min\Big\{ \pi([0, m-1]), \pi([m, N]) \Big\}}.
  \end{align}
\end{proposition}
\begin{proof}
  The lower bound on the gap follows directly from~\cite[Theorem 13.10]{levin_markov_2017} applied to the lazy discrete time transition matrix \(P := \mathds{1} + L / \| L \|_{\infty, \infty}\) and
  \begin{equation}
  \Phi_{*} := \min_{A} \frac{\sum_{x\in A} \sum_{y\in A^c}\pi(x)L_{x, y}}{\min\{ \pi(A), \pi(A^{c}) \}}.
  \end{equation}
  For the reduction from arbitrary sets to intervals, from \cite[Corollary 4.4]{lawler_bounds_1988} the conductance is minimized by choosing $A$ so that both $A$ and $A^c$ are connected.
  For the birth--death chain this is only possible if both $A$ and $A^c$ are intervals.
\end{proof}

Hence to prove a lower bound on the gap of \(L\), by Cheeger's inequality, it is enough to prove a lower bound on $\Phi_*$.
Throughout this section, let \(s_{\min}:=\min \mathcal{S}\) and \(s_{\max}:=\max \mathcal{S}\).
Since both \(\pi(x)\) and \(L_{x, y}\) are known, the only unknown quantity in the expression for the conductance is
\begin{equation}
  \label{eq:min-cdf-birth-death}
  \min \Big\{\pi([s_{\min}, m - 1]),\, \pi([m, s_{\max}])\Big\}
\end{equation}
which we must upper bound to lower bound \(\Phi_{*}\).

With this knowledge, the main goal of this section is to prove the following proposition.
\begin{proposition}
  \label{prop:cdf-upper-bound}
  Let \(\pi\) denote the stationary distribution of the coarse-grained Pauli master equation.
  For any system size \(n\) and any \(m \in \mathcal{S}\), we have the following upper bound:
  \begin{equation}
    \label{eq:prop-cdf-upper-bound}
    \min\Big\{ \pi([s_{\min}, m]), \pi([m, s_{\max}]) \Big\} \leq
    \begin{cases}
      O(1) n^{\frac{1}{2}} \pi(m) & \beta \neq 2 \\[1ex]
      O(1) n^{\frac{3}{4}} \pi(m) & \beta = 2.
    \end{cases}
  \end{equation}
\end{proposition}
Since, for any spin \(m\),
\begin{equation}
  \min\Big\{ \pi([s_{\min}, m - 1 ]), \pi([m, s_{\max}]) \Big\} \leq
  \min\Big\{ \pi([s_{\min}, m - 1]), \pi([m - 1, s_{\max}]) \Big\},
\end{equation}
\cref{prop:cdf-upper-bound} can be used to lower bound \(\Phi_{*}\) and hence lower bound the gap.
Changing the \([m, s_{\max}]\) in the minimum to \([m - 1, s_{\max} ]\) will make the later arguments slightly cleaner and does not impact scaling with \(n\).
Since \(L_{s + 1, s} = \Omega(n)\) and \(\| L \|_{\infty, \infty} = O(n)\), the following result is an immediate corollary of~\cref{thm:cheeger,prop:cdf-upper-bound}.
\begin{corollary}
  If \(L\) is the generator for the coarse-grained Pauli master equation then
  \begin{equation}
    \mathrm{gap}(L) =
    \begin{cases}
      \Omega(1) & \beta \neq 2 \\
      \Omega(n^{-\frac{1}{2}}) & \beta = 2
    \end{cases}
  \end{equation}
\end{corollary}
We now outline the proof of~\cref{prop:cdf-upper-bound} contained in the following subsections. In~\cref{sec:sharp-stationary-measure}, we first define the rescaled variable \(x_{s} := (2 s) / n\), so that $x_s\in [0,1]$, and derive an expression for $\pi(s)$ in terms of $x_s$ (\Cref{lem:sharp-pi-s}).
We will then find that the cumulative distribution function \(\pi([s_{\min}, m])\) can be estimated using Laplace's method \cite[Chapter II.1]{wong_asymptotic_2001}.
In \cref{sec:laplace-and-three-regions}, we briefly review Laplace's method and show that it naturally partitions the interval \([0, 1]\) into three regions. Each region is separately analyzed in~\cref{sec:cdf-bound-right,,sec:cdf-bound-left,,sec:cdf-bound-laplace}.

\subsection{Sharp Bounds for the Stationary Measure}
\label{sec:sharp-stationary-measure}
From the spectral decomposition of \(H\) (\cref{prop:diagonalization of H}), we have that for any \(\beta \geq 0\) the partition function \(Z_{\beta}\) and stationary measure \(\pi(s)\) are written
\begin{equation}
  \label{eq:stationary}
  \begin{split}
    Z_{\beta} & := \Tr{(e^{- \beta H})}= \sum_{s}^{}  e^{(\beta / n) s(s + 1)} \dim{(\Pi_{s})}, \\[1ex]
    \pi(s) & := \frac{1}{Z_{\beta}} e^{(\beta / n) s(s + 1)} \dim{(\Pi_{s})}.
  \end{split}
\end{equation}
To clean up our results in this section, we will use the notation \(F(s) \sim G(s)\) to mean there exists constants \(C, c > 0\) independent of \(s, \beta, n\) so that for all \(n\) sufficiently large \(c G(s) \leq F(s) \leq C G(s)\).
With this notation, the main result of this section is the following lemma.
\begin{lemma}
  \label{lem:sharp-pi-s}
  Let \(\pi(s)\) and \(Z_{\beta}\) be defined as in~\cref{eq:stationary}.
  For all \(n \geq 2\),
  \begin{equation}
    \pi(s) \sim  \frac{n^{\frac{1}{2}} }{Z_{\beta}} g_{\beta}(x_{s}) e^{- n f_{\beta}(x_{s})} 
  \end{equation}
  where the functions \(f_{\beta}(x_{s})\) and \(g_{\beta}(x_{s})\) are:
  \begin{align}
    f_{\beta}(x_{s}) & := - \frac{\beta}{4} x_{s}^{2} - H_{b}\left( \frac{1 + x_{s}}{2} \right) \label{eq:fbeta} \\[1ex]
    g_{\beta}(x_{s}) & := \frac{(x_{s} + n^{-1})^{2}}{(1 - x_{s}^{2}  + 4n^{-1})^{\frac{1}{2}}} e^{(\beta / 2) x_{s}}.   \label{eq:g}
  \end{align}
and \(H_{b}(x) := - x \ln(x) - (1 - x) \ln(1 - x)\) denotes the binary entropy function.
\end{lemma}

One easily sees that \(e^{(\beta / n) s (s + 1)} = e^{n (\beta / 4) x_{s}^{2} + (\beta / 2) x_{s}}\) so the main work is in simplifying \(\dim(\Pi_{s})\).
Using the non-asymptotic version of Stirling's approximation due to Robbins~\cite{robbins_remark_1955} we can obtain an expression for \(\dim{(\Pi_{s})}\) in terms of the normalized variables which is sharp up to multiplicative constants:
\begin{lemma}
  \label{eq:pi-s-dim-bound}
  For all \(n \geq 2\),
  \begin{equation}
    \dim{(\Pi_{s})}
    \sim
    \frac{n^{\frac{1}{2}} (x_{s} + n^{-1})^{2}}{(1 - x_{s}^{2} + 4 n^{-1})^{\frac{1}{2}}} \exp\left\{ n H_{b}\left( \frac{1 + x_{s}}{2} \right) \right\}.
  \end{equation}
\end{lemma}
Note that proving this lemma implies~\cref{lem:sharp-pi-s}.
\begin{proof}
  From the dimensions provided in~\cref{thm:qubit schur-weyl},
  \begin{equation}
    \dim{(\Pi_{s})} = \frac{(2 s + 1)^{2}}{n / 2 + s + 1} \binom{n}{\frac{n}{2}- s}\label{eq:formula-for-dim}
  \end{equation}
  so the main complication arises from simplifying the binomial coefficient.
  Since \(\binom{n}{0} = 1\), one can verify the conclusion of the lemma is true when \(s = \frac{n}{2}\) so we only need to consider \(s \leq \frac{n}{2} - 1\).
  
  By Robbins' improvement to Stirling's approximation \cite{robbins_remark_1955}, for all \(n \geq 1\) we have
  \begin{equation}
    n! \sim n^{\frac{1}{2}} n^{n} e^{-n}.
  \end{equation}
  Applying this approximation to the binomial coefficient gives
  \begin{equation}
    \binom{n}{\frac{n}{2} - s} \sim \sqrt{\frac{n}{(n / 2 + s) (n / 2 - s)}} n^{n} \left( \frac{n}{2} + s \right)^{- (n / 2 + s)}  \left( \frac{n}{2} - s \right)^{- (n / 2 - s)}.
  \end{equation}
  Elementary algebra shows that
  \begin{equation}
    \ln\left( n^{n} \left( \frac{n}{2} + s \right)^{- (n / 2 + s)}  \left( \frac{n}{2} - s \right)^{- (n / 2 - s)} \right)
    = n H_{b}\left( \frac{n + 2 s}{2 n} \right)
  \end{equation}
  so
  \begin{equation}
    \binom{n}{\frac{n}{2} - s} \sim \sqrt{\frac{n}{(n / 2 + s) (n / 2 - s)}} \exp\left\{  n H_{b}\left( \frac{1 + x_{s}}{2} \right) \right\}.
  \end{equation}
  Hence the above expression gives upper and lower bounds for the binomial coefficient so long as \(0 \leq s < \frac{n}{2}\).
  
  To get upper and lower bounds which also hold for \(s = \frac{n}{2}\), we introduce a \(+ 2\) to the denominator under the square root.
  One easily checks that for all \(n \geq 2\) and \(s \leq \frac{n}{2} - 1\)
  \begin{equation}
     \frac{1}{(n / 2)^{2} - s^{2} + 2} \leq 
    \frac{1}{(n / 2 + s)(n / 2 - s)} \leq \frac{2}{(n / 2)^{2} - s^{2} + 2},
  \end{equation}
  so
  \begin{equation}
    \binom{n}{\frac{n}{2} - s} \sim \sqrt{\frac{n}{(n / 2)^{2} - s^{2} + 2}} \exp\left\{  n H_{b}\left( \frac{1 + x_{s}}{2} \right) \right\}.
  \end{equation}
  The result now follows by algebra and noting that
  \begin{equation}
    \frac{(2 s + 1)^{2}}{n / 2 + s + 1} = \frac{n^{2} (x_{s} + n^{-1})^{2}}{n / 2 + s + 1} \sim \frac{n (x_{s} + n^{-1})^{2}}{1 + x_{s} + 2 n^{-1}} \sim n (x_{s} + n^{-1})^{2}
  \end{equation}
  where we used that \(0 \leq x_{s} \leq 1\).
\end{proof}

\subsection{A Tale of Three Regions: Laplace's Method}
\label{sec:laplace-and-three-regions}
To make the review of Laplace's method more concrete, let's consider estimating the partition function \(Z_{\beta}\) at some inverse temperature \(\beta\).
While this estimate will not be used for the proof of~\cref{prop:cdf-upper-bound}, this calculation highlights the main ideas.

From the calculations in the previous subsection, up to multiplicative constants, \(Z_{\beta}\) can be expressed as
\begin{equation}
  Z_{\beta} \sim \sum_{x_{s}}^{} n^{\frac{1}{2}} g(x_{s}) e^{- n f_{\beta}(x_{s})}.
\end{equation}
We'll now give an informal argument that for this choice of normalization the partition function \(Z_{\beta}\) has the following scaling in \(n\):
\begin{equation}
  Z_{\beta} =
  \begin{cases}
    O(n e^{-n f_{\beta}(x_{*})}) & \beta \neq 2 \\
    O(n^{\frac{5}{4}} e^{-n f_{\beta}(x_{*})}) & \beta = 2.
  \end{cases}
  \quad
  \text{where}
  \quad
  x_{*} := \argmin_{x \in [0,1]} f_{\beta}(x).
\end{equation}

Since the normalized variable \(x_s\) is uniformly spaced on \([0,1]\) with mesh \(2/n\), by Riemann sums, we should expect that
\begin{equation}
 \frac{2}{n} \sum_{x_{s}}^{} g(x_{s}) e^{- n f_{\beta}(x_{s})} \approx \int_{0}^{1} g(x) e^{- n f_{\beta}(x)} d x.
\end{equation}
Integrals of this form can be estimated using ``Laplace's method'' which suggests the behavior of this integral is dominated by the behavior near the minimum \(x_*\) \cite[Chapter II.1]{wong_asymptotic_2001}.
More specifically, suppose that \(f_{\beta}(x)\) has a unique minimum at \(x_{*}\) on \([0,1]\) and \(f_{\beta}''(x_{*}) > 0\).
Laplace's method gives the estimate
\begin{equation}
  \int_{0}^{1} g(x) e^{- n f_{\beta}(x)} d x \approx \sqrt{\frac{2\pi}{n |f_{\beta}''(x_{*})|}} g(x_{*}) e^{- n f_{\beta}(x_{*})}.
\end{equation}
This approximation essentially comes Taylor expanding \(f_{\beta}(x)\) to second order at minimum \(x_{*}\) and using properties of the Gaussian.
It can be checked that when \(\beta \neq 2\), \(f_{\beta}''(x_{*}) > 0\) so Laplace's method applies and \(Z_{\beta} = O(n e^{-n f_{\beta}(x_{*})})\).
The key difference at \(\beta = 2\) is that \(f_{\beta}''(x_{*}) = f_{\beta}'''(x_{*}) = 0\) so the usual Laplace estimate does not apply.
Instead, we can Taylor expand to fourth order at \(x_{*}\) which leads to \(Z_{\beta} = O(n^{\frac{5}{4}} e^{-n f_{\beta}(x_{*})})\).

While Laplace's method gives the correct asymptotic behavior, there are a few technical issues which prevent one from applying standard results directly.
The first difficultly lies in replacing the sum with an integral.
Estimates for replacing a sum with an integral usually rely on derivatives of the integrand which will introduce an undesirable dependence on \(n\).
Since we only want an upper bound, we can avoid this dependence using the following lemma which slightly generalizes the integral test for monotone functions:
\begin{lemma}
  \label{lem:integral-test}
  Let \(h : [0, 1] \to \mathbb{R}\) be a non-negative, continuously differentiable function and let \(\{ x_{i} \}_{i=1}^{n} \subseteq [0,1]\) be a discrete set of points listed in increasing order with minimum distance \(\Delta := \min_{i \neq j} | x_{i} - x_{j}|\).
  If the derivative of \(h\) changes sign no more than \(d\) times on \([0,1]\), then
  \begin{equation}
    \sum_{i = j}^{\ell} h(x_{i}) \leq \Delta^{-1} \int_{x_{j}}^{x_{\ell}} h(y) dy + (d + 1) \left( \max_{x \in [x_{j}, x_{\ell}]} h(x) \right)
\end{equation}
\end{lemma}
\begin{proof}
  Partition \([0, 1]\) into \((d + 1)\) intervals so that \(h\) is monotone on each interval; the result follows by applying the integral test to the ordered points lying in each of these intervals.
  The additional term multiplied by \((d + 1)\) accounts for the endpoints.
\end{proof}
The second technical difficulty lies with the properties of the function \(g_{\beta}(x_{s})\).
In particular, when \(x_{s} \approx 1\) the factor \((1 - x_{s}^{2} + 4 n^{-1})^{-\frac{1}{2}}\) is of order \(n^{\frac{1}{2}}\), so the prefactor is enhanced near the right endpoint.
One therefore has to check that this square-root singularity is dominated by the exponential decay from \(e^{- n f_{\beta}(x_{s})}\).
A careful analysis shows that this is indeed the case, so the endpoint singularity does not change the \(n\)-dependence predicted by Laplace's method.

Following this discussion, we define three regions: \(R_{\mathrm{Laplace}}\), \(R_{\mathrm{Right}}\), \(R_{\mathrm{Left}}\).
The definition of these regions depend on the minimum \(x_{*}\) and a constant \(c_{\beta}\) to be chosen in the next section.
\begin{enumerate}
\item Laplace Region: Main contribution; most of the mass of the stationary distribution occurs here
  \begin{equation}
    R_{\mathrm{Laplace}} := [ x_{*} - c_{\beta}, x_{*} + c_{\beta} ] \cap [0, 1].
  \end{equation}

\item Right Edge: Subleading contribution; here we will show that that the singularity of \((1 - x_{s}^{2})^{-\frac{1}{2}}\) is exponentially suppressed.
  \begin{equation}
    R_{\mathrm{Right}} := [ x_{*} + c_{\beta}, 1 ] \cap [0, 1].
  \end{equation}
\item Left Edge: Subleading contribution. 
  \begin{equation}
      R_{\mathrm{Left}} := [0, x_{*} - c_{\beta}] \cap [0, 1].
  \end{equation}
\end{enumerate}

\subsubsection{Defining the Three Regions}
To choose the constant \(c_{\beta}\) defining the three regions, we will need to make use of some properties of the function \(f_{\beta}\).
We will take this chance to collect all of the properties of \(f_{\beta}(x)\) we need for our proof in a single lemma; the most relevant properties for the choice of \(c_{\beta}\) are related to the behavior near the minimum:
\begin{lemma}
  \label{lem:properties-f_beta}
  The function \(f_{\beta}\) has the following properties for all \(\beta \geq 0\):
  \begin{enumerate}[topsep=1ex,itemsep=1ex,label=(\roman*)]
  \item \(f_{\beta}\) is continuous on \([0,1]\) and infinitely continuously differentiable on \((0,1)\).
  \item Both \(f_{\beta}''\) and \(f_{\beta}^{(4)}\) are both strictly increasing on \([0,1)\).
  \item \(f_{\beta}\) has a unique minimum on \([0,1]\).
  \item If \(x_{*}\) denotes the unique minimum of \(f_{\beta}\), then \(0 \leq x_{*} \leq 1 - e^{-\beta}\). \label{item:minimum_bound}
  \end{enumerate}
  With regards to the minimum, we have the following properties which depend on temperature:
  \begin{itemize}[itemsep=1ex]
  \item For \(\beta < 2\): The minimum occurs at \(x_{*} = 0\) and \(f_{\beta}''(x_{*}) > 0\).  
  \item For \(\beta = 2\): The minimum occurs at \(x_{*} = 0\) and \(f_{\beta}''(x_{*}) = f_{\beta}'''(x_{*}) = 0\) but \(f_{\beta}^{(4)}(x_{*}) > 0\).
  \item For \(\beta > 2\): The minimum occurs at some \(x_{*} > 0\) and \(f_{\beta}''(x_{*}) > 0\).
  \end{itemize}
\end{lemma}
We prove this lemma in \cref{sec:properties-f_beta}.

\begin{figure}[h]
  \centering

  \caption{Plots of \(f_{\beta}\) for different values of \(\beta\). Most of the mass of the stationary distribution concentrates around the minimum of $f_\beta$, as predicted by Laplace's method. The phase transition at \(\beta=2\) occurs when the minimum moves from \(x_*=0\) to positive \(x_*>0\).}
  \label{fig:f-beta}
  \includegraphics[width=.9\linewidth]{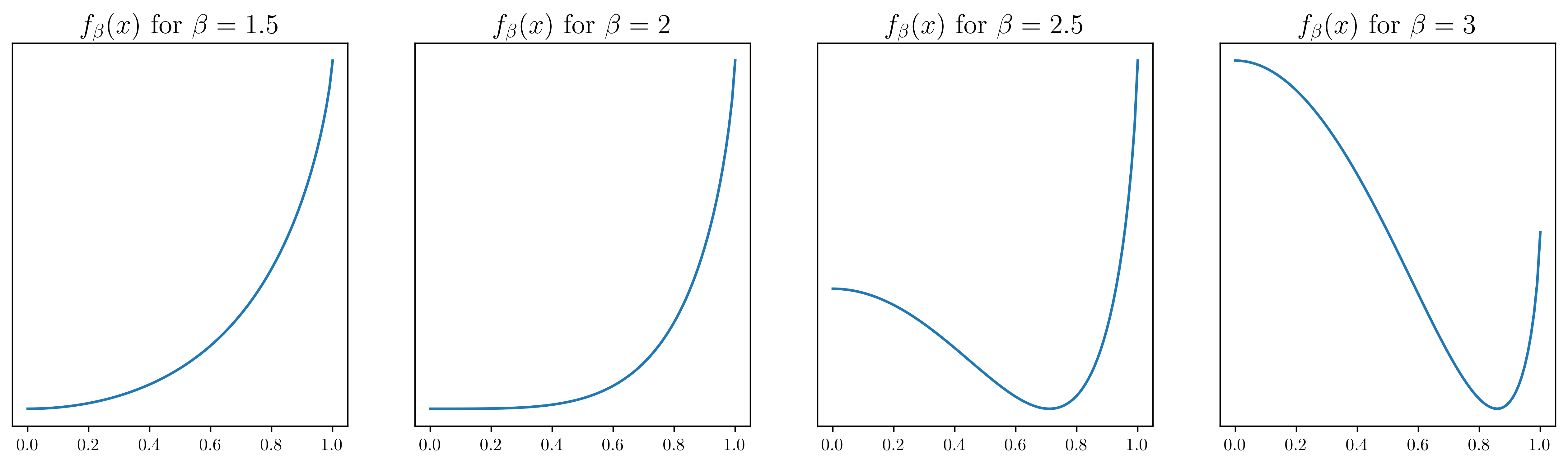}
  
\end{figure}

By the above lemma, when \(\beta \neq 2\), we know that \(f_{\beta}''(x_{*}) > 0\) so by continuity there exists a small interval \([ x_{*} - \delta, x_{*} + \delta]\) where \(f_{\beta}''\) is close to its value at \(x_{*}\).
Therefore, we know the following constant is strictly positive
\begin{equation}
  \tilde{c}_{\beta} := \max \left\{ c : f_{\beta}''(y) \geq \frac{1}{2} f_{\beta}''(x_{*}), ~\forall y \in [x_{*} - c, x_{*} + c ] \cap [0,1) \right\}.
\end{equation}
This \textit{almost} defines \(c_{\beta}\), however to formally handle the singularity at \(x_s = 1\), we'll need to ensure \(x_{*} + c_{\beta}\) is bounded away from \(1\).
Therefore, we set \(c_{\beta}\) to be the minimum
\begin{equation}
  c_{\beta} := \min\left\{ \tilde{c}_{\beta}, \frac{1 - x_{*}}{2} \right\} \quad \text{when} \quad \beta \neq 2.
\end{equation}
Note that this definition implies that \(x_{*} + c_{\beta} \leq \frac{1}{2} (1 + x_{*}) < 1\).

The definition of \(c_{\beta}\) for \(\beta = 2\) is nearly identical to \(\beta \neq 2\), however we instead ask the fourth derivative is close to its value at \(x_{*}\).
More specifically, we define
\begin{equation}
c_{\beta} := \min\left\{ \max \left\{ c : f_{\beta}^{(4)}(y) \geq \frac{1}{2} f_{\beta}^{(4)}(x_{*}),~ \forall y \in [x_{*} - c, x_{*} + c ] \cap [0,1)  \right\}, \frac{1 - x_{*}}{2} \right\} \quad \text{when} \quad \beta = 2.
\end{equation}

Now that the three regions are properly defined, we can prove \cref{prop:cdf-upper-bound}.
We will divide into three cases \(m \in R_{\mathrm{Right}}\), \(m \in R_{\mathrm{Left}}\), and \(m \in R_{\mathrm{Laplace}}\).
As one might expect from the Laplace analysis, the analysis for \(m \not\in R_{\mathrm{Laplace}}\) does not strongly depend on temperature since the minimum is not contained in this interval. 

We will first show \cref{prop:cdf-upper-bound} holds at all temperatures for \(m \in R_{\mathrm{Right}}\) and \(m \in R_{\mathrm{Left}}\) in \cref{sec:cdf-bound-right,sec:cdf-bound-left} respectively.
We will then consider \(m \in R_{\mathrm{Laplace}}\) away from the critical temperature in \cref{sec:cdf-bound-laplace} and note the modifications needed for the critical temperature (\(\beta = 2\)) in \cref{sec:cdf-bound-laplace-critical}.

\subsubsection{Analysis for \texorpdfstring{\(R_{\mathrm{Right}}\)}{R_Right}}
\label{sec:cdf-bound-right}
We start by fixing an arbitrary point \(x_{m} \in R_{\mathrm{Right}}\).
Since most of the probability mass is contained near \(x_{*}\), we should expect that \(\pi([x_{m}, 1]) \leq \pi([0, x_{m}])\) so we'll prove \(\pi([x_{m}, 1]) \lesssim n^{\frac{1}{2}} \pi(m)\) to upper bound the minimum of the two.
The key observation for upper bounding \(\pi([x_{m}, 1])\) is that when \(x_{m} \in R_{\mathrm{Right}}\) the derivative of \(f_{\beta}\) is bounded below by a constant independent of \(n\):
\begin{lemma}
  \label{lem:r-right-derivative-lower-bound}
  For all \(\beta \geq 0\), the derivative of \(f_{\beta}\) is uniformly lower bounded by a positive constant on \(R_{\mathrm{Right}}\).
  That is,
  \begin{equation}
    d_{\beta} := \inf_{y \in R_{\mathrm{Right}}} f_{\beta}'(y) > 0.
  \end{equation}
\end{lemma}
\begin{proof}
  First, observe that \(\lim_{y \to 1} f_{\beta}'(y) = \infty\) so we only need to prove a uniform lower bound for \(y \neq 1\).
  The proof for \(\beta \neq 2\) follows from the observation that since \(f_{\beta}'(x_{*}) = 0\) for any \(y \in R_{\mathrm{Right}}\) we have
  \begin{equation}
    f_{\beta}'(y) = f_{\beta}'(y) - f_{\beta}'(x_{*}) = \int_{x_{*}}^{y} f_{\beta}''(z)  dz \geq f_{\beta}''(x_{*}) c_{\beta}
  \end{equation}
  where lower bound is due to the fact that \(f_{\beta}''(x_{*}) > 0\) and \(f_{\beta}''\) is monotone increasing.

  For \(\beta = 2\), let us pick a point \(y \in [ x_{*} + c_{\beta}, 1)\) and apply Taylor's theorem to \(f_{\beta}'\).
  Since when \(\beta = 2\), \(f_{\beta}'(x_{*}) = f_{\beta}''(x_{*}) = f_{\beta}'''(x_{*}) = 0\) for any \(y \in [ x_{*} + c_{\beta}, 1)\) there exists a \(\xi \in [x_{*}, y]\) so that
  \begin{equation}
    \begin{split}
      f_{\beta}'(y)
      & = f_{\beta}'(x_{*}) + f_{\beta}''(x_{*}) (y - x_{*}) + \frac{f_{\beta}'''(x_{*})}{2!}  (y - x_{*})^{2} + \frac{f_{\beta}^{(4)}(\xi)}{3!}  (y - x_{*})^{3} = \frac{1}{3!} f_{\beta}^{(4)}(\xi) (y - x_{*})^{3} .
    \end{split}
  \end{equation}
  When \(\beta = 2\), \(f_{\beta}^{(4)}(x_{*}) > 0\) and \(f_{\beta}^{(4)}\) is monotone increasing we conclude that for \(y \geq x_{*} + c_{\beta}\) we have \(f_{\beta}'(y) \geq \frac{1}{3!} f_{\beta}^{(4)}(x_{*}) c_{\beta}^{3}\) which proves the lemma.
\end{proof}
Since \(f_{\beta}\) is convex on \(R_{\mathrm{Right}}\), we can lower bound it by its tangent line (which has non-vanishing slope due to \cref{lem:r-right-derivative-lower-bound}).
In particular, when \(x_{m} \in R_{\mathrm{Right}}\) and \(x_{s} \geq x_{m}\) we have
\begin{equation}
  \label{eq:r-right-tangent-lower-bd}
  \begin{split}
    f_{\beta}(x_{s})
    \geq f_{\beta}(x_{m}) + d_{\beta} (x_{s} - x_{m}).
  \end{split}
\end{equation}
Then, from the estimates for the stationary distribution (\cref{lem:sharp-pi-s}), we have
\begin{equation}
  \begin{split}
    \pi([x_{m}, 1]) 
    & \lesssim \frac{n^{\frac{1}{2}}}{Z_{\beta}} \sum_{x_{m} \leq x_{s} \leq 1}^{} \frac{(x_{s} + n^{-1})^{2}}{(1 - x_{s}^{2} + 4 n^{-1})^{\frac{1}{2}}} e^{(\beta / 2) x_{s}} e^{- n f_{\beta}(x_{s})} \\
    & \lesssim \frac{n}{Z_{\beta}} e^{\beta / 2} \sum_{x_{m} \leq x_{s} \leq 1}^{} (x_{s} + n^{-1})^{2} e^{- n f_{\beta}(x_{s})}  \\
    & \lesssim \frac{n}{Z_{\beta}} e^{\beta / 2} e^{- n f_{\beta}(x_{m}) } \sum_{x_{m} \leq x_{s} \leq 1}^{} (x_{s} + n^{-1})^{2} e^{ - n d_{\beta} (x_{s} - x_{m}) },
  \end{split}
\end{equation}
where in the last line we've used the tangent line lower bound (\cref{eq:r-right-tangent-lower-bd}).

Using calculus, it is easy to verify that for \(n\) sufficiently large, the map $x \mapsto (x + n^{-1})^{2} e^{ - n d_{\beta} (x - x_{m}) }$ is decreasing on \([x_{*} + c_{\beta}, 1]\).
 Hence, the maximum of the summand occurs at \(x_{m}\) and so using \cref{lem:integral-test}, we have
\begin{equation}
  \sum_{x_{m} \leq x_{s} \leq 1}^{} (x_{s} + n^{-1})^{2} e^{ - n d_{\beta} (x_{s} - x_{m}) }   \leq \frac{n}{2} \int_{x_{m}}^{1} (y + n^{-1})^{2} e^{- n d_{\beta} (y - x_{m})} d y + 2 (x_{m} + n^{-1})^{2}.
\end{equation}
Using Laplace's method \cite[Chapter II, Theorem 1]{wong_asymptotic_2001}, we can upper bound this integral by
\begin{equation}
  \frac{n}{2} \int_{x_{m}}^{1} (y + n^{-1})^{2} e^{- n d_{\beta} (y - x_{m})} d y \leq \frac{1}{2 d_{\beta}} (x_{m} + n^{-1})^{2} ( 1 + O(n^{-1}) ).
\end{equation}
It follows that for all \(n\) sufficiently large
\begin{equation}
    \sum_{x_{m} \leq x_{s} \leq 1}^{} (x_{s} + n^{-1})^{2} e^{ - n d_{\beta} (x_{s} - x_{m}) } \lesssim (x_{m} + n^{-1})^{2}, 
\end{equation}
and therefore we obtain the desired upper bound
\begin{equation}
  \begin{split}
    \pi([x_{m}, 1])
    & \lesssim \frac{n}{Z_{\beta}} e^{\beta / 2} e^{-n f_{\beta}(x_{m})}  (x_{m} + n^{-1})^{2} \\[1ex]
    & \lesssim n^{\frac{1}{2}} e^{(\beta / 2)(1 - x_{m})} (1 - x_{m}^{2} + 4 n^{-1})^{\frac{1}{2}}  \pi(m)  \\[1ex]
    & \lesssim  n^{\frac{1}{2}} e^{\beta / 2} \pi(m) ,
  \end{split}
\end{equation}
where in the last line, we have used that \((1 - x_{m}^{2} + 4 n^{-1})^{\frac{1}{2}} \leq \sqrt{5}\) when \(n \geq 1\) and \(x_{m} \in [0,1]\).
This proves \cref{prop:cdf-upper-bound} for \(m \in R_{\mathrm{Right}}\) at all temperatures.

\subsubsection{Analysis for \texorpdfstring{\(R_{\mathrm{Left}}\)}{R_Left}}
\label{sec:cdf-bound-left}
Note that when \(\beta \leq 2\), \(x_{*} = 0\) so \(R_{\mathrm{Left}} = \emptyset\) unless \(\beta > 2\) and the claim is trivial.
When \(x_{m} \in R_{\mathrm{Left}}\), we should expect that \(\pi([0, x_{m}]) \leq \pi([x_{m}, 1])\) since \([0, x_{m}]\) does not contain \(x_{*}\).
Following this intuition, we prove that \(\pi([0, x_{m}]) \lesssim n^{\frac{1}{2}} \pi(m)\).

For the analysis of \(R_{\mathrm{Right}}\), used the fact that \(f_{\beta}\) is convex on \([x_{*}, 1]\) to get a lower bound from the tangent line at \(x_{m}\).
We aren't so lucky for \(R_{\mathrm{Left}}\) as \(f_{\beta}''(0) < 0\) and \(f_{\beta}''(x_{*}) > 0\) so \(f_{\beta}\) is neither convex nor concave on \([0, x_{*}]\).
Despite the loss of convexity, analysis on \(R_{\mathrm{Left}}\) is still tractable since the second derivative \(f_{\beta}''\) is strictly monotone.

In particular, since \(f_{\beta}''\) is strictly monotone, there exists a unique inflection point \(x_{\mathrm{infl}}\) so that \(f_{\beta}''(x_{\mathrm{infl}}) = 0\), with \(f_{\beta}''(x) < 0\) on \([0, x_{\mathrm{infl}})\) and \(f_{\beta}''(x) > 0\) on \((x_{\mathrm{infl}}, x_{*} ]\).
Since \(f_{\beta}''(x) < 0\) on \([0, x_{\mathrm{infl}}]\), the function \(f_{\beta}\) is concave on this interval.
Recalling that concave functions are lower bounded by their secant line we have for any \(x, x_{m} \in [0, x_{\mathrm{infl}}]\):
\begin{equation}
  f_{\beta}(x) \geq \frac{f_{\beta}(0) - f_{\beta}(x_{m})}{0 - x_{m}} (x - x_{m}) + f_{\beta}(x_{m}).
\end{equation}
A minor technical issue when using the secant line lower bound is that the slope of the tangent line tends to zero as \(m \to 0\) since
\begin{equation}
  \lim_{x_{m} \to 0}\frac{f_{\beta}(0) - f_{\beta}(x_{m})}{0 - x_{m}} = f_{\beta}'(0) = 0.
\end{equation}
We can avoid this issue by splitting \(R_{\mathrm{Left}}\) into three regions as follows:
\begin{equation}
  \begin{split}
    R_{\mathrm{Left}}^{(1)} & := [0, n^{-\frac{1}{2}}) \\[1ex]
    R_{\mathrm{Left}}^{(2)} & := [n^{-\frac{1}{2}}, x_{\mathrm{infl}}] \\[1ex]
    R_{\mathrm{Left}}^{(3)} & := (x_{\mathrm{infl}}, x_{*} - c_{\beta}]
  \end{split}
\end{equation}
On \(R_{\mathrm{Left}}^{(2)}\), we will show that \(|f_{\beta}'(x)| = \Omega(n^{-\frac{1}{2}})\) (see~\cref{lem:r-left-derivative-lower-bound}) which is large enough to guarantee the desired upper bound on \(\pi([0, x_{m}])\).
On \(R_{\mathrm{Left}}^{(3)}\), we'll use that the function \(f_{\beta}\) is convex so we can apply the same argument as used in \cref{sec:cdf-bound-right} to obtain the desired upper bound.
Finally, the region \(R_{\mathrm{Left}}^{(1)}\) becomes vanishingly small as \(n \to \infty\) so the cumulative distribution on the region can be upper bounded with a union bound.

To formalize this discussion, we state the following proposition:
\begin{proposition}
  For any \(\beta > 2\), we have the following upper bounds
  \begin{itemize}
  \item If \(x_{m} \in R_{\mathrm{Left}}^{(1)}\) then
    \begin{equation}
      \label{eq:r-left-bd-1} 
      \pi([0, x_{m}]) \lesssim \frac{n}{Z_{\beta}} e^{\beta / 2} (x_{m} + n^{-1})^{2} e^{- n f_{\beta}(x_{m})};
    \end{equation}
  \item If \(x_{m} \in R_{\mathrm{Left}}^{(2)}\) then
    \begin{equation}
      \label{eq:r-left-bd-2}
      \pi([ n^{-\frac{1}{2}}, x_{m}]) \lesssim \frac{n}{Z_{\beta}} e^{\beta / 2} (x_{m} + n^{-1})^{2} e^{- n f_{\beta}(x_{m})};
    \end{equation}
  \item If \(x_{m} \in R_{\mathrm{Left}}^{(3)}\) then
    \begin{equation}
      \label{eq:r-left-bd-3}
      \pi([x_{\mathrm{infl}}, x_{m}]) \lesssim \frac{n}{Z_{\beta}} e^{\beta / 2} (x_{m} + n^{-1})^{2} e^{- n f_{\beta}(x_{m})},
    \end{equation}
  \end{itemize}
  where \(\pi([0, x_{m}]) = \sum_{s = s_{\min}}^{m} \pi(s)\) and similarly for \(\pi([n^{-\frac{1}{2}}, x_{m}])\) and \(\pi([x_{\mathrm{infl}}, x_{m}])\).
\end{proposition}
The upper bounds in this proposition imply that when \(x_{m} \in R_{\mathrm{Left}}\) then \(\pi([0, x_{m}]) \lesssim n^{\frac{1}{2}} e^{\beta / 2} \pi(m)\). 
To see why, let's consider the case where \(x_{m} \in R_{\mathrm{Left}}^{(2)}\); the other cases follow by similar reasoning.
When \(x_{m} \in R_{\mathrm{Left}}^{(2)}\) we have that
\begin{equation}
  \pi([0, x_{m}]) = \pi([0, n^{-\frac{1}{2}}]) + \pi([n^{-\frac{1}{2}}, x_{m}]) 
\end{equation}
Applying the bounds for \(R_{\mathrm{Left}}^{(1)}\) and \(R_{\mathrm{Left}}^{(2)}\) for the first and second terms respectively, we have that
\begin{equation}
  \begin{split}
    \pi([0, x_{m}])
    & \lesssim \frac{n}{Z_{\beta}} e^{\beta / 2} (n^{-\frac{1}{2}} + n^{-1})^{2} e^{-n f_{\beta}(n^{-\frac{1}{2}})} + \frac{n}{Z_{\beta}} e^{\beta / 2} (x_{m} + n^{-1})^{2} e^{-n f_{\beta}(x_{m})} \\[1ex]
    & \lesssim \frac{1}{Z_{\beta}} e^{\beta / 2} e^{- n f_{\beta}(n^{-\frac{1}{2}})} + \frac{n}{Z_{\beta}} e^{\beta / 2} (x_{m} + n^{-1})^{2} e^{-n f_{\beta}(x_{m})},
  \end{split}
\end{equation}
which implies the claim since \(f_{\beta}\) is decreasing on \([0, x_{*}]\).

\paragraph{Bound for \texorpdfstring{\(R_{\mathrm{Left}}^{(1)}\)}{R_Left 1} \texorpdfstring{(Proof of \cref{eq:r-left-bd-1})}{ }}
From the estimates for the stationary distribution \cref{lem:sharp-pi-s}, we have
\begin{equation}
  \begin{split}
    \pi([0, x_{m}]) 
    & \lesssim \frac{n^{\frac{1}{2}}}{Z_{\beta}} \sum_{0 \leq x_{s} \leq x_{m}}^{} \frac{(x_{s} + n^{-1})^{2}}{(1 - x_{s}^{2} + 4 n^{-1})^{\frac{1}{2}}} e^{(\beta / 2) x_{s}} e^{- n f_{\beta}(x_{s})} \\
    & \lesssim \frac{n^{\frac{1}{2}}}{Z_{\beta}} \frac{(x_{m} + n^{-1})^{2}}{(1 - x_{m}^{2} + 4 n^{-1})^{\frac{1}{2}}} e^{(\beta / 2) x_{m}} e^{- n f_{\beta}(x_{m})}  \sum_{0 \leq x_{s} \leq x_{m}}^{} 1
  \end{split}
\end{equation}
where in the second line, we have used that the map
\begin{equation}
  x \mapsto \frac{(x + n^{-1})^{2}}{(1 - x^{2} + 4 n^{-1})^{\frac{1}{2}}} e^{(\beta / 2) x} e^{- n f_{\beta}(x)} 
\end{equation}
is increasing on \(R_{\mathrm{Left}}\).
Since \(x_{m} \leq n^{-\frac{1}{2}}\) in normalized variables, using \(x_{m} = (2m) /n\) it follows that \(m \leq \frac{1}{2} n^{\frac{1}{2}}\).
Therefore, the number of terms in the sum is \(O(n^{\frac{1}{2}})\) and \(\pi([0, x_{m}]) \lesssim n^{\frac{1}{2}} e^{\beta / 2} \pi(m)\) 
as was claimed.

\paragraph{Bound for \texorpdfstring{\(R_{\mathrm{Left}}^{(2)}\)}{R_Left 2}  \texorpdfstring{(Proof of \cref{eq:r-left-bd-2})}{ }}
We start by showing that on \(R_{\mathrm{Left}}^{(2)}\), \(|f_{\beta}'(x)| = \Omega(n^{-\frac{1}{2}})\).
\begin{lemma}
  \label{lem:r-left-derivative-lower-bound}
  If \(x \in R_{\mathrm{Left}}^{(2)}\), then for all \(n\) sufficiently large
  \begin{equation}
    | f_{\beta}'(x) | \geq \frac{1}{2} | f_{\beta}''(0) | n^{-\frac{1}{2}}.
  \end{equation}
\end{lemma}
\begin{proof}
  Since \(f_{\beta}'(0) = 0\) and \(f_{\beta}''(x) < 0\) on \(R_{\mathrm{Left}}^{(2)}\) we have
  \begin{equation}
    \begin{split}
      f_{\beta}'(x)
       = \int_{0}^{x} f_{\beta}''(y) d y 
       = \int_{0}^{n^{-\frac{1}{2}}} f_{\beta}''(y) dy + \int_{n^{-\frac{1}{2}}}^{x} f_{\beta}''(y) dy 
       \leq \int_{0}^{n^{-\frac{1}{2}}} f_{\beta}''(y) dy.
    \end{split}
  \end{equation}
  To complete the proof, recall that \(f_{\beta}''(x)\) is continuous near zero so by choosing \(n\) sufficiently large, we can ensure \(f_{\beta}''(y) \leq \frac{1}{2} f_{\beta}''(0)\) on \([0, n^{-\frac{1}{2}}]\).
  Therefore, for all \(n\) sufficiently large,
  \begin{equation}
    f_{\beta}'(x) \leq \frac{1}{2}  f_{\beta}''(0) n^{-\frac{1}{2}}.
  \end{equation}
  Since \(f_{\beta}''(0) < 0\) and \(f_{\beta}'(x) < 0\) on \(R_{\mathrm{Left}}^{(2)}\) the result follows.
\end{proof}
By ensuring that \(f_{\beta}'\) is not too small, we can get the desired upper bound for \(\pi([n^{-\frac{1}{2}}, x_{m}])\) for any \(x_{m} \in R_{\mathrm{Left}}^{(2)}\).
From \cref{lem:sharp-pi-s}, we have the upper bound
\begin{equation}
  \begin{split}
    \pi([n^{-\frac{1}{2}}, x_{m}])
    & \lesssim \frac{n^{\frac{1}{2}}}{Z_{\beta}} \sum_{n^{-\frac{1}{2}} \leq x_{s} \leq x_{m}}^{} \frac{(x_{s} + n^{-1})^{2}}{(1 - x_{s}^{2} + 4 n^{-1})^{\frac{1}{2}}} e^{(\beta / 2) x_{s}} e^{- n f_{\beta}(x_{s})} \\
    & \lesssim \frac{n^{\frac{1}{2}}}{Z_{\beta}} \frac{(x_{m} + n^{-1})^{2}}{(1 - x_{m}^{2} + 4 n^{-1})^{\frac{1}{2}}} e^{(\beta / 2) x_{m}} \sum_{n^{-\frac{1}{2}} \leq x_{s} \leq x_{m}}^{}  e^{- n f_{\beta}(x_{s})}.
  \end{split}
\end{equation}
Therefore, to finish the bound in this case we only need to show that
\begin{equation}
  \sum_{n^{-\frac{1}{2}} \leq x_{s} \leq x_{m}}^{} e^{- n f_{\beta}(x_{s})} \lesssim n^{\frac{1}{2}} e^{- n f_{\beta}(x_{m})}.
\end{equation}
Similarly to before, due to concavity, for any \(y \in [ n^{-\frac{1}{2}}, x_{m}]\) we have the lower bound
\begin{equation}
  f_{\beta}(y) \geq \frac{f_{\beta}(n^{-\frac{1}{2}}) - f_{\beta}(x_{m})}{n^{-\frac{1}{2}} - x_{m}} (y - x_{m}) + f_{\beta}(x_{m}) = f_{\beta}'(c) (y - x_{m}) + f_{\beta}(x_{m}),
\end{equation} where \(c \in [n^{-\frac{1}{2}}, x_{m}]\) is a point whose existence is granted by the mean value theorem.
But due to \cref{lem:r-left-derivative-lower-bound}, \(f_{\beta}'(c) \leq - \frac{1}{2} n^{-\frac{1}{2}} |f_{\beta}''(0)|\).
Hence, for \(y \leq x_{m}\) we have the lower bound
\begin{equation}
  f_{\beta}(y) \geq - \frac{1}{2} |f_{\beta}''(0) | n^{-\frac{1}{2}} (y - x_{m}) + f_{\beta}(x_{m}).
\end{equation}
Therefore,
\begin{equation}
\sum_{n^{-\frac{1}{2}} \leq x_{s} \leq x_{m}}^{} e^{- n f_{\beta}(x_{s})}
  \leq e^{-n f_{\beta}(x_{m})} \sum_{n^{-\frac{1}{2}} \leq x_{s} \leq x_{m}}^{} \exp\left(\frac{1}{2} n^{\frac{1}{2}} |f_{\beta}''(0)| (x_{s} - x_{m})\right).
\end{equation}
Since the summand is monotone increasing in \(s\) and positive, by \cref{lem:integral-test} we conclude that
\begin{equation}
  \sum_{n^{-\frac{1}{2}} \leq x_{s} \leq x_{m}}^{} \exp\left(\frac{1}{2} n^{\frac{1}{2}} |f_{\beta}''(0)| (x_{s} - x_{m})\right) \leq \frac{n}{2} \int_{0}^{x_{m}} \exp\left(\frac{1}{2} n^{\frac{1}{2}} |f_{\beta}''(0)| (y - x_{m})\right) dy + 1 \lesssim n^{\frac{1}{2}},
\end{equation}
which implies that \(\pi([n^{-\frac{1}{2}}, x_{m}]) \lesssim n^{\frac{1}{2}} \pi(m)\) as we wanted to show.

\paragraph{Bound for \texorpdfstring{\(R_{\mathrm{Left}}^{(3)}\)}{R_Left 3}  \texorpdfstring{(Proof of \cref{eq:r-left-bd-3})}{ }}
Since \(f_{\beta}\) is convex on this region, the same analysis as used for \(m \in R_{\mathrm{Right}}\) with the appropriate modifications proves \(\pi([x_{\mathrm{infl}}, x_{m}]) \lesssim n^{\frac{1}{2}} e^{\beta / 2} \pi(m)\) as was claimed.

\subsubsection{Analysis for \texorpdfstring{\(R_{\mathrm{Laplace}}\)}{R_Laplace}}
\label{sec:cdf-bound-laplace}
When \(x_{m} \in R_{\mathrm{Laplace}}\) we have two cases \(x_{m} \leq x_{*}\) and \(x_{m} > x_{*}\).
For simplicity, let us assume \(x_{m} > x_{*}\); the other case follows by an analogous argument.
We first consider \(\beta \neq 2\) then note the changes which must be made for \(\beta = 2\) in \cref{sec:cdf-bound-laplace-critical}.

Since the mass of the stationary distribution concentrates at \(x_{*}\), when \(x_{m} > x_{*}\), we should expect that \(\pi([x_{m}, 1]) \leq \pi([0, x_{m}])\).
Following this intuition, we show that \(\pi([x_{m},1]) \lesssim n^{\frac{1}{2}} \pi(m)\).
We begin by noting that
\begin{equation}
  \pi([x_{m}, 1]) = \pi([x_{m}, x_{*} + c_{\beta}]) + \pi([x_{*} + c_{\beta}, 1])
\end{equation}
By the argument given in \cref{sec:cdf-bound-right}, we know that \(\pi([x_{*} + c_{\beta}, 1]) \lesssim  n^{\frac{1}{2}} \pi(x_{*} + c_{\beta})\) and so
\begin{equation}
  \pi([x_{m}, 1]) \lesssim \pi([x_{m}, x_{*} + c_{\beta}]) + n^{\frac{1}{2}} \pi(x_{*} + c_{\beta}).
\end{equation}
Therefore, to complete the proof it suffices to show \(\pi([x_{m}, x_{*} + c_{\beta}]) \lesssim n^{\frac{1}{2}} \pi(m)\).

When \(\beta \neq 2\), by Taylor's theorem, for every \(x \in R_{\mathrm{Laplace}}\) there exists \(\xi\) between \(x\) and \(x_m\) such that
\begin{equation}
  \begin{split}
    f_{\beta}(x)
    & = f_{\beta}(x_{m}) + f_{\beta}'(x_{m}) (x - x_{m}) + \frac{1}{2} f_{\beta}''(\xi) (x - x_{m})^{2}  \\[1ex]
    & \geq f_{\beta}(x_{m}) + f_{\beta}'(x_{m}) (x - x_{m})  + \frac{1}{4} f_{\beta}''(x_{*}) (x - x_{m})^{2},  \\
  \end{split}
\end{equation}
where in the second line we have used that \(f_{\beta}''(\xi) \geq \frac{1}{2} f_{\beta}''(x_{*})\) on \([x_{*} - c_{\beta}, x_{*} + c_{\beta}] \cap [0,1]\).
When \(x \geq x_{m}\), we can use the fact that \(f_{\beta}'(x_{m}) > 0\) to further lower bound
\begin{equation}
  f_{\beta}(x) \geq f_{\beta}(x_{m}) + \frac{1}{4} f_{\beta}''(x_{*}) (x - x_{m})^{2} \quad \text{when} \quad x \geq x_{m}.
\end{equation}
Therefore, using \cref{lem:sharp-pi-s}, we have
\begin{equation}
  \label{eq:laplace-upper-bd-1}
  \begin{split}
    \pi([x_{m}, x_{*} + c_{\beta}])
    & \lesssim \frac{n^{\frac{1}{2}}}{Z_{\beta}} \sum_{x_{m} \leq x_{s} \leq x_{*} + c_{\beta}}^{} \frac{(x_{s} + n^{-1})^{2}}{(1 - x_{s}^{2} + 4 n^{-1})^{\frac{1}{2}}} e^{(\beta / 2) x_{s}} e^{- n f_{\beta}(x_{s})} \\[1ex]
    & \lesssim \frac{n^{\frac{1}{2}}}{Z_{\beta}} e^{\beta / 2} e^{-n f_{\beta}(x_{m})} \sum_{x_{m} \leq x_{s} \leq x_{*} + c_{\beta}}^{} \frac{(x_{s} + n^{-1})^{2}}{(1 - x_{s}^{2} + 4 n^{-1})^{\frac{1}{2}}} \exp\left(- \frac{n}{4} f_{\beta}''(x_{*}) (x_{s} - x_{m})^{2} \right) \\[1ex]
    & \lesssim \frac{n^{\frac{1}{2}}}{Z_{\beta}} e^{\beta} e^{-n f_{\beta}(x_{m})} \sum_{x_{m} \leq x_{s} \leq x_{*} + c_{\beta}}^{} (x_{s} + n^{-1})^{2} \exp\left(- \frac{n}{4} f_{\beta}''(x_{*}) (x_{s} - x_{m})^{2} \right), \\[1ex]
  \end{split}
\end{equation}
where in the third line we have used that, for \(0 \leq x_{s} \leq \frac{1}{2} (1 + x_{*})\), the following holds:
\begin{equation}
  \frac{1}{1 - x_{s}^{2} + 4n^{-1}} \leq \frac{1}{1 - x_{s}^{2}} = \frac{1}{(1 + x_{s})(1 - x_{s})} \leq \frac{1}{1 - \frac{1}{2} (1 + x_{*})} \leq 2 e^{\beta}.
\end{equation}
The last inequality is since \(x_{*} \leq 1 - e^{-\beta}\) from~\cref{lem:properties-f_beta}\ref{item:minimum_bound}.

Using calculus, it's easily checked that summand appearing in the last line of~\cref{eq:laplace-upper-bd-1} is bounded for all \(n\) and its derivative changes sign at most once.
Therefore, by \cref{lem:integral-test} we conclude that
\begin{equation}
  \begin{split}
    & \sum_{x_{m} \leq x_{s} \leq x_{*} + c_{\beta}}^{} (x_{s} + n^{-1})^{2} \exp\left(- \frac{n}{4} f_{\beta}''(x_{*}) (x_{s} - x_{m})^{2} \right) \\
    & \hspace{2em} \lesssim n \int_{x_{m}}^{x_{*} + c_{\beta}} (y + n^{-1})^{2} \exp\left(- \frac{n}{4} f_{\beta}''(x_{*}) (y - x_{m})^{2} \right) + 1 \\[1ex]
    & \hspace{2em} \lesssim n^{\frac{1}{2}} (x_{m} + n^{-1})^{2} (1 + O(n^{-\frac{1}{2}})),
  \end{split}
\end{equation}
where in the last line we have used Laplace's method \cite[Chapter II, Theorem 1]{wong_asymptotic_2001}.
Therefore, for all \(n\) sufficiently large,
\begin{equation}
  \pi([x_{m}, x_{*} + c_{\beta}]) \lesssim \frac{n}{Z_{\beta}} e^{\beta} e^{- n f_{\beta}(x_{m})}   (x_{m} + n^{-1})^{2} \lesssim n^{\frac{1}{2}} e^{\beta}\pi(m),
\end{equation}
which by the previous discussion proves the claim.

\subsubsection{Modification for the Critical Temperature}
\label{sec:cdf-bound-laplace-critical}
When \(\beta = 2\), the probability mass concentrates at \(x_{*} = 0\) so we are interested in upper bounding \(\pi([x_{m}, 1])\).
Similar to before, we note that
\begin{equation}
  \pi([x_{m}, 1]) = \pi([x_{m}, c_{\beta}]) + \pi([ c_{\beta}, 1]).
\end{equation}
The term \(\pi([c_{\beta}, 1])\) can be upper bounded using the analysis for \(R_{\mathrm{Right}}\) so we only need to control \(\pi([x_{m}, c_{\beta}])\).
Using the fourth-order Taylor expansion, for every \(x \in [x_{m},c_{\beta}]\) there exists \(\xi\) between \(x\) and \(x_{m}\) such that
\begin{equation}
  f_{\beta}(x)
    = f_{\beta}(x_{m}) + f_{\beta}'(x_{m}) (x - x_{m}) + \frac{1}{2} f_{\beta}''(x_{m}) (x - x_{m})^{2} + \frac{1}{3!} f_{\beta}'''(x_{m}) (x - x_{m})^{3} + \frac{1}{4!} f_{\beta}^{(4)}(\xi) (x - x_{m})^{4}.
\end{equation}
When \(x > x_{m} \geq x_{*}\), since \(f_{\beta}', f_{\beta}'', f_{\beta}'''\) are all positive we have
\begin{equation}
  \begin{split}
    f_{\beta}(x)
    & \geq f_{\beta}(x_{m}) + \frac{1}{4!} f_{\beta}^{(4)}(\xi) (x - x_{m})^{4} \\
    & \geq f_{\beta}(x_{m}) + \frac{1}{2 \cdot 4!} f_{\beta}^{(4)}(x_{*}) (x - x_{m})^{4},
  \end{split}
\end{equation}
where the second line is due the definition of \(c_{\beta}\).
Following the same steps using \cref{lem:integral-test} and Laplace's method, we can conclude that
\begin{equation}
  \pi([x_{m}, 1]) \lesssim e^{\beta} n^{\frac{3}{4}} \pi(m),
\end{equation}
as was claimed.

\subsection{Proof of \texorpdfstring{\cref{lem:properties-f_beta}}{Properties of f beta}}
\label{sec:properties-f_beta}
For this proof, we begin by recording the first four derivatives of \(f_{\beta}(x)\):
\begin{equation}
\renewcommand{\arraystretch}{1.3}
\begin{array}{lcl}
    \displaystyle f_{\beta}'(x) = - \frac{\beta}{2} x + \atanh{(x)}, && \displaystyle  f_{\beta}''(x) = -\frac{\beta}{2} + (1 - x^{2})^{-1}, \\[2ex]
    \displaystyle  f_{\beta}'''(x)  = 2x (1 - x^{2})^{-2}, &&  \displaystyle  f_{\beta}^{(4)}(x) = 2 (3x^2 + 1) (1 - x^{2})^{-3}.
\end{array}
\end{equation}
From these formulas, one can easily verify that
\begin{equation}
    f_{\beta}'(0)  = 0, \qquad
    f_{\beta}''(0) = \frac{1}{2} (2 - \beta), \qquad
    f_{\beta}'''(0) = 0, \qquad
    f_{\beta}^{(4)}(0) = 2.
\end{equation}

\noindent \textit{Proofs of (i), (ii)}. Immediate from the functional forms of \(f_{\beta}(x)\) and \(f_{\beta}^{(4)}(x)\). \\

\noindent \textit{Proof of (iii)}.
First, note that since \(f_{\beta}\) is continuous, it must have a minimum on \([0,1]\).
Since \(\lim_{x \to 1} f_{\beta}'(x) = + \infty\), therefore \(x = 1\) cannot be a minimizer.
Since \(f_{\beta}(x)\) is continuously differentiable on \([0,1)\): either the minimum occurs at the endpoint \(x_{*} = 0\) or the minimum occurs at a critical point.
From the calculation above, \(f'_{\beta}(0) = 0\) so \(x = 0\) is always a critical point however it may be a local minimum or maximum depending on \(\beta\). \\

If \(\beta \leq 2\), then \(f_{\beta}''(x_{*}) \geq 0\) and so \(x_{*} = 0\) is a local minimum.
Since \(f_{\beta}''(x)\) is strictly increasing and \(f_{\beta}''(0) \geq 0\), it follows that \(f_{\beta}'(x)\) is strictly positive for \(x > 0\) and so for \(\beta \leq 2\), \(x_{*} = 0\) is the unique local minimum.
If \(\beta > 2\), then \(f_{\beta}''(0) < 0\) so \(x = 0\) is a local maximum and the minimum of \(f_{\beta}\) lies in \((0,1)\).
Since \(f_{\beta}''(x)\) is strictly increasing it is injective, and so \(f_{\beta}'(x)\) strongly convex on \([0,1]\).
Since strongly convex functions can only attain the same value twice on a closed interval, it follows that there is only one critical point on \((0,1)\). \\

\noindent \textit{Proof of (iv)}. By the proof of (iii), we know that when \(0 \leq \beta \leq 2\), \(x_{*} = 0\) so the claim is obviously true in this case.
To see that \(x_{*} \leq 1 - e^{-\beta}\) for \(\beta > 2\), we can plug \(1 - e^{-\beta}\) into the expression for \(f_{\beta}'\) and find that
\begin{equation}
  f_{\beta}'(1 - e^{- \beta}) = -\frac{\beta}{2} (1 - e^{-\beta}) + \frac{1}{2} \ln(2 + e^{-\beta}) + \frac{\beta}{2} \geq e^{-\beta} + \frac{1}{2} \ln(2) > 0.
\end{equation}
Similarly, plugging in \(1 - e^{-\beta / 2 + 1}\) one can verify \(f_{\beta}'(1 - e^{-\beta / 2 + 1}) < 0\) for all \(\beta > 2\).
Therefore, by intermediate value theorem for all \(\beta > 2\), \(x_{*} \in (1 - e^{-\beta / 2 + 1}, 1 - e^{-\beta}) \subseteq [0, 1 - e^{-\beta})\). \\

\noindent \textit{Proof of second derivative at the minimum}.
When \(\beta < 2\), we know that \(x_{*} = 0\) and \(f_{\beta}''(0) = \frac{1}{2} (2 - \beta)\) so the claim is true.
For \(\beta = 2\), we know that \(x_{*} = 0\) and one can compute \(f_{\beta}'''(x_{*}) = 0\) and \(f_{\beta}^{(4)}(x_{*}) = 2\).

Since we don't have an explicit expression for \(x_{*}\) when \(\beta > 2\), we argue the result by continuity.
When \(\beta > 2\), \(f_{\beta}''(0) = \frac{1}{2} (2 - \beta) < 0\) and so by continuity there exists a \(\delta > 0\) so that \(f_{\beta}''(x) < 0\) on \([0, \delta]\).
By fundamental theorem of calculus,
\begin{equation}
  0 = f_{\beta}'(x_{*}) - f_{\beta}'(0) = \int_{0}^{x_{*}} f_{\beta}''(y) dy = \int_{0}^{\delta} f_{\beta}''(y) d y + \int_{\delta}^{x_{*}} f_{\beta}''(y) dy.
\end{equation}
Since the integral over \([0, \delta]\) is strictly negative, for the equality to hold it must be that the integral over \([\delta, x_{*}]\) is strictly positive.
Therefore, there exists a point \(y \in [\delta, x_{*}]\) so that \(f_{\beta}''(y) > 0\).
Since \(f_{\beta}''\) is strictly increasing, it follows that \(f_{\beta}''(x_{*}) > 0\). \\

\section{Group Mixer Comparison Arguments}\label{secapp:the simulation argument}
In this section, we prove \cref{prop:removing_NN,prop:removing_su(2)}, which reason that the 
Davies generators $\lind_{\Sn}$ and $\lind_{\su(2)}$, which contain non-local updates, can nonetheless be ``simulated'' using only single-site updates $\CL_\loc$. In some sense, these statements resemble path-comparison arguments (see, e.g. \cite{Diaconis1993COMPARISONTF}) in that they relate the spectral gaps of two detailed-balanced Markov chains defined on different transition matrices. We find these statements may be of independent interest.   

\subsection{Single-site Paulis Simulate \texorpdfstring{$\calL_{\su(2)}$}{Lsu(2)}}\label{sec:removing_su(2)}

In this section we prove \cref{prop:removing_su(2)}, which relates the Dirichlet forms of the $\calL_{\su(2)}$ generator to that of $\calL_\loc$, up to a system-size dependent factor of $n^{-1}$.
 The key ingredient is the following lemma, which may be thought of as a differential version of the unitary freedom of choice of Kraus operators for completely-positive maps. We remark that similar computations can be found in e.g. \cite{BCL2024, paezvelasco2025efficientsimplegibbsstate}.

\begin{lemma}[Unitary freedom of jump operators]\label{lem:unitary_freedom_of_jumps}
Suppose we have two collections of jump operators $\{A^a\}, \{B^b\}$ which are related by a unitary matrix $U:=[U_{ab}]_{a,b}$:
\begin{align}
A^a = \sum_{b} U_{ab} B^b.
\end{align}
Then, the Davies generators defined by these sets are equal $\lind_{\{A^a\}} = \lind_{\{B^b\}}.$

\end{lemma}
\begin{proof}
    
    The jump operators in the frequency basis  \cref{eq:bohr_frequency_decomposition} are trivially also related by the same unitary transformation:
    \begin{align}
        A^a(\omega) = \sum_\lambda \Pi_{\lambda+\omega} A^a \Pi_\lambda = \sum_\lambda \Pi_{\lambda+\omega} \sum_b U_{ab} B^b  \Pi_\lambda = \sum_b U_{ab} B^b(\omega).
    \end{align}
    Referring to the Lindbladian form in \cref{def:davies_generator}, we compute directly on the completely-positive term: for any $X\in \calB(\calH)$,
    \begin{align}
        \sum_{a} A^a(\omega)^\dagger X A^a(\omega) &= \sum_{a,b,b'} (U_{ab}B^b(\omega))^\dagger X(U_{ab'}B^{b'}(\omega)) = \sum_{a,b,b'} \overline{U_{ab}} U_{ab'} B^{b}(\omega)^\dagger X B^{b'}(\omega) \nonumber \\
        &= \sum_{b,b'}\indicator{b=b'} B^{b'}(\omega)^\dagger X B^b(\omega) 
        = \sum_b B^b(\omega)^\dagger X B^b(\omega) . 
    \end{align} 
    The computation on the anticommutator is similar, and we arrive at the desired result.
\end{proof}

The proof of~\cref{prop:removing_su(2)} follows quickly by a judicious choice of unitary.

\begin{proof}[Proof of \cref{prop:removing_su(2)}]
Organize the local Paulis by $\calP = \{S_1^X, \dots, S_n^X, S_1^Y, \dots, S_n^Y, S_1^Z, \dots, S_n^Z\}$ and choose a block-diagonal unitary matrix of the form
\begin{equation}
    U = \begin{pmatrix}
       \wt{U} & 0 & 0\\
       0 & \wt{U} & 0 \\
       0 & 0 & \wt{U} 
    \end{pmatrix}, \qquad \wt{U} = \frac{1}{\sqrt{n}}\begin{pmatrix}
        1 & 1 & \dots & 1 \\
        1 & * & \dots & * \\
        \vdots & * & \ddots & * \\
        1 & * & \dots & * 
    \end{pmatrix},
\end{equation} where each block $\wt{U}$ is also unitary. One such example is given by choosing $\wt{U}$ to be the discrete Fourier transform. We then obtain the rotated set of jump operators $\{B^b\}:= \{\frac{1}{\sqrt{n}}S_{\tot}^X, \frac{1}{\sqrt{n}}S_{\tot}^Y, \frac{1}{\sqrt{n}}S_{\tot}^Z, B_4, \dots, B_{3n}\}$. Given \cref{lem:unitary_freedom_of_jumps}, it follows that
\begin{equation}
    \calL_{\loc} = \sum_{i,\alpha} \calL_{S_i^{\alpha}} = \paran{\sum_{\alpha} \calL_{S_{\tot}^{\alpha}/\sqrt{n}} } + \paran{\sum_{b=4}^{3n} \calL_{B^b} }  = \frac{1}{n}\calL_{\su(2)} + \calL_{\{B_4,\dots , B_{3n}\}}.
\end{equation}
This establishes \cref{prop:removing_su(2)}. 
\end{proof}

\subsection{Single-site Paulis Simulate \texorpdfstring{$\calL_{\Sn}$}{LSn}}\label{sec:removing_NN}

In this section we prove \cref{prop:removing_NN}, which relates the Dirichlet forms of the $\calL_{\Sn}$ generator to that of $\CL_\loc$, up to a $\beta$ dependent constant. In this case, unitary freedom will not suffice, since $\calL_{\Sn}$ has 2-body jump operators while $\calL_{\loc}$ only has 1-local jump operators. To proceed, we reason that the convergence of complex-time evolution (\cref{lemma:heisenberg-cte}, below) is sufficient to relate the two. We will employ a comparison technique resembling tools from~\cite{bergamaschi2025quantumspinchainsthermalize, BCV25}.

\begin{lemma}
    [Multi-site to single-site commutator norm]\label{lemma:multi-single-K} Let $\beta\in \mathbb{R}^+$, and let $\rho\propto e^{-\beta H}$ be the Gibbs state of an arbitrary $n$-qubit Hamiltonian $H\in \calB(\calH)$. Suppose $A_1, A_2\in \calB(\calH)$ are operators whose complex-time evolution converges in operator norm:
    \begin{equation}
    b\in \{1, 2\}, \quad z\in \{-\beta, \beta\}: \qquad \norm{e^{-z H/4} A_b e^{z H/4}} \leq c_\beta,
    \end{equation}
    for some $c_\beta\in \mathbb{R}$. Then, for arbitrary $O\in \calB(\calH)$,
    \begin{equation}
         \|[A_1A_2, O]\|_{\rho}^2 \leq 2 c_\beta^2\cdot \bigg(\|[A_1, O]\|_{\rho}^2+\|[A_2, O]\|_{\rho}^2\bigg)
    \end{equation}
\end{lemma}

\begin{proof}
    By the Leibniz rule and the triangle inequality for the KMS norm:
    \begin{align}
        \|[A_1A_2, O]\|_{\rho} \leq  \|[A_1, O]A_2\|_{\rho}+ \|A_1[A_2, O]\|_{\rho}.
    \end{align}
    By the KMS H\"older inequality (see \cite[Lemma IX.4]{chen2025quantumMarkov}),
    \begin{equation}
            \|[A_1, O]A_2\|_{\rho} \leq   \|[A_1, O]\|_{\rho}\cdot \|\rho^{-1/4}A_2\rho^{1/4}\| \leq c_\beta\cdot \|[A_1, O]\|_{\rho}
    \end{equation}
        and similarly $\|A_1[A_2, O]\|_{\rho} \leq c_\beta\|[A_2, O]\|_{\rho}$. The inequality $(x+y)^2 \leq 2(x^2+y^2)$ then concludes the claim.
\end{proof}

We now have the tools to prove \cref{prop:removing_NN}, which we reproduce here for convenience. Henceforth, we fix our attention to the Heisenberg Hamiltonian.
\begin{proposition}[\cref{prop:removing_NN}, restated]
Let $H$ be the Heisenberg Hamiltonian, $\beta\in \mathbb{R}^+$ an inverse-temperature, $ \calL_{\Sn}$ the Davies generator defined by the set of transpositions $\{\mathsf{SWAP}_{uv}\}_{(u, v)\in G}$ on a $d$-regular graph $G$, and $ \calL_\loc$ the Davies generator defined by all single-site Pauli jumps. Then, for any $O\in \calB(\calH):$
\begin{align}
    \langle O, - \calL_{\Sn}(O)\rangle_\rho  \leq   c_1\cdot d e^{c_2\beta}\cdot  \langle O,-\calL_\loc(O)\rangle_\rho .
\end{align}
with $c_{1}, c_2\in \mathbb{R}^+$ universal constants. 
\end{proposition}

\begin{proof}
    First, let us recollect that for the Heisenberg model Hamiltonian, local Pauli operators are $\Omega$-band diagonal in the energy eigenbasis (\cref{cor:band-diagonal-local-operators}, as a consequence of the Wigner-Eckart theorem), for constant $\Omega = O(1)$, which in turn implies the convergence of complex time evolution by \cref{lemma:heisenberg-cte} (below). This enables the application of \cref{lemma:multi-single-K}.

    We can then start by examining the effect of one transposition $\mathsf{SWAP}_{uv}$. In the Heisenberg Hamiltonian, $[\mathsf{SWAP}_{uv}, H] = 0$, and thus the Bohr frequency decomposition of $\mathsf{SWAP}_{uv}$ contains only elements of frequency $\nu=0$. This allows us to simplify the Davies Dirichlet form (\cref{lem:dirichlet_divergence}) of an arbitrary operator $O = \sum_\nu O(\nu)$, by observing that components of $O$ with different Bohr frequencies are $\rho$-orthogonal:
    \begin{equation}
        \langle O, -\calL_{(uv)}(O)\rangle_{\rho} = h_0 \|[\mathsf{SWAP}_{uv}, O]\|_{\rho}^2 = h_0\sum_\nu \|[\mathsf{SWAP}_{uv}, O(\nu)]\|_{\rho}^2,
    \end{equation}
    where we recall that $h_\omega$ is the filter function in the Dirichlet form.  To proceed, we express the $\mathsf{SWAP}_{uv}$ in a Pauli basis expansion,   
\begin{align}
   \mathsf{SWAP}_{uv} =  \frac{1}{2} \idty_u\idty_{v} +2 \sum_{\alpha\in \{X, Y, Z\}} S_u^\alpha S_v^\alpha
\end{align}
    and invoke (in sequence) the triangle inequality, the Cauchy-Schwarz inequality, and the comparison argument in \cref{lemma:multi-single-K}:
    \begin{align}
        \|[\mathsf{SWAP}_{uv}, O(\nu)]\|_{\rho}^2\leq 4\cdot 3\cdot 2c_\beta^2 \sum_\alpha \bigg(\|[S^\alpha_{u}, O(\nu)]\|_{\rho}^2+\|[S^\alpha_{v}, O(\nu)]\|_{\rho}^2\bigg),
    \end{align}
    with $c_\beta$ the constant from \cref{lemma:heisenberg-cte}. Finally, it only remains to write the norms of commutators of $O(\nu)$ with single-site Paulis, in terms of the Dirichlet form of the single-site Pauli: 
    \begin{align}
       \sum_\nu \|[S^\alpha_{u}, O(\nu)]\|_{\rho}^2 = \sum_{\nu, \omega} \|[S^\alpha_{u}(\omega), O(\nu)]\|_{\rho}^2 = \sum_{\omega} \|[S^\alpha_{u}(\omega), O]\|_{\rho}^2 \leq e^{\beta \Omega/2}\cdot \sum_{\omega} h_\omega\cdot \|[S^\alpha_{u}(\omega), O]\|_{\rho}^2.
    \end{align}
    In the first equality, we use that $[S^\alpha_u(\omega),O(\nu)]$ has Bohr frequency $\omega+\nu$, hence the sum over $\omega$ is orthogonal for fixed $\nu$. The second equality uses the same orthogonality in the sum over $\nu$ for fixed $\omega$. Finally, $S^\alpha_{u}(\omega)=0$ for $|\omega|>\Omega = O(1)$ by \cref{cor:band-diagonal-local-operators}, so under the Metropolis weight $h_\omega=e^{-\beta|\omega|/2}$ we have $1\leq e^{\beta\Omega/2} h_\omega$ on the support of the sum.

    To conclude, we recall the collection of jump operators in the generator $\CL_{\Sn}$ is the set of transpositions $\{\mathsf{SWAP}_{uv}\}_{(u, v)\in T}$ over all the edges of a $d$-regular graph. Therefore, summing over all $(u,v) \in T$ gives
    \begin{align}
          \langle O, -\calL_{\Sn}(O)\rangle_{\rho}  &= \sum_{(u, v)\in T}   \langle O, -\calL_{(uv)}(O)\rangle_{\rho} \\
          &\leq c_1 e^{\beta c_2}\cdot  \sum_{(u,v)\in T}\sum_{\alpha\in\{X,Y,Z\}} \sum_{\omega} h_\omega\Big(\|[S^\alpha_{u}(\omega), O]\|_{\rho}^2+\|[S^\alpha_{v}(\omega), O]\|_{\rho}^2\Big) \\
          &\leq c_1 d\cdot e^{\beta c_2}\cdot  \sum_{u\in [n]}\sum_{\alpha\in\{X,Y,Z\}} \sum_{\omega} h_\omega\|[S^\alpha_{u}(\omega), O]\|_{\rho}^2 \\
          &= c_1 d\cdot e^{\beta c_2}\cdot\langle O, -\calL_\loc(O)\rangle_{\rho}
    \end{align}
    \noindent for an appropriate choice of explicit constants $c_1, c_2$.
\end{proof}

Finally we prove the deferred lemma on the convergence of complex-time evolution in the Heisenberg model. 
\begin{lemma}
    [Convergence of complex-time evolution in the Heisenberg model]\label{lemma:heisenberg-cte}
    For any $z\in \mathbb{C}$ and single-site operator $A$, the complex-time evolution of $A$ under the Heisenberg Hamiltonian satisfies the bound:
    \begin{equation}
        \|e^{zH}Ae^{-zH}\|\leq 3\cdot e^{|z|\cdot \Omega}
    \end{equation}
    with $\Omega$ the constant from \cref{cor:band-diagonal-local-operators}.
\end{lemma}
\begin{proof}
    We first decompose the single-site operator $A$ in terms of the change to the spin $s\in \calS$, using the Wigner-Eckart theorem (\cref{cor:band-diagonal-local-operators}):
\begin{align}
    e^{zH}Ae^{-zH} = \sum_{\delta\in \{-1, 0, 1\}} \sum_s e^{z(E_{s+\delta}-E_s)}\Pi_{s+\delta} A\Pi_s:= \sum_{\delta\in \{-1, 0, 1\}} A_{\delta}
\end{align}
For any state $\ket{\psi}$, we denote $\ket{\psi_s}:=\Pi_s \ket{\psi}$, such that the norm of $A_{\delta}$ can be computed via:
\begin{align}
    \|A_{\delta}\|^2 \leq \sup_\psi \| A_{\delta}\ket{\psi}\|^2&\leq \sup_\psi\bigg[ \sum_s |e^{z(E_{s+\delta}-E_s)}|^2\cdot  \|\Pi_{s+\delta} A\Pi_s\ket{\psi}\|^2\bigg] \\& \leq e^{2|z|\Omega}\cdot \|A\|^2\cdot  \sup_\psi\bigg[ \sum_s \|\ket{\psi_s}\|^2\bigg] = e^{2|z|\Omega}\cdot \|A\|^2,
\end{align}
where we used the fact $|E_{s+\delta}-E_s|\leq \Omega$ (\cref{cor:band-diagonal-local-operators}). The triangle inequality then concludes the proof. 
\end{proof}

\section{Monotonicity and the \texorpdfstring{$\calA^{(\ell)}$}{A(ell)} Spaces}
\label{section:monotonicity-and-al-gap}
In this section we provide proofs for the statements appearing in~\cref{sec:high temp monotonicity}. The main content of this section is a bound on the minimum eigenvalue of $-\CL_\loc$ at high temperatures, when restricted to elements of $\comm(\Sn)$ orthogonal to $\calA^{(0)}$. 

\begin{proposition}[\cref{thm:gap-Al}, restated]
    Fix $\beta < 2$. Then, the minimum eigenvalue of $\CL_\loc$ when restricted to $\calA^{(\ell)}$ is lower bounded by a constant independent of system size and the integer $\ell\geq 1$. That is, there exists a universal constant $c\in \mathbb{R}^+$ such that:
    \begin{equation}
      \forall\ell\geq 1:\quad  \min_{\substack{O \in \calA^{(\ell)} \\[.2ex] O \neq 0}}  \frac{\langle O, -\CL_\loc(O)\rangle_\rho}{\|O\|_\rho^2} \geq c\cdot (2-\beta).\label{eq:gap-Al}
    \end{equation}
\end{proposition}

The proof of \cref{thm:gap-Al} hinges on three steps. 
\begin{enumerate}
    \item In \cref{section:al-decomposition}, we give an explicit basis of operators $T_{s,\ell,q}$ for the spaces $\calA^{(\ell)}$, written in terms of the eigenbasis $\ket{s,m,r}$ of \cref{eq:smr eigenbasis}. 
    \item In \cref{section:monotonicity-Tlm}, we prove a key monotonicity statement (\cref{lemma:monotone-gaps}) which grants that spectra of $\CL_\loc$ when restricted to $\calA^{(\ell)}$ are controlled by that of $\calA^{(1)}$.
    \item In \cref{section:chain-T10}, we prove a lower bound on the Dirichlet form of $\CL_\loc$ restricted to $\calA^{(1)}$ at high temperatures (\cref{prop:calA10-gap}).
\end{enumerate}
The combination of \cref{lemma:monotone-gaps} and \cref{prop:calA10-gap} then concludes the proof of \cref{thm:gap-Al}.

\subsection{A Decomposition of \texorpdfstring{$\calA^{(\ell)}$}{A(ell)}}\label{section:al-decomposition}

We refer the reader to \cref{eq:smr eigenbasis} for a description of the  eigenbasis $\ket{s,m,r}$ of $\calH$ which simultaneously diagonalizes the operators $S_{\tot}^Z$ and $\S^2_{\tot}$. Referring back to~\cref{eq:continued decomp A ell q}, the spaces $ \calA^{(\ell)}$ admit the decomposition:
\begin{equation}
    \calA^{(\ell)} = \bigoplus_{q=-\ell}^\ell \calA^{(\ell,q)}.
\end{equation}
As discussed in the main text (\cref{obs:L preserves l and q eigenvalues}), since $\calL_{\loc}$ is an intertwiner for $\SU(2)$, each $\calA^{(\ell,q)}$ is an invariant subspace of $\calL_{\loc}$.\footnote{The same holds for the Pauli twirl $\calT$, as it is also an intertwiner.}
We dedicate this subsection to a description of $\calA^{(\ell,q)}$, in terms of an explicit choice of basis $\{T_{s, \ell, q}\}$. 

\begin{definition}\label{def:CG-basis}
    We define a collection of operators 
    \begin{equation}
        \bigg\{T_{s, \ell, q}\in \calB(\calH): \quad s\in \calS, \quad  2s\geq \ell, \quad -\ell\leq q\leq \ell\bigg\},
    \end{equation}
    whose matrix entries in the $|s, m, r\rangle$ basis are defined by the following scaled Clebsch-Gordan coefficients:
    \begin{equation}
        \bra{ s_1, m_1, r_1} T_{s, \ell, q}\ket{ s_2, m_2, r_2} = a_s \cdot  \delta_{s_1=s_2=s}\cdot \delta_{r_1=r_2}  \cdot (-1)^{m_1}\cdot C^{\ell, q}_{s, m_1, s, -m_2}.
    \end{equation} where $a_s = \paran{\dim(W_{Q(s)})}^{-1/2}$ is a normalization factor for the dimension of the degeneracy of the spin-$s$ irrep appearing in~\cref{thm:qubit schur-weyl}.
\end{definition}

\begin{lemma}\label{lem:CG-basis}
  The operators $T_{s,\ell,q}$ in~\cref{def:CG-basis} are spherical tensor operators in the sense of~\cref{def:spherical tensor operators}.
\end{lemma}

\begin{proof} [Proof of \cref{lem:CG-basis}]
We need to show that the collection of $T_{s,\ell,q}$ forms a simultaneous eigenbasis of $\ad_{S_{\tot}^Z}$, left and right multiplication by $\S_{\tot}^2$, and $-2\calL_{\su(2)}$.\footnote{KMS orthogonality follows from recalling that this representation of $\SU(2)$ is KMS unitary and so these superoperators are all either KMS Hermitian or KMS skew-Hermitian.} We will do this in sequence, throughout using properties of the Clebsch-Gordan coefficients later presented in \cref{section:wigner-properties}.
\begin{enumerate}
    \item Since $C^{\ell, q}_{s, m_1, s, -m_2}=0$ unless $m_1-m_2=q$ (\cref{fact:3j-properties}.1), we have \begin{equation} \label{eq:Tslq eigenvec of adSZ}
        \ad_{S_{\tot}^Z}(T_{s,\ell,q}) = [S^Z_{\tot},T_{s, \ell, q}]=q\cdot T_{s, \ell, q}. 
    \end{equation}
    \item $T_{s,\ell,q}$ is diagonal with respect to $s$ and so we freely have $T_{s,\ell,q} = T_{s,\ell,q} \Pi_s = \Pi_s T_{s,\ell,q}$. This implies that $T_{s,\ell,q}$ is a simultaneous eigenvector of left and right multiplication by $\S_\tot^2$ with eigenvalue $s(s+1)$.
    
    \item To show that $T_{s,\ell,q}$ is an eigenvector of $-2\calL_{\su(2)}$, we start by explicitly computing the matrix entries of $[S^\pm_\tot, T_{s, \ell, q}]$ by applying the recursion relation from \cref{fact:recursions} for the Clebsch-Gordan coefficients to their definition:
\end{enumerate}
   \begin{align}
        &\langle s, m_1, r|[S^\pm_\tot, T_{s, \ell, q}]|s, m_2, r\rangle \begin{aligned}[t] 
        &=\sqrt{s(s+1)-m_1(m_1\mp 1)} \langle s, m_1\mp 1, r|T_{s, \ell, q}|s, m_2, r\rangle \\
        &\quad -\sqrt{s(s+1)-m_2(m_2\pm 1)} \langle s, m_1, r| T_{s, \ell, q}|s, m_2\pm 1, r\rangle \\
        &= a_s(-1)^{m_1+1}\bigg(\sqrt{s(s+1)-m_1(m_1\mp 1)} C_{s, m_1\mp 1, s, -m_2}^{\ell, q} \\
        &\quad +\sqrt{s(s+1)-m_2(m_2\pm 1)}C_{s, m_1, s, -m_2\mp 1}^{\ell, q}\bigg) \\
        &= a_s(-1)^{m_1} \cdot \sqrt{\ell(\ell+1)-q(q\pm 1)} \cdot C_{s, m_1, s, -m_2}^{\ell, q\pm 1}
        \end{aligned}
    \end{align}
    which then gives
    \begin{equation}\label{eq:s+-adj-tslq}
        \langle s, m_1, r|[S^\pm_\tot, T_{s, \ell, q}]|s, m_2, r\rangle = \sqrt{\ell(\ell+1)-q(q\pm 1)}\cdot \langle s, m_1, r| T_{s, \ell, q\pm 1}|s, m_2, r\rangle . 
    \end{equation}

Combining this with~\cref{eq:Tslq eigenvec of adSZ} we obtain that $T_{s, \ell, q}$ is an eigenvector of $\CL_{\su(2)}$:
\begin{align}
    -2\CL_{\su(2)}[T_{s, \ell, q}] &= [S_\tot^Z, [S_\tot^Z, T_{s, \ell, q}]]+ \frac{1}{2}\bigg([S_\tot^-, [S_\tot^+, T_{s, \ell, q}]]+[S_\tot^+, [S_\tot^-, T_{s, \ell, q}]]\bigg) \\
    &= \bigg(q^2 + \ell(\ell+1) - \frac{1}{2}q(q+1)- \frac{1}{2}q(q-1)\bigg)\cdot T_{s, \ell, q} = \ell(\ell+1)\cdot T_{s, \ell, q}.
\end{align}
\end{proof}

\begin{lemma}
     The collection of operators $T_{s, \ell, q}$ forms a Hilbert-Schmidt orthonormal basis for $\comm(\sS_n)$.
\end{lemma}
\begin{proof}
    Orthogonality is guaranteed by the spectral theorem and distinct eigenvalues, as $\ad_{S_{\tot}^Z}$ is skew-Hermitian, while $-2\calL_{\su(2)}$ and left-multiplication by $\S^2_{\tot}$ are Hermitian. 
    From Schur-Weyl duality (c.f.~\cref{thm:qubit schur-weyl}), the dimension of $\comm(\sS_n)$ is $\sum_{s\in\calS} (2s+1)^2$, whence a counting argument implies this collection forms a basis. 
    Finally, the normalization condition follows orthogonality relations of the Clebsch-Gordan coefficients: 
    \begin{equation}
        \tr[T_{s, \ell, q}^\dagger  T_{s, \ell, q}] = |a_s|^2 \dim(W_{Q(s)}) \cdot \sum_{m_1, m_2} |C^{\ell, q}_{s, m_1, s, -m_2}|^2=\sum_{m_1, m_2} |C^{\ell, q}_{s, m_1, s, -m_2}|^2= 1.
    \end{equation}
\end{proof}

Recall \cref{obs:L preserves l and q eigenvalues}, which ensures that for any $\ell$ and $q$, $\calA^{(\ell,q)}$ is an invariant subspace of $\calL_{\loc}$. Unraveling this statement, we have that for all $s\in \calS$ subject to $s \geq \ell/2$
\begin{equation} \label{eq:matrix elements Alq}
    -\calL_{\loc}(T_{s,\ell,q}) = \sum_{\substack{s'\in \calS \\ s'\geq \ell/2} }d_{s, s'}^{(\ell,q)} T_{s',\ell,q},
\end{equation} for some complex scalars $d_{s, s'}^{(\ell,q)}$. But more can be said: as in the special case of the coarse-grained Pauli master equation (see~\cref{section:L-tridiagonal}), one can leverage the Wigner-Eckart theorem to show that one cannot change the spin variable $s$ by too much. We further find that the resultant matrix entries for $\calL_{\loc}$ are independent of $q$, reflecting $\SU(2)$ invariance.

\begin{lemma}[$\CL_\loc|_{\calA^{(\ell,q)}}$ is tridiagonal]\label{lem:Llq is tridiagonal} 
On the basis $\{T_{s, \ell, q}\}$ for $\calA^{(\ell,q)}$ of~\cref{def:CG-basis}, $-\CL_\loc$ admits a tridiagonal action:
    \begin{equation}
    -\CL_\loc(T_{s, \ell, q}) = \rd_{s, s}^{(\ell)} \cdot  T_{s, \ell, q} + \rd_{s, s+1 }^{(\ell)}\cdot  T_{s+1, \ell, q} + \rd_{s, s-1}^{(\ell)} \cdot T_{s-1, \ell, q}\label{eq:tri-diag-Tslq}
\end{equation}
where the coefficients depend on $s$ and $\ell$, but not on $q$. 

This similarly holds for the Pauli twirl $\calT$, whose corresponding coefficients we call $\rd t_{s, s'}^{(\ell)}$.
\end{lemma}
\begin{proof}
    To see $q$ independence, we apply raising/lowering operators $\ad_{S_{\tot}^\pm}$ to both sides of~\cref{eq:matrix elements Alq} and use that $\calL_{\loc}$ is an intertwiner and so satisfies 
    \begin{equation}
        [S_{\tot}^\pm, \calL_{\loc}(O)] = \calL_{\loc}([S_{\tot}^\pm,O]) \qquad \text{for all } O\in \calB(\calH).
    \end{equation} Thus, given the matrix elements for $\calL_{\loc}:\calA^{(\ell,q)}\to \calA^{(\ell,q)}$, one may immediately determine the matrix elements for any other $q'$ via repeated raising/lowering $[S_{\tot}^\pm ,T_{s,\ell,q}] = \sqrt{\ell(\ell+1)-q(q\pm 1)} T_{s,\ell,q\pm 1}$ from~\cref{eq:s+-adj-tslq}. In particular, the resulting equation loses all $q$ dependence, i.e. $d_{s, s'}^{(\ell,q)} =: \rd_{s, s'}^{(\ell)}$ for all $-\ell \leq q \leq \ell$.

    To see tridiagonality, we again use that the Wigner-Eckart theorem (\cref{thm:wigner eckart}) implies that for any single-site Pauli $S_{i}^{\alpha}$, 
    $\Pi_s S_i^{\alpha} \Pi_{s'} = 0$ unless $s\in \{s'-1,s',s'+1\}$, and so
    \begin{equation}
        \Tr [ T_{s',\ell,q}^\dagger \calL_{\loc} (T_{s,\ell,q})] = 0 \qquad \text{unless }s\in \{s'-1,s',s'+1\}.
    \end{equation}
\end{proof}
\begin{remark}
    The matrices corresponding to $-\calL_{\loc} \vert_{\calA^{(\ell,q)}}$ do not have the same dimension as $\ell$ varies. However, we may regard them all as $\abs{\calS}\times \abs{\calS}$ matrices by padding with zeros for invalid spins $s<\ell/2$.
\end{remark}

\subsection{Monotonicity in \texorpdfstring{$\ell$}{ell}}
\label{section:monotonicity-Tlm}
The central claim of this section is the following monotonicity statement on the spectra of $\CL_{\loc}$ when restricted to the spaces $\calA^{(\ell,q)}$.

\begin{proposition}
    [\cref{prop:monotonicity}, restated]\label{lemma:monotone-gaps} The minimum eigenvalue of $-\CL_{\loc}$ when restricted to operators in $\calA^{(\ell,q)}$ is monotonically increasing in $\ell\geq 1$ and independent of $q$. That is,
    \begin{equation}
    \min_{\substack{O \in \calA^{(\ell+1,q)} \\[.2ex] O \neq 0}}  \frac{\langle O, -\CL_{\loc}(O)\rangle_\rho}{\|O\|_\rho^2} \geq \min_{\substack{O \in \calA^{(\ell,q')} \\[.2ex] O \neq 0}} \frac{\langle O, -\CL_{\loc}(O)\rangle_\rho}{\|O\|_\rho^2}
    \end{equation} whenever $\ell\geq 1$ and for all valid $q,q'$.
\end{proposition}
The true workhorse of this proposition is the following~\cref{lemma:monotone-coefficients} , whose rather involved proof will be presented momentarily in~\cref{section:proof-of-monotone-coeffs}.
\begin{lemma}
    [Monotonicity of transition coefficients]\label{lemma:monotone-coefficients} Under the choice of basis $\{T_{s, \ell, q}\}_{s, \ell, q}$ from~\cref{def:CG-basis}, the matrix elements $\rd_{s, s'}^\ell$ (\cref{lem:Llq is tridiagonal}) are real-valued, monotonically increasing in $\ell$,\footnote{That is, monotonically increasing in $\ell$ for valid spins $s\geq \ell/2$.} and the off-diagonal elements $s'\neq s$ are non-positive.
\end{lemma}
Given this, the proof of the main proposition is a short exercise in linear algebra.
\begin{proof}
    [Proof of \cref{lemma:monotone-gaps}] Consider any $O\in \calA^{(\ell,q)}$ and let $O = \sum x_s T_{s, \ell, q}$ denote its decomposition into the $T_{s, \ell, q}$ basis. We then have that the Dirichlet form can be written as a quadratic form: 
    \begin{equation}
        \langle O, -\CL_{\loc}(O)\rangle_\rho = x^\dagger M^{(\ell)} x
    \end{equation}
    where $M^{(\ell)}_{s,s'}=\mu_{s'}\rd^{(\ell)}_{s, s'}$ and $\mu_s:=Z(\beta)^{-1}e^{-\beta E_s}=\pi(s)/\Tr[\Pi_s]$, since the Hilbert-Schmidt-normalized operator $T_{s,\ell,q}$ is supported on the spin-$s$ block and has $\|T_{s,\ell,q}\|_\rho^2=\mu_s$.  By \cref{lemma:monotone-coefficients}, the ($q$-independent) entries of $M^{(\ell)}$ are (1) real-valued, (2) monotonically increasing in $\ell$, and (3) the off-diagonal elements $s\neq s'$ are $M^{(\ell)}_{s,s'}\leq 0$.
    
    Observe that points (1, 3) above entail we can restrict our attention to vectors $x$ with real and non-negative entries, since entry-wise $x\rightarrow |x|$ does not increase the Dirichlet form:
    \begin{equation}
        \frac{x^\dagger M^{(\ell)} x}{x^\dagger \text{diag}(\mu)x} \geq \frac{|x|^\top M^{(\ell)} |x|}{|x|^\top \text{diag}(\mu)|x|}
    \end{equation}
    Now, we claim that given a vector $x$ with non-negative entries $x_s\geq 0$, the Dirichlet form corresponding to $M^{(\ell)}$ is bounded by that of $M^{(\ell-1)}$. Indeed,
    \begin{align}
       x^\dagger M^{(\ell)} x &= \bigg(\sum_{s} M^{(\ell)}_{s,s} x_s^2+\sum_{s\neq s'}M^{(\ell)}_{s,s'}x_s x_{s'}\bigg) \\
        &\geq \bigg(\sum_s M^{(\ell-1)}_{s,s} x_s^2+\sum_{s\neq s'}M^{(\ell-1)}_{s,s'}x_s x_{s'}\bigg) =  x^\dagger M^{(\ell-1)} x , 
    \end{align}
    where we have used the assumed monotonicity $M^{(\ell)}_{s,s'}\geq M^{(\ell-1)}_{s,s'}$. Taking the infimum over coefficient vectors proves the claimed monotonicity statement of the minimum eigenvalue.
\end{proof}

\subsubsection{Proof of \texorpdfstring{\cref{lemma:monotone-coefficients} (Monotonicity of Transition Coefficients)}{Monotonicity of Transition Coefficients}}
\label{section:proof-of-monotone-coeffs}

We dedicate this section to a proof of the monotonicity statement~\cref{lemma:monotone-coefficients} for the transition coefficients. At a high level, the proof strategy starts by relating the monotonicity of the matrix corresponding to $-\calL_{\loc}$ to that of the Pauli twirl $\calT$. Then, to show monotonicity there, one expands in the $|s, m, r\rangle$ basis in terms of the Clebsch-Gordan coefficients, and then leverages a certain contraction formula to obtain a Wigner 6j coefficient. Throughout we use properties of the Wigner 3j coefficients as presented in \cref{section:wigner-properties}.

\paragraph{Reduction to the matrix elements of the Pauli twirl.}

We start from the explicit form of the Davies generator (c.f. \cref{eq:lindblad_def}) and write:
\begin{align}
   - \rd_{s, s'}^{(\ell)}&=\langle T_{s', \ell, q}, \CL_\loc(T_{s, \ell, q})\rangle \\&= \sum_\omega \gamma(\omega) \sum_{i, \alpha}\underbrace{ 
        \tr\bigg[ T_{s', \ell, q}^\dagger S_{i}^\alpha(\omega)^\dagger T_{s, \ell, q}S_{i}^\alpha(\omega)\bigg]}_{\text{Transition}} - \frac{1}{2}\underbrace{\tr\bigg[   S_{i}^\alpha(\omega)^\dagger S_{i}^\alpha(\omega)\{T_{s', \ell, q}^\dagger, T_{s, \ell, q}\}\bigg]}_{\text{Dissipative}}. \label{eq:trans-dis-d}
\end{align}

Note that $\sum_{i,\alpha,\omega}\gamma(\omega) S_i^\alpha(\omega)^\dagger S_i^\alpha(\omega)$ is both $\SU(2)$ and $\Sn$ invariant and hence, by Schur's lemma, acts as a scalar on each spin-$s$ block. In particular, its contribution to $\rd_{s, s'}^{(\ell)}$ is independent of $\ell$, and the dissipative contribution vanishes when $s'\ne s$. Thus all $\ell$-dependence comes from the transition term.

We use a familiar strategy as in~\cref{eq:lind_on_Pi_expanded} to compress the transition term. Recalling the Pauli twirl (\cref{def:twirl_def}) and repeatedly using that $T_{s,\ell,q} \Pi_{s'} = \Pi_{s'}T_{s,\ell,q} = \delta_{s=s'} \cdot T_{s,\ell,q}$, we obtain
\begin{align}
    \sum_\omega \gamma(\omega) \sum_{i, \alpha}\underbrace{ 
        \tr\bigg[ T_{s', \ell, q}^\dagger S_{i}^\alpha(\omega)^\dagger T_{s, \ell, q}S_{i}^\alpha(\omega)\bigg]}_{\text{Transition}} &= \gamma(E_{s'}-E_s) \, \Tr \bigg[ T_{s',\ell,q}^\dagger \calT(T_{s,\ell,q})\bigg] \\
        &= \gamma(E_{s'}-E_s) \, \rd t_{s, s'}^{(\ell)},
\end{align}
where $E_s=-s(s+1)/n$ denotes the energy of the spin-$s$ sector and $\rd t_{s, s'}^{(\ell)}$ is a matrix element of the Pauli twirl from~\cref{lem:Llq is tridiagonal}.
In turn, since the weight function $\gamma$ is strictly positive, it suffices to show a monotonicity statement for the matrix elements $\rd t_{s, s'}^{(\ell)}$.

\paragraph{Monotonicity of the Pauli twirl.}

We will now show that the matrix elements $\rd t^{\ell}_{s, s'}$ are monotonically decreasing in $\ell$ in the regime of valid spins $s,s'\geq \ell/2$, which by the previous paragraph will imply that $\rd^{\ell}_{s, s'}$ are monotonically increasing in $\ell$.

To match the notation of the Wigner-Eckart~\cref{thm:wigner eckart}, we relabel the sum over $\alpha\in \{X, Y, Z \}$ into those over $p\in \{-1, 0, 1\}$:
\begin{align}
    \rd t^{\ell}_{s, s'} =\sum_{i}\sum_\alpha 
        \tr\bigg[ T_{s', \ell, q}^\dagger S_{i}^{\alpha}T_{s, \ell, q}S_{i}^\alpha\bigg]
        &= \sum_{i, p} (-1)^p
        \cdot \tr\bigg[ T_{s', \ell, q}^\dagger S_{i}^pT_{s, \ell, q}S_{i}^{-p}\bigg]. \label{eq:transition-coef}
\end{align}
To proceed, we begin by expanding the expression above in the $\ket{s ,\mu, r}$ basis, leveraging the explicit structure for $T_{s,\ell, q}$ given in \cref{lem:CG-basis} and the basis elements of single-site Pauli operators from the Wigner-Eckart Theorem. For fixed $i, p$:
\begin{align}
     \tr\bigg[ T_{s', \ell, q}^\dagger S_{i}^pT_{s, \ell, q}S_{i}^{-p}\bigg] = \sum_{r_1, r_2}\sum_{\mu} & \bra{s', \mu_1, r_1}T_{s', \ell, q}^\dagger\ket{s', \mu_2, r_1} \bra{s', \mu_2, r_1} S_{i}^p \ket{s, \mu_3, r_2}  \\& \times \bra{s, \mu_3, r_2}T_{s, \ell, q}\ket{s, \mu_4, r_2}\bra{s, \mu_4, r_2}S_{i}^{-p} \ket{s', \mu_1, r_1}. \label{eq:transition-single-p} 
\end{align}
We recall that Wigner-Eckart \cref{thm:wigner eckart} implies the factorization of single-site Pauli matrix elements into a term dependent on the permutation, and a term dependent on the magnetization:
\begin{equation}
    \bra{s', \mu_2, r_1} S_{i}^p \ket{s, \mu_3, r_2} = C^{s', \mu_2}_{s, \mu_3, 1, p}\cdot \langle s', r_1 || \mathbf{S}_i ||s, r_2\rangle.
\end{equation}
This factorization simplifies the expression in \cref{eq:transition-single-p} into an $(s', s, i)$-dependent component, and an $(s', s, \ell, p)$-dependent component, thereby decoupling the permutation indices from $\ell$ and $p$:
\begin{align}
     \tr\bigg[ T_{s', \ell, q}^\dagger S_{i}^pT_{s, \ell, q}S_{i}^{-p}\bigg] = &\sum_{r_1, r_2}|\langle s', r_1|| \mathbf{S}_i || s, r_2\rangle|^2 \\& \times \sum_{\mu} (-1)^{\mu_2+\mu_3} C^{\ell,q}_{s',\mu_2, s', -\mu_1 }C^{s',\mu_2}_{s,\mu_3,1, p }C^{\ell, q}_{s, \mu_3, s, -\mu_4} C^{s,\mu_4}_{s',\mu_1,1, -p },\label{eq:coeff-factorization}
\end{align}
where we have inserted the matrix elements of the $T_{s,\ell,q}$ from~\cref{def:CG-basis}.

What remains is the contraction over the various Clebsch-Gordan coefficients. To proceed, we first re-write the formula above in terms of the Wigner $3j$ coefficients. We recall their definition (\cref{def:wigner-6j}):
\begin{equation}
    C_{j_1, m_1, j_2, m_2}^{j_3, m_3}=(-1)^{j_1-j_2+m_3}\cdot \sqrt{2j_3+1}\cdot  \begin{pmatrix}
        j_1 & j_2 & j_3\\
        m_1 & m_2 & -m_3
    \end{pmatrix}.
\end{equation}
We can now express \cref{eq:coeff-factorization} in terms of the 3j symbols:
\begin{align}
\cref{eq:coeff-factorization}\propto &\sum_{\mu} (-1)^{\mu_2+\mu_3}
 C^{\ell,q}_{s',\mu_2,s',-\mu_1}
 C^{s',\mu_2}_{s,\mu_3,1,p}
 C^{\ell,q}_{s,\mu_3,s,-\mu_4}
 C^{s,\mu_4}_{s',\mu_1,1,-p}
\\[-0.3em]
=&a_{s,s'}\cdot (2\ell+1)\cdot (-1)^{s'+s}
  \cdot \sum_{\mu}(-1)^{\mu_3+\mu_4}\!\cdot\!
  \begin{aligned}[t]
    &\begin{pmatrix}
        s' & s' & \ell \\
        \mu_2 & -\mu_1 & -q
      \end{pmatrix}
    \begin{pmatrix}
        s & 1 & s' \\
        \mu_3 & p & -\mu_2
      \end{pmatrix} \\
      &\times \begin{pmatrix}
        s & s & \ell \\
        \mu_3 & -\mu_4 & -q
      \end{pmatrix}
      \begin{pmatrix}
        s' & 1 & s \\
        \mu_1 & -p & -\mu_4
      \end{pmatrix}
  \end{aligned}
\\[-0.3em]
=&a_{s,s'}\cdot (2\ell+1)\cdot (-1)^{s'+s}
  \cdot \sum_{\mu}(-1)^{\mu_3+\mu_4}\!\cdot\!
  \begin{aligned}[t]
    &\begin{pmatrix}
        s' & s' & \ell \\
        \mu_2 & -\mu_1 & -q
      \end{pmatrix}
      \begin{pmatrix}
        s' & s & 1 \\
        -\mu_2 & \mu_3 & p
      \end{pmatrix}
    \\[-0.3em]
    &\times \begin{pmatrix}
        s & s & \ell \\
        \mu_4 & -\mu_3 & q
      \end{pmatrix}
      \begin{pmatrix}
        s & s' & 1 \\
        -\mu_4 & \mu_1 & -p
      \end{pmatrix}.
  \end{aligned}
\end{align}

In the first equality, we suppressed the constant $a_{s', s}:=\sqrt{2s+1}\sqrt{2s'+1}$. In the second, we leveraged three properties of the $3j$ coefficients: the even-cycle property on the second symbol, and the negation and odd-cycle properties on the third symbol (\cref{fact:3j-properties}, 3, 2, 4).

Next, we simplify the above to apply the $6j$ re-summation formula in \cref{fact:6j}. For this purpose, we reincorporate the summation over $p$, including the $(-1)^p$ phase from \cref{eq:transition-coef}. We further observe the constraint $\mu_2=\mu_1+q$ otherwise the first $3j$ symbol is $0$ (\cref{fact:3j-properties}, 1). Finally, leverage the $q$ invariance of the transition coefficients to average over the $(2\ell+1)$ values of $q$ (\cref{lem:Llq is tridiagonal}). 

To apply \cref{fact:6j}, we identify
\begin{align}
    &j_1, j_2= \max(s', s), \quad j_4, j_5=\min(s, s'),\quad j_3 = \ell, \quad j_6 = 1, \\\
    & m_1=\mu_2  \quad  m_2 = -\mu_1, \quad m_3=-q, \quad m_4=\mu_4, \quad m_5 =-\mu_3, \quad m_6 = p,
\end{align}
resulting in the formula:
\begin{align}
    \rd t^{\ell}_{s, s'} = a_{s, s'}\cdot (-1)^{s+s'+\ell+1}\cdot  \begin{Bmatrix} s' & s' & \ell \\ s & s & 1 \end{Bmatrix}.
\end{align}
We restrict our attention to the setting $s'\in \{s-1, s, s+1\}$, using \cref{fact:special-cases}:
\begin{equation}
    \rd t^{\ell}_{s, s} = a'_{s, s}\cdot \big( 2s(s+1)-\ell(\ell+1)\big), \qquad \rd t^{\ell}_{s, s-1}= a'_{s, s-1}\cdot \big( (4s^2-\ell^2)(4s^2-(\ell+1)^2)\big)^{1/2},
\end{equation}
where the prefactors $a'_{s,s}$ and $a'_{s,s-1}$ are positive and independent of $\ell$. These expressions are both monotonically decreasing in $\ell$ in the regime of valid spins $s,s'\geq \ell/2$. The case $s'=s+1$ is analogous, completing the proof.

\subsubsection{Properties of the Wigner \texorpdfstring{$3j$ and $6j$}{3j and 6j} Symbols}\label{section:wigner-properties}

We rely on a series of properties and contraction identities of Clebsch-Gordan coefficients. We dedicate this subsection to a collection of these facts, which may be found in~\cite{wigner_matrices_1993}.

\begin{definition}
    [Wigner $3j$ symbol]\label{def:wigner-6j} For positive half-integer $j_1, j_2, j_3$ and magnetic quantum numbers $m_1,m_2,m_3$ of the corresponding parity, the Wigner $3j$ symbol is written in terms of associated Clebsch-Gordan coefficients as:
    \begin{equation}
    W_{j, m}:=
    \begin{pmatrix}
        j_1 & j_2 & j_3\\
        m_1 & m_2 & m_3
    \end{pmatrix} :=\frac{1}{\sqrt{2j_3+1}}\cdot C_{j_1, m_1, j_2, m_2}^{j_3, -m_3}\cdot (-1)^{j_1-j_2+m_3} 
\end{equation}
\end{definition}

\noindent Written in terms of the Wigner $3j$ symbols, one can now identify certain permutation-invariance properties: 

\begin{fact}
    [Permutation invariance of the Wigner $3j$ symbols]\label{fact:3j-properties}
    The Wigner $3j$ symbols $W_{j, m}$ satisfy:
    \begin{enumerate}
        \item \emph{Consistency.} $W_{j, m}=0$ unless $m_1+m_2+m_3 = 0$.
        \item \textbf{Negation.} $W_{j, m} = (-1)^{j_1+j_2+j_3}\cdot W_{j, -m}$.
        \item \textbf{Even-Cycles.} For any even permutation of the columns, e.g.:
        \begin{equation}
            \begin{pmatrix}
        j_1 & j_2 & j_3\\
        m_1 & m_2 & m_3
    \end{pmatrix} = \begin{pmatrix}
        j_3 & j_1 & j_2\\
        m_3 & m_1 & m_2
    \end{pmatrix} 
        \end{equation}
        \item \textbf{Odd-Cycles.} For any odd permutation of the columns, e.g.:
         \begin{equation}
            \begin{pmatrix}
        j_1 & j_2 & j_3\\
        m_1 & m_2 & m_3
    \end{pmatrix} =  (-1)^{j_1+j_2+j_3}\cdot \begin{pmatrix}
        j_2 & j_1 & j_3\\
        m_2 & m_1 & m_3
    \end{pmatrix} 
    \end{equation}
    \end{enumerate}
\end{fact}

We will rely on two summation formulas for the CG coefficients / Wigner 3j symbols. The first is the following recursion relation. 
\begin{fact}
    [Recursion relations]\label{fact:recursions}
    For any positive half-integer $j_1, j_2, j_3$, and admissible magnetic quantum numbers $m_1, m_2, m_3$:
    \begin{align}
      -\sqrt{j_3(j_3+1)-m_3(m_3\pm 1)}  \cdot W_{j, (m
      _1, m_2, m_3\pm 1)} &=\sqrt{j_1(j_1+1)-m_1(m_1\pm 1)}  W_{j, (m_1\pm 1, m_2, m_3)} \\&+ \sqrt{j_2(j_2+1)-m_2(m_2\pm 1)}  W_{j, (m_1, m_2\pm 1, m_3)}.
    \end{align}
    Written equivalently in terms of the Clebsch-Gordan coefficients:
    \begin{align}
         \sqrt{j_3(j_3+1)-m_3(m_3\pm 1)}  \cdot C^{j_3, -(m_3\pm 1)}_{j_1, m_1, j_2, m_2}=&\sqrt{j_1(j_1+1)-m_1(m_1\pm 1)}C^{j_3, -m_3}_{j_1, m_1\pm 1, j_2, m_2}\\&+ \sqrt{j_2(j_2+1)-m_2(m_2\pm 1)} C^{j_3, -m_3}_{j_1, m_1, j_2, m_2\pm 1}.
    \end{align}
\end{fact}

The crux of our analysis lies in the following contraction formula, which consists of the sum of products of four $3j$ symbols with permuted $j$ indices. The sum ranges over all indices $m$ afforded by the selection rules. 

\begin{fact}
[Wigner $6j$ symbols]\label{fact:6j} For positive half-integers $j_1\cdots j_6$, we define the Wigner $6j$ symbol via the following contraction formula:
\begin{align}
   \begin{Bmatrix} j_1 & j_2 & j_3 \\ j_4 & j_5 & j_6 \end{Bmatrix} := \sum_{m} (-1)^{\sum_i (j_i + m_i)}
&\begin{pmatrix} j_1 & j_2 & j_3 \\ m_1 & m_2 & m_3 \end{pmatrix} \\
\times& \begin{pmatrix} j_4 & j_5 & j_3 \\ m_4 & m_5 & -m_3 \end{pmatrix}
\begin{pmatrix} j_1 & j_5 & j_6 \\ -m_1 & -m_5 & m_6 \end{pmatrix}
\begin{pmatrix} j_4 & j_2 & j_6 \\ -m_4 & -m_2 & -m_6 \end{pmatrix} 
\end{align}

\end{fact}

Finally, we highlight 2 special cases of the $6j$ contraction formula which we rely on in our proof. 

\begin{fact}
    [Special cases of the Wigner $6j$ symbol] \label{fact:special-cases}When $j_1=j_2=j_4=j_5:=j$ and $j_6=1$, 
    \begin{equation}
        \begin{Bmatrix}
j & j & \ell \\
j & j & 1
\end{Bmatrix}
=
(-1)^{\ell+1}
\frac{
2j(j+1)-\ell(\ell+1)
}{
\bigl[(2j^2)(2j+1)^2(2j+2)(j+1)\bigr]^{1/2}}
    \end{equation}
    \noindent When $j_1=j_2=j_4+1=j_5+1$ and $j_6=1$:
    \begin{equation}
\begin{Bmatrix}
j & j & \ell \\
j - 1 & j - 1 & 1
\end{Bmatrix}=
(-1)^{\ell}
\left[
\frac{
(2j+\ell)(2j+\ell+1)(2j-\ell)(2j-\ell-1)}{
(2j-1)^2(2j)^2(2j+1)^2
}
\right]^{1/2}
    \end{equation}
\end{fact}

\subsection{The Base of the Tower, \texorpdfstring{$\calA^{(1)}$}{A(1)}}
\label{section:chain-T10}

The goal of this section is to prove that the minimum eigenvalue of $-\calL_{\loc}$ restricted to $\calA^{(1,0)}$ is constant at high temperatures. An orthogonal basis for this space is given by the collection $\{ S_{\tot}^Z \cdot \Pi_s: s\in \calS, s\geq 1/2 \}$,\footnote{Indeed, this is an unnormalized version of the earlier basis, i.e. $S_{\tot}^Z \cdot \Pi_s \propto T_{s,1,0}$.} and so a generic observable $O\in \calA^{(1,0)}$ may be expressed as
\begin{align}
    O = S_\tot^Z\cdot \sum_{s\in \calS } o_{s}\cdot  \Pi_s ,\qquad o_s\in \C.
\end{align}
Throughout this section, we will prove the following lower bound for the Rayleigh quotient for arbitrary $O\in \calA^{(1,0)}$.

\begin{proposition}[\cref{prop:CL-A1}, restated]\label{prop:calA10-gap} Fix $\beta < 2.$ The generator $-\CL_\loc$ restricted to $\calA^{(1,0)}$ admits a lower bound on its minimum eigenvalue: 
    \begin{equation}
        \min_{\substack{O \in \calA^{(1,0)} \\[.2ex] O \neq 0}} \frac{ \langle O, -\CL_\loc(O)\rangle_{\rho}}{\|O\|_\rho^2}  \geq c\cdot (2-\beta), 
    \end{equation} 
    where $c\in \mathbb{R}^+$ is a universal constant.
\end{proposition}

The proof of this proposition will appear in~\cref{sec:proof of A10 rayleigh quotient}.

\subsubsection{The Variance and Dirichlet Form of \texorpdfstring{$\calA^{(1,0)}$}{A(1,0)}}

The proof of \cref{prop:calA10-gap} will rely on the following two lemmas, which compute the variance, and a lower bound on the Dirichlet form of $O\in \calA^{(1,0)}$:

\begin{lemma}
    [The variance of $\calA^{(1,0)}$]\label{lemma:var-t10} For any $O\in \calA^{(1,0)}$ defined by the coefficients $\{o_s\}_{s\in \calS}$, 
    \begin{equation}
        \mathsf{Var}_{\rho}[O]=\|O\|_{\rho}^2 = \frac{1}{3}\cdot  \mathbb{E}_{s\sim \pi} \big[s(s+1)\cdot |o_s|^2\big]
    \end{equation}
    where we recall $\pi(s):=\Tr[\Pi_s\rho]$ is the distribution over spin eigenspaces in the Gibbs state.
\end{lemma}
\begin{proof}
    We first note that the expectation of any such observable in the Gibbs state vanishes by $\SU(2)$ symmetry, i.e. $\Tr [\rho O] = 0$. Proceeding, we have
    \begin{equation}
      \tr[\rho O^\dagger O] = \sum_s \pi(s) \cdot |o_s|^2 \cdot  \frac{\tr[\Pi_s  (S_\tot^Z)^2]}{\tr[\Pi_s]} = \frac{1}{3}  \sum_s \pi(s)\cdot |o_s|^2  \cdot   \frac{\tr[\Pi_s \S^2_\tot]}{\tr[\Pi_s]} = \frac{1}{3} \mathbb{E}_{s\sim \pi}\big[s(s+1)\cdot |o_s|^2\big],
    \end{equation} where we have again used $\SU(2)$ symmetry to relate the expectations of $(S_{\tot}^Z)^2$ and $\S^2_{\tot}$, yielding the desired claim.
\end{proof}

Next, we turn to the Dirichlet form. 

\begin{lemma}
    [The Dirichlet form of $\calA^{(1,0)}$]\label{lem:diri-tau10} For any $O\in \calA^{(1,0)}$ defined by the coefficients $\{o_s\}_{s\in \calS}$, 
    \begin{equation}
        \langle O, -\CL_\loc(O)\rangle_{\rho} \geq n\cdot c_2\cdot e^{-c_1\beta}\cdot \mathbb{E}_{s\sim \pi}\big[|o_s|^2 + |o_s-o_{s-1}|^2\cdot (s^2-1)\big] 
    \end{equation}
    where $c_1, c_2\in \mathbb{R}$ are universal constants, and WLOG we assume $o_0=0$. 
\end{lemma}

\begin{proof}
We start from the divergence form of the Dirichlet form (\cref{lem:dirichlet_divergence}).
Since $O$ commutes with $H$ and $\Sn$, \cref{cor:band-diagonal-local-operators} -- which limits the Bohr frequencies of single-site Pauli operators to a constant band -- allows one to simplify the Dirichlet form into:
\begin{align}
     \langle O, -\CL_\loc(O)\rangle_{\rho} &= \sum_{\omega} h_\omega \sum_{i, \alpha} \|[S_{i}^\alpha(\omega), O]\|_{\rho}^2 \\&
     \geq n\cdot e^{-c_1\beta}\cdot \sum_s \pi(s)  \sum_{Q\in \{S_1^Z, S^+_1, S^-_1\}}\tr[\Pi_s [Q, O]^\dagger [Q, O]]\cdot  \frac{1}{\tr[\Pi_s]}
\end{align}
which corresponds to the gradient of $O$ with the jumps on a single site (WLOG site $i=1$, by $\Sn$ symmetry), on the subspace $\Pi_s$. In the above $c_1$ is a universal constant. 
To proceed, we expand in the $|s, m, r\rangle$ basis:
\begin{align}
    \tr[\Pi_s [Q, O]^\dagger [Q, O]] &= \sum_{t} \sum_{\substack{m_1, m_2\\r_1, r_2}} |\langle t, m_2, r_2| [Q, O]| s, m_1, r_1\rangle|^2 \\&= \sum_t \sum_{\substack{m_1, m_2\\r_1, r_2}} |m_1\cdot o_s - m_2\cdot o_t|^2\cdot |\langle t, m_2, r_2|Q|s, m_1, r_1\rangle|^2, \label{eq:basis-expand-dirichlet-tau10}
\end{align}
thereby reducing the computation to the evaluation of certain single-site matrix elements. Fortunately, to derive a lower bound on the above we will only require the special case of Wigner-Eckart in \cref{fact:matrix-elements-raising}. We note whenever $s\geq 1$:
\begin{align}
    \sum_{r}  |\langle s, m+1, r|S^+_1|s, m, r\rangle|^2 &\geq \frac{s^2-m^2}{4s^2} \cdot m_{s-\frac{1}{2}}^{(n-1)}, \\  \sum_{r}  |\langle s-1, m, r|S^Z_1|s, m, r\rangle|^2 &\geq \frac{s^2-m^2}{s^2} \cdot m_{s-\frac{1}{2}}^{(n-1)}, \quad \text{ and }\quad   \frac{m_{s-\frac{1}{2}}^{(n-1)}}{m^{(n)}_s} = \frac{s(n+2s+2)}{n(2s+1) }\geq \frac{1}{3},
\end{align}
where $m_{s-\frac{1}{2}}^{(n-1)}$ is the multiplicity of the $s-\frac{1}{2}$ irrep of $\SU(2)$. Put together, using $\tr[\Pi_s] = (2s+1)\cdot m_s^{(n)}$ yields
\begin{align}
     \cref{eq:basis-expand-dirichlet-tau10} &\geq  |o_s|^2\cdot\sum_{m=-s}^s \frac{s^2-m^2}{12s^2(2s+1)} + |o_s-o_{s-1}|^2\cdot \sum_{m=-s}^s \frac{m^2(s^2-m^2)}{3s^2(2s+1)} \\
     &\geq  |o_s|^2 \cdot \frac{1}{36}  + |o_s-o_{s-1}|^2\cdot  \frac{s^2-1}{45},
\end{align}
which implies the desired claim. 
\end{proof}

\subsubsection{Proof of \texorpdfstring{\cref{prop:calA10-gap} (The Spectra of $\CL_\loc|_{\calA^{(1,0)}}$)}{the Spectra of L on A(1,0)}} \label{sec:proof of A10 rayleigh quotient}

We are now in a position to prove the spectral gap for $\calA^{(1,0)}.$ We will rely on certain tail bounds for the second and third moments of $s$ under $\pi$, which we defer to the bottom of this section for clarity. By combining \cref{lem:diri-tau10} and \cref{lemma:var-t10}, the desired claim reduces to showing that for any choice of non-zero coefficients $\{o_s\}$, the following ratio satisfies the unconditional lower bound.
    \begin{equation}
        \frac{\mathbb{E}_{s\sim \pi}\big[|o_s|^2 + |o_s-o_{s-1}|^2\cdot s^2\big]}{\mathbb{E}_{s\sim \pi}\big[|o_s|^2\cdot  s^2\big]} =\Omega\big((2-\beta)/n\big).
    \end{equation}
    We note that WLOG we may assume $o_s\in \mathbb{R}$ and non-negative, and furthermore that missing boundary coefficients, such as $o_0$ when needed, are set to $0$. The following proof is loosely inspired by Hardy's inequality, in which we would like to sequentially apply the Cauchy-Schwarz inequality to (roughly speaking) relate the variance of the coefficients $o_s$ to its discrete derivatives \cite{miclo_example_1999}. 

\begin{proof}[Proof of \cref{prop:calA10-gap}]
     To proceed, we introduce a cutoff $z=16\sqrt{n/(2-\beta)}$ (twice the threshold from \autoref{lem:tail-2nd}), and divide the sum in the denominator accordingly.\footnote{If $z$ is larger than the maximal spin, the small-$s$ estimate below already covers the full denominator, so assume in the large-$s$ part that the interval $[z/2,z]\cap\calS$ is nonempty.}\\

    \noindent \textbf{Small $s$ regime.} One can bound the denominator in the range $s\in [0, z]$ by the numerator via the trivial bound:
    \begin{align}
        \sum_{s=0}^z \pi_s s^2 |o_s|^2 \leq z^2\cdot \sum_s \pi_s |o_s|^2.\label{eq:small-s}
    \end{align}

    \noindent \textbf{Large $s$ regime.} To analyze the large $s$ regime, we will need to introduce  yet another cutoff $r$. We pick $r$ such that 
    \begin{align}
      r:=  \text{argmin}_{\frac{1}{2}z\leq s\leq z} \bigg[|o_s|^2 \cdot s^2\cdot \pi_s\bigg]\quad  \Rightarrow \quad |o_r|^2 \cdot r^2\cdot \pi_r\leq \frac{2}{z} \cdot \sum_{s = z/2}^z |o_s|^2 \cdot s^2\cdot \pi_s \leq 2z\cdot \sum \pi_s |o_s|^2,
    \end{align}
    i.e. it minimizes the contribution to the denominator in a range around $z$. The second inequality uses \cref{eq:small-s}. For any $s\geq r$, we can express $o_s$ via the following expansion:
    \begin{align}
        o_s=o_r+\sum_{j=r}^{s-1}(o_{j+1}-o_j)\quad \Rightarrow\quad  o_s^2 \leq 2o_r^2 + 2(s-r)\sum_{j=r}^{s-1}(o_{j+1}-o_j)^2.
    \end{align}
    Placed in the denominator, we bound: 
\begin{align}
\sum_{s=r}^{\infty} \pi_s\cdot  s^2\cdot  |o_s|^2 &\leq
2|o_r|^2\cdot  \sum_{s\ge r} s^2\cdot  \pi_s
+
2\sum_{s\ge r} s^2(s-r)\cdot \pi_s  \sum_{j=r}^{s-1}(o_{j+1}-o_j)^2. \\
&=
2|o_r|^2 \cdot \sum_{s\ge r} s^2\cdot \pi_s
+
2\sum_{j\geq r}(o_{j+1}-o_j)^2 \sum_{s\ge j}(s-r)\pi_s\cdot s^2.
\end{align}
In the last equality above, we simply re-ordered the summation. At this point we require certain tail bounds on the second and third moments of $s$ under the distribution $\pi$ (at high temperatures). The first of two terms above is bounded by the first statement in  \cref{lem:tail-2nd}. For the second term, re-express the sum and then use both statements in \cref{lem:tail-2nd}: 
\begin{align}
    \sum_{s\ge j}(s-r)\cdot \pi_s\cdot s^2 &\leq  \pi_j\cdot j^2\cdot \bigg(\frac{n^2}{(2-\beta)^2\cdot j^2}+(j-r)\cdot \frac{n}{(2-\beta)j}\bigg)\cdot 32 = \frac{64n}{2-\beta}\cdot \pi_j\cdot j^2.
\end{align}

Put together, we can now conclude the bound on the denominator in the large $s$ regime:
\begin{align}
    \sum_{s=r}^{\infty} \pi_s\cdot  s^2\cdot  |o_s|^2  &\leq  |o_r|^2 \cdot r^2\cdot \pi_r \cdot 32\sqrt{\frac{n}{2-\beta}}+ \frac{128n}{2-\beta}\cdot \sum_{s\geq r}|o_{s+1} - o_s|^2\cdot \pi_s\cdot s^2 \\
    &\leq \Theta\bigg(\frac{n}{2-\beta}\bigg)\cdot \bigg(\sum_s \pi_s\cdot    |o_s|^2 + \sum_{s\geq r}\pi_s\cdot |o_{s+1}-o_s|^2\cdot s^2\bigg)
\end{align}
Put together with the expression for small $s$ in \cref{eq:small-s}, concludes the lower bound on the Dirichlet form in terms of the variance.

\end{proof}

\begin{lemma}[Tail bounds on second moments]\label{lem:tail-2nd}
Fix $\beta<2$, and let $A_s:= s^2 \cdot \pi_s$. Then, there exists a threshold $m:=8\cdot \sqrt{\frac{n}{2-\beta}}$ such that for any $s\geq m$:
\begin{equation}
    \sum_{\ell=s}^\infty A_\ell \le \frac{4n}{s\cdot (2-\beta)}\cdot  A_s,
\qquad
\sum_{\ell=s}^\infty (\ell-s)A_\ell \le  \frac{16n^2}{s^2(2-\beta)^2}\cdot  A_s.
\end{equation}
\end{lemma}

\begin{proof}
For $1\le s<n/2$, the exact ratio is (see \cref{eq:stationary,eq:formula-for-dim})
\begin{align}
\frac{A_{s+1}}{A_s}
=
e^{2\beta (s+1)/n}
\frac{(s+1)^2}{s^2}
\frac{(2s+3)^2}{(2s+1)^2}
\frac{n-2s}{n+2s+4}. 
\end{align}
We claim (and prove shortly) that for $s\geq m$, there exists a suitable $c_\beta\in \BR^+$ such that:
\begin{equation}
    \frac{A_{s+1}}{A_s} \leq \exp\bigg[-\frac{c_\beta\cdot  s}{n}\bigg].\label{eq:geom-decay-as}
\end{equation}
We observe this implies the following ratio-of-coefficients for $\ell \geq s \geq m$:
\begin{equation}
    \frac{A_\ell}{A_s}=\prod_{j=s}^{\ell-1} \frac{A_{j+1}}{A_j} \leq \exp\bigg[-c_\beta \sum_{j=s}^{\ell-1} \frac{j}{n}\bigg]  \leq \exp\bigg[-c_\beta  \frac{s(\ell-s)}{n}\bigg].
\end{equation}
The claimed moment bounds are then immediate, with $C = \min(1, c_\beta \cdot s/n):$
\begin{align}
    &\sum_{\ell=s}^nA_\ell\leq  A_s\sum_{\ell=s}^n \exp\bigg[-c_\beta \frac{s(\ell-s)}{n}\bigg]\leq A_s\sum_{t=0}^\infty e^{-C\cdot t}\leq\frac{2A_s}{C},\\
    &\sum_{\ell\geq s} A_\ell\cdot  (\ell-s)\leq A_s \sum_{t=0}^\infty t\cdot e^{-C\cdot t} \leq \frac{4A_s}{C^2}.
\end{align}
It remains now only to prove \cref{eq:geom-decay-as}. For this purpose, we observe first that $x\geq 0:1+x\leq e^x$ implies
\begin{equation}
    \frac{(s+1)^2}{s^2}\frac{(2s+3)^2}{(2s+1)^2}
\leq e^{4/s}, \quad \text{and} \quad  \log \bigg(\frac{1-2x}{1+2x}\bigg) \leq -4 x \text{ for all }x\in \left[0, \frac{1}{2}\right].
\end{equation}
Put together, we conclude
\begin{equation}
    \log  \frac{A_{s+1}}{A_s} \leq  2(\beta-2)\cdot \frac{s}{n} + \frac{4}{s}+\frac{2\beta}{n} \leq -(1-\frac{\beta}{2})\cdot  \frac{s}{n} \quad \text{ if } s\geq 8\cdot \sqrt{\frac{n}{2-\beta}},
\end{equation}
which justifies the claim with $c_\beta = 1-\beta/2.$ We further remark the constant $C=\min(1, c_\beta \cdot s/n)=c_\beta s/n$ due to the limits on $s$.

\end{proof}

\end{document}